\documentclass{article}[11]
\usepackage{fullpage}
\usepackage{appendix}
\usepackage{amsmath}
\usepackage{amssymb}
\usepackage{amsthm}
\usepackage{ifpdf}
\usepackage{algorithmic}
\usepackage{subfig}
\usepackage{wrapfig}
\usepackage{cite}
\usepackage{colortbl}
\usepackage{tikz}
\usepackage{morefloats}

\usepackage{commands-tam}

\ifpdf

  \usepackage[pdftex]{epsfig}
  \usepackage[pdftex]{hyperref}

\else

    \usepackage[dvips]{epsfig}
    \newcommand{\href}[2]{#2}

\fi

\newif\ifabstract
\newif\iffull

\ifabstract
	\fullfalse
\else
	\fulltrue
\fi

\newtoks\magicAppendix
\magicAppendix={}
\newtoks\magictoks
\newif\iflater
\laterfalse
\ifabstract
\long\def\later#1{\magictoks={#1}%
  \edef\magictodo{\noexpand\magicAppendix={\the\magicAppendix \par
    \the\magictoks%
  }}
  \magictodo}
\long\def\both#1{\magictoks={#1}%
  \edef\magictodo{\noexpand\magicAppendix={\the\magicAppendix \par
    \noexpand\setcounter{theorem-preserve}{\noexpand\arabic{theorem}}%
    \noexpand\setcounter{theorem}{\arabic{theorem}}%
    \noexpand\setcounter{section-preserve}{\noexpand\arabic{section}}%
    \noexpand\setcounter{section}{\arabic{section}}%
	\noexpand\let\noexpand\oldsection=\noexpand\thesection
	\noexpand\def\noexpand\thesection{\thesection}
	\noexpand\let\noexpand\oldlabel=\noexpand\label
	\noexpand\let\noexpand\label=\noexpand\blank
    \the\magictoks%
    \noexpand\setcounter{theorem}{\noexpand\arabic{theorem-preserve}}%
    \noexpand\setcounter{section}{\noexpand\arabic{section-preserve}}%
	\noexpand\let\noexpand\thesection=\noexpand\oldsection
	\noexpand\let\noexpand\label=\noexpand\oldlabel
  }}
  \magictodo
  \the\magictoks}
\else
\long\def\later#1{#1}
\long\def\both#1{#1}
\fi
\long\def\magicappendix{
	\latertrue%
	\the\magicAppendix%
}

\theoremstyle{definition}
\newtheorem{theorem}{Theorem}[section]
\newtheorem{lemma}[theorem]{Lemma}
\newtheorem{definition}[theorem]{Definition}

\title{Universal Computation with Arbitrary Polyomino Tiles in Non-Cooperative Self-Assembly}

\author{
    S\'andor P. Fekete
        \thanks{Department of Computer Science, TU Braunschweig, 38106 Braunschweig, Germany.
        \protect\url{ s.fekete@tu-bs.de}.}
\and
    Jacob Hendricks
        \thanks{Department of Computer Science and Computer Engineering, University of Arkansas, Fayetteville, AR, USA.
        \protect\url{jhendric@uark.edu}.
        This author's research was supported in part by National Science Foundation Grants CCF-1117672 and CCF-1422152.}
\and
	Matthew J. Patitz
        \thanks{Department of Computer Science and Computer Engineering, University of Arkansas, Fayetteville, AR, USA.
        \protect\url{mpatitz@self-assembly.net}.
        This author's research was supported in part by National Science Foundation Grants CCF-1117672 and CCF-1422152.}
\and
    Trent A. Rogers
        \thanks{Department of Computer Science and Computer Engineering, University of Arkansas, Fayetteville, AR, USA.
        \protect\url{tar003@uark.edu}.
        This author's research was supported by the National Science Foundation Graduate Research Fellowship Program under Grant No. DGE-1450079, and National Science Foundation grants CCF-1117672 and CCF-1422152.}
\and
    Robert T. Schweller
        \thanks{University of Texas--Pan American, Edinburg, TX 78539, USA.
        \protect\url{rtschweller@utpa.edu}.
        This author's research was supported in part by National Science Foundation Grant CCF-1117672.}
}

\date{}

\begin{document}

\maketitle

\begin{abstract}
In this paper we explore the power of geometry to overcome the limitations of
non-cooperative self-assembly.  We define a generalization of the abstract
Tile Assembly Model (aTAM), such that a tile system consists of a collection of
polyomino tiles, the Polyomino Tile Assembly Model (polyTAM), and investigate the
computational powers of polyTAM systems at {\em temperature 1}, where attachment
among tiles occurs without glue cooperation (i.e., without the enforcement that
more than one tile already existing in an assembly must contribute to the
binding of a new tile).  Systems composed of the unit-square tiles of the aTAM
at temperature 1 are believed to be incapable of Turing universal computation
(while cooperative systems, with temperature $> 1$, are able).  As
our main result, we prove that for any polyomino $P$ of size 3 or greater, there
exists a temperature-1 polyTAM system containing only shape-$P$ tiles that is
computationally universal.  Our proof leverages the geometric properties of
these larger (relative to the aTAM) tiles and their abilities to effectively
utilize geometric blocking of particular growth paths of assemblies, while
allowing others to complete.  In order to prove the computational powers of
polyTAM systems, we also prove a number of geometric properties held by all
polyominoes of size $\ge 3$.

To round out our main result, we provide strong evidence that size-1 (i.e. aTAM
tiles) and size-2 polyomino systems are unlikely to be computationally
universal by showing that such systems are incapable of \emph{geometric
bit-reading}, which is a technique common to all currently known temperature-1
computationally universal systems.  We further show that larger polyominoes
with a limited number of binding positions are unlikely to be computationally
universal, as they are only as powerful as temperature-1 aTAM systems.
Finally, we connect our work with other work on domino self-assembly to show
that temperature-1 assembly with at least 2 distinct shapes, regardless of the
shapes or their sizes, allows for universal computation.  \end{abstract}

\pagebreak

\section{Introduction}

Theoretical studies of algorithmic self-assembly have produced a wide variety
of results that establish the computational power of tile-based self-assembling
systems.  From the introduction of the first and perhaps simplest model, the
abstract Tile Assembly Model (aTAM) \cite{Winf98}, it was shown that
self-assembling systems, which are based on relatively simple components
autonomously coalescing to form more complex structures, are capable of Turing-universal computation.
This computational power exists within the aTAM, and
has been harnessed to algorithmically guide extremely complex theoretical
constructions (e.g. \cite{SolWin07,RotWin00,jCCSA,jSADS,IUSA,jSADSSF}) and has
even been exploited within laboratories to build nanoscale self-assembling
systems from DNA-based tiles which self-assemble following algorithmic behavior
(e.g.
\cite{RothTriangles,SchWin07,MaoLabReiSee00,BarSchRotWin09,LaWiRe99,evans2014crystals}).

While physical implementations of these systems are constantly increasing in
scale, complexity, and robustness, they are orders of magnitude shy of
achieving results similar to those of many naturally occurring self-assembling
systems, especially those found in biology (e.g. the formation of many cellular
structures or viruses).  This disparity motivates theoretical studies that can
focus efforts on first discovering the ``tricks'' used so successfully by
nature, and then on incorporating them into our own models and systems.  One of
the fundamental properties so successfully leveraged by many natural systems,
but absent from models such as the aTAM, is geometric complexity of components.
For instance, self-assembly in biological systems relies heavily upon the
complex 3-dimensional structures of proteins, while tile-assembly systems are
typically restricted to basic square (or cubic) tiles.  In this paper, we
greatly extend previous work that has begun to incorporate geometric aspects
of self-assembling components \cite{GeoTiles,SFTSAFT,Duples,OneTile} with the
development of a model allowing for more geometrically complex tiles, called
\emph{polyominoes}, and an examination of the surprising computational powers
that systems composed of polyominoes possess.

The process of tile assembly begins from a \emph{seed} structure, typically a
single tile, and proceeds with tiles attaching one at a time to the growing
assembly.  Tiles have \emph{glues}, taken from a set of glue types, around
their perimeters which allow them to attach to each other if their glues match.
Algorithmic self-assembling systems developed by researchers, both theoretical
and experimental, tend to fundamentally employ an important aspect of tile
assembly known as \emph{cooperation}.  In theoretical models, cooperation is
available when a particular parameter, known as the \emph{temperature}, is set
to a value $> 1$ which can then enforce that the binding of a tile to a growing
assembly can only occur if that tile matches more than one glue on the
perimeter of the assembly.  Using cooperation, it is simple to show that
systems in the aTAM are capable of universal computation (by simulating
particular cellular automata \cite{Winf98}, or arbitrary Turing machines
\cite{jSADS,jCCSA,SolWin07}, for instance).  However, it has long been
conjectured that in the aTAM without cooperation, i.e. in systems where the
temperature $=1$, universal computation is impossible
\cite{jLSAT1,IUNeedsCoop,ManuchTemp1}.  Interestingly, though, a collection of
``workarounds'' have been devised in the form of new models with a variety of
properties and parameters which make computation possible at temperature 1
(e.g. \cite{SingleNegative,CookFuSch11,GeoTiles,Signals,Duples}).

In this paper, we introduce the Polyomino Tile Assembly Model (polyTAM), in which
each tile is composed of a collection of unit squares joined along their edges.
This allows for tiles with arbitrary geometric complexity and a much larger
variety of shapes than in earlier work involving systems composed of both
square and $2\times1$ rectangular tiles \cite{Duples,BreakableDuples}, or those
with tiles composed of square bodies and edges with bumps and dents
\cite{GeoTiles}.  Our results prove that geometry, in the polyTAM as in natural
self-assembling systems, affords great power.  Namely, \emph{any} polyomino
shape which is composed of only 3 or more unit squares has enough geometric
complexity to allow a polyTAM system at temperature 1, composed only of tiles of
that shape, to perform Turing universal computation.  This impressive potency
is perhaps even more surprising when it is realized that while a single unit-square
polyomino (a.k.a.\ a \emph{monomino}, or a standard aTAM tile) is
conjectured not to provide this power, the same shape expanded in scale to a
$2\times2$ square polyomino does.  The key to this power is the ability of
arbitrary polyominoes of size 3 or greater to both combine with each other to
form regular grids, as well as to combine in a variety of relative offsets that
allow some tiles to be shifted relatively to those grids and then perform geometric
blocking of the growth of specific configurations of paths of tiles, while
allowing other paths to complete their growth.  With just this seemingly simple
property, it is possible to design temperature-1 systems of polyominoes that
can simulate arbitrary Turing machines.

In addition to this main positive result about the
computational abilities of all polyominoes of size $\ge 3$, we also provide
negative results that further help to refine understanding of exactly
what geometric properties are needed for Turing universal computation in
temperature-1 self-assembly.  We prove that a fundamental gadget (which we call the \emph{bit-reading gadget})
used within all known systems that can compute at temperature 1 in any tile-assembly model,
which we call a \emph{bit-reading gadget}, is impossible to
construct with either the square tiles of the aTAM or with dominoes (a.k.a.\ duples).
This provides further evidence that systems composed solely of those shapes are
incapable of universal computation.  Furthermore, we prove that regardless of
the size and shape of a polyomino, systems composed of polyominoes
with only (1) $\le 3$ positions
on its perimeter at which to place glues, or (2) 4 positions for glues that
are restricted to binding with each other as complementary pairs of sides,
are no more powerful
than aTAM temperature-1 systems, again providing evidence that they are
incapable of performing Turing universal computation.

This paper is organized as follows.  In Section~\ref{sec:polyTAM} we define the
Polyomino Tile Assembly Model and related terminology.  In
Section~\ref{sec:comp-overview} we formally define a bit-reading gadget and
then present an overview of how they can be used in a temperature-1 tile
assembly system to simulate arbitrary Turing machines.  Then, in
Section~\ref{sec:tech-lemmas} we prove some fundamental lemmas about the
geometric properties of polyominoes and the ways in which they can combine in the
plane to form grids. Section~\ref{sec:main-result} contains the proof of
our main result, while Section~\ref{sec:limited-systems} contains our results
that hint at the computational weakness of some systems.  Finally,
Section~\ref{sec:multi-poly} describes how the positive results of this paper
along with that of \cite{Duples} proves that any polyomino system composed of
any two polyomino shapes is capable of universal computation.

\section{Polyomino Tile Assembly Model}\label{sec:polyTAM}

In this section we define the Polyomino Tile Assembly Model (polyTAM) and relevant terminology.

\paragraph{Polyomino Tiles}
A \emph{polyomino} is a plane geometric figure formed by joining one or more equal unit squares edge to edge;
it can also be considered a finite subset of the regular square tiling with a connected interior.
For convenience, we will assume that each unit square is centered
on a point in $Z^2$. We define the set of \emph{edges} of a polyomino to be
the set of faces from the constituent unit squares of the polyomino that lie on
the boundary of the polyomino shape.  A \emph{polyomino tile} is a polyomino with
a subset of its edges labeled from some {\em glue} alphabet $\Sigma$, with
each glue having an integer \emph{strength} value.  Two tiles are said to
\emph{bind} when they are placed so that they have non-overlapping interiors
and adjacent edges with matching glues; each matching glue binds with
force equal to its strength value. An \emph{assembly} is any connected set of
polyominoes whose interiors do not overlap.  Given a positive integer $\tau \in
\mathbb{N}$, an assembly is said to be \emph{$\tau$-stable} or (just
\emph{stable} if $\tau$ is clear from context), if any partition of the assembly
into two non-empty groups (without cutting individual polyominoes) must
separate bound glues whose strengths sum to $\ge \tau$.

The \emph{bounding rectangle} $B$ around a polyomino $P$ is the rectangle with minimal area (and corners lying in $\mathbb{Z}^2$) that contains $P$.  For each polyomino shape, we designate one pixel (i.e. one of the squares making up $P$) $p$ as a distinguished pixel that we use as a reference point.  More formally, a \emph{pixel} $p$ in a polyomino $P$ (or polyomino tile) is defined in the following manner.  Place $P$ in the plane so that the southwest corner of the bounding rectangle of $P$ lies at the origin.  Then a pixel $p=(p_1,p_2) \in P$ is a point in $\mathbb{Z}^2$ which is occupied by a unit square composing the polyomino $P$.  We say that a pixel $p' \in P$ lies on the perimeter of the bounding rectangle $B$ if an edge of the pixel $p'$ lies on an edge of $B$.

\paragraph{Tile System}
A \emph{tile assembly system} (TAS) is an ordered triple $\calT =
(T,\sigma,\tau)$ (where $T$ is a set of polyomino tiles, and $\sigma$ is a
$\tau$-stable assembly called the \emph{seed}) that consists of integer
translations of elements of $T$. $\tau$ is the \emph{temperature} of the
system, specifying the minimum binding strength necessary for a tile to attach
to an assembly.  Throughout this paper, the temperature of all systems is
assumed to be 1, and we therefore frequently omit the temperature from the
definition of a system (i.e. $\calT = (T,\sigma)$).

If the tiles in $T$ all have the same polyomino shape, $\calT$ is said to be a
\emph{single-shape} system; more generally $\calT$ is said to be a
$c$-shape system if there are $c$ distinct shapes in $T$.  If not stated
otherwise, systems described in this paper should by default be assumed to be
single-shape systems.  If $T$ consists of unit-square tiles, $\calT$ is said to
be a \emph{monomino} system.

\paragraph{Assembly Process}
Given a tile-assembly system $\calT = (T,\sigma,\tau)$, we now define the set
of \emph{producible} assemblies $\prodasm{T}$ that can be derived from $\calT$,
as well as the \emph{terminal} assemblies, $\termasm{T}$, which are the
producible assemblies to which no additional tiles can attach.  The assembly
process begins from $\sigma$ and proceeds by single steps in which any single
copy of some tile $t \in T$ may be attached to the current assembly $A$, provided that it
can be translated so that its placement does not overlap any previously placed
tiles and it binds with strength $\ge \tau$.  For a system $\calT$ and assembly
$A$, if such a $t\in T$ exists, we say $A \rightarrow^\calT_1 A'$ (i.e. $A$
grows to $A'$ via a single tile attachment).  We use the notation $A
\rightarrow^\calT A''$, when $A$ grows into $A''$ via 0 or more steps.
Assembly proceeds asynchronously and nondeterministically, attaching one tile
at a time, until no further tiles can attach.  An assembly sequence in a TAS
$\mathcal{T}$ is a (finite or infinite) sequence $\vec{\alpha} = (\alpha_0 =
\sigma,\alpha_1,\alpha_2,\ldots)$ of assemblies in which each $\alpha_{i+1}$ is
obtained from $\alpha_i$ by the addition of a single tile.
The set of producible assemblies $\prodasm{T}$ is defined to be the set of all
assemblies $A$ such that there exists an assembly sequence for $\calT$ ending
with $A$ (possibly in the limit).  The set of \emph{terminal} assemblies
$\termasm{T} \subseteq \prodasm{T}$ is the set of producible assemblies such
that for all $A \in \termasm{T}$ there exists no assembly $B \in \prodasm{T}$
in which $A\rightarrow^\calT_1 B$.  A system $\calT$ is said to be directed if
$|\termasm{T}|=1$, i.e., if it has exactly one terminal assembly.

Note that the aTAM is simply a specific case of the polyTAM in which all tiles are monominoes, i.e.,
single unit squares.

\section{Universal Computation with Geometric Bit-Reading}\label{sec:comp-overview}

In this section we provide an overview of how universal computation can be
performed in a temperature-1 system with appropriate use of geometric aspects
of tiles and assemblies. Refer to Figure~\ref{fig:bit-gadget-definition} for an
intuitive illustration.

\subsection{Bit-Reading Gadgets}\label{sec:bit-gadgets}

First, we discuss a primitive tile-assembly component that enables computation
by self-assembling systems.  This component is called the \emph{bit-reading
gadget}, and essentially consists of pre-existing assemblies that appropriately
encode bit values (i.e., $0$ or $1$) and paths that grow past them and are able
to ``read'' the values of the encoded bits; this results in those bits being
encoded in the tile types of the paths beyond the encoding assemblies.  In tile-assembly systems
in which the temperature is $\ge 2$, a bit-reader gadget is
trivial:  the assembly encoding the bit value can be a single tile with an
exposed glue that encodes the bit value, and the path that grows past to read
the value simply ensures that a tile must be placed cooperatively with, and
adjacent to, that encoding the bit (i.e., the path forces a tile to be placed
that can only bind if one of its glues matches that exposed by the last tile
of the path, and the other matches the glue encoding the bit value). However,
in a temperature-1 system, cooperative binding of tiles cannot be enforced, and
therefore the encoding of bits must be done using geometry.
Figure~\ref{fig:bit-gadget-definition} provides an intuitive overview of a
temperature-1 system with a bit-reading gadget.  Essentially, depending on
which bit is encoded by the assembly to be read, exactly one of two types of
paths can complete growth past it, implicitly specifying the bit that was
{\em read}.  It is important that the bit reading must be unambiguous, i.e.,
depending on the bit {\em written} by the pre-existing assembly, exactly one type
of path (i.e., the one that denotes the bit that was written) can possibly
complete growth, with all paths not representing that bit being prevented.
Furthermore, the correct type of path must always be able to grow.  Therefore,
it cannot be the case that either all paths can be blocked from growth, or that
any path not of the correct type can complete, regardless of whether a path of
the correct type also completes, and these conditions must hold for any valid
assembly sequence to guarantee correct computation.

\begin{figure}[htp]
\begin{center}
\vspace{-15pt}
\includegraphics[width=5.0in]{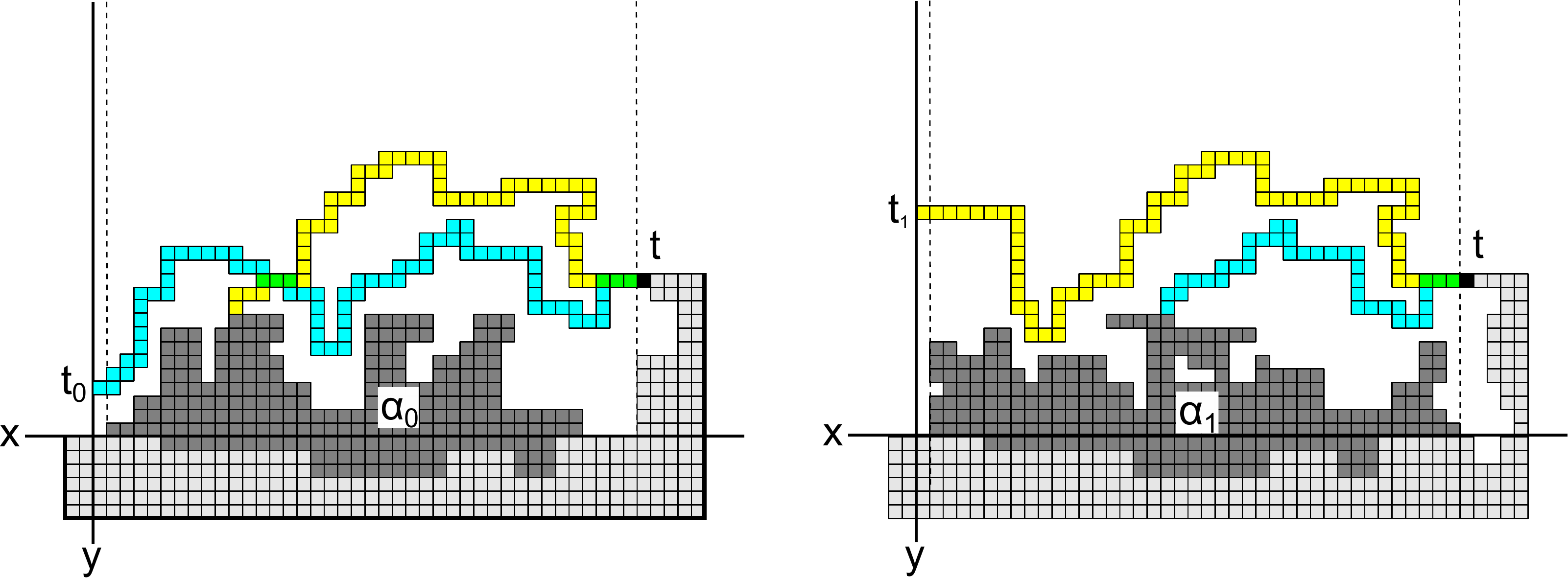}
\caption{Abstract schematic of a bit-reading gadget. (Left) The path grown from
$t$ ``reads'' the bit 0 from $\alpha_0$ (by being allowed to grow to $x=0$ and
placing a tile $t_0 \in T_0$), while the path which could potentially read a 1
bit is blocked by $\alpha_0$. (Right) The path grown from $t$ reads the bit 1
from $\alpha_1$, while the path that could potentially read a 0 is blocked by
$\alpha_1$.  Clearly, the specific geometry of the polyomino tiles used and
assemblies is important in allowing the yellow path in the left figure to be
blocked without also blocking the blue path.}
\label{fig:bit-gadget-definition}
\vspace{-10pt}
\end{center}
\end{figure}

\begin{definition}\label{def:bit-reader}
We say that a \emph{bit-reading gadget} exists for a tile assembly system
$\mathcal{T} = (T,\sigma,\tau)$, if the following hold.  Let $T_0 \subset T$ and
$T_1 \subset T$, with $T_0 \cap T_1 = \emptyset$, be subsets of tile types
which represent the bits $0$ and $1$, respectively.  For some producible
assembly $\alpha \in \prodasm{T}$, there exist two connected subassemblies,
$\alpha_0,\alpha_1 \sqsubseteq \alpha$ (with $w$ equal to the maximal width of
$\alpha_0$ and $\alpha_1$), such that if: \begin{enumerate}
    \item $\alpha$ is translated so that $\alpha_0$ has its minimal $y$-coordinate $\le 0$ and its minimal $x$-coordinate $= 1$,
    \item a tile of some type $t \in T$ is placed at $(w+n,h)$, where $n,h \ge 1$, and
    \item the tiles of $\alpha_0$ are the only tiles of $\alpha$ in the first quadrant to the left of $t$,
\end{enumerate}
then at least one path must grow from $t$ (staying strictly above the $x$-axis)
and place a tile of some type $t_0 \in T_0$ as the first tile with
$x$-coordinate $= 0$, while no such path can place a tile of type $t' \in (T
\setminus T_0)$ as the first tile to with $x$-coordinate $= 0$.  (This
constitutes the reading of a $0$ bit.)

Additionally, if $\alpha_1$ is used in place of $\alpha_0$ with the same
constraints on all tile placements, $t$ is placed in the same location as
before, and no other tiles of $\alpha$ are in the first quadrant to the left of
$t$, then at least one path must grow from $t$ and stay strictly above the
$x$-axis and strictly to the left of $t$, eventually placing a tile of some
type $t_1 \in T_1$ as the first tile with $x$-coordinate $= 0$, while no such
path can place a tile of type $t' \in (T \setminus T_1)$ as the first tile with
$x$-coordinate $= 0$.  (Thus constituting the reading of a $1$ bit.)
\end{definition}

We refer to $\alpha_0$ and $\alpha_1$ as the {\em bit writers}, and the paths
which grow from $t$ as the {\em bit readers}.  Also, note that while this
definition is specific to a bit-reader gadget in which the bit readers grow from
right to left, any rotation of a bit reader is valid by suitably rotating the
positions and directions of Definition~\ref{def:bit-reader}.  As mentioned in
Figure~\ref{fig:bit-gadget-definition}, depending on the actual geometries of
the polyominoes used and their careful placement, it may be possible to enforce
the necessary blocking of all paths of the wrong type, while still allowing at
least one path of the correct type to complete growth in any valid assembly
sequence.  The necessary requirements on these geometries and placements are
the subject of the novel results of this paper.

\subsection{Turing-Machine Simulation} \label{sec:Turing-overview}

In order to show that a polyomino shape (i.e., a system composed of tiles of
only that shape) is computationally universal at $\tau=1$, we show how it is
possible to simulate an arbitrary Turing machine using such a polyomino system.
In order to simulate an arbitrary Turing machine, we show how to self-assemble
a zig-zag Turing machine~\cite{CookFuSch11,SingleNegative}.  A zig-zag Turing
machine at $\tau=1$ works by starting with its input row as the seed
assembly, then growing rows one by one, alternating growth from left to right with
growth from right to left.  As a row grows across the top of the row immediately beneath
it, it does so by forming a path of single tile width, with tiles connected by glues,
which pass information horizontally through their glues, while the geometry of
the row below causes only one of two choices of paths to grow at regular
intervals, effectively passing information vertically via the geometry, using
bit-reading gadgets.

\begin{figure}[ht!]
\begin{center}
\includegraphics[width=3.2in]{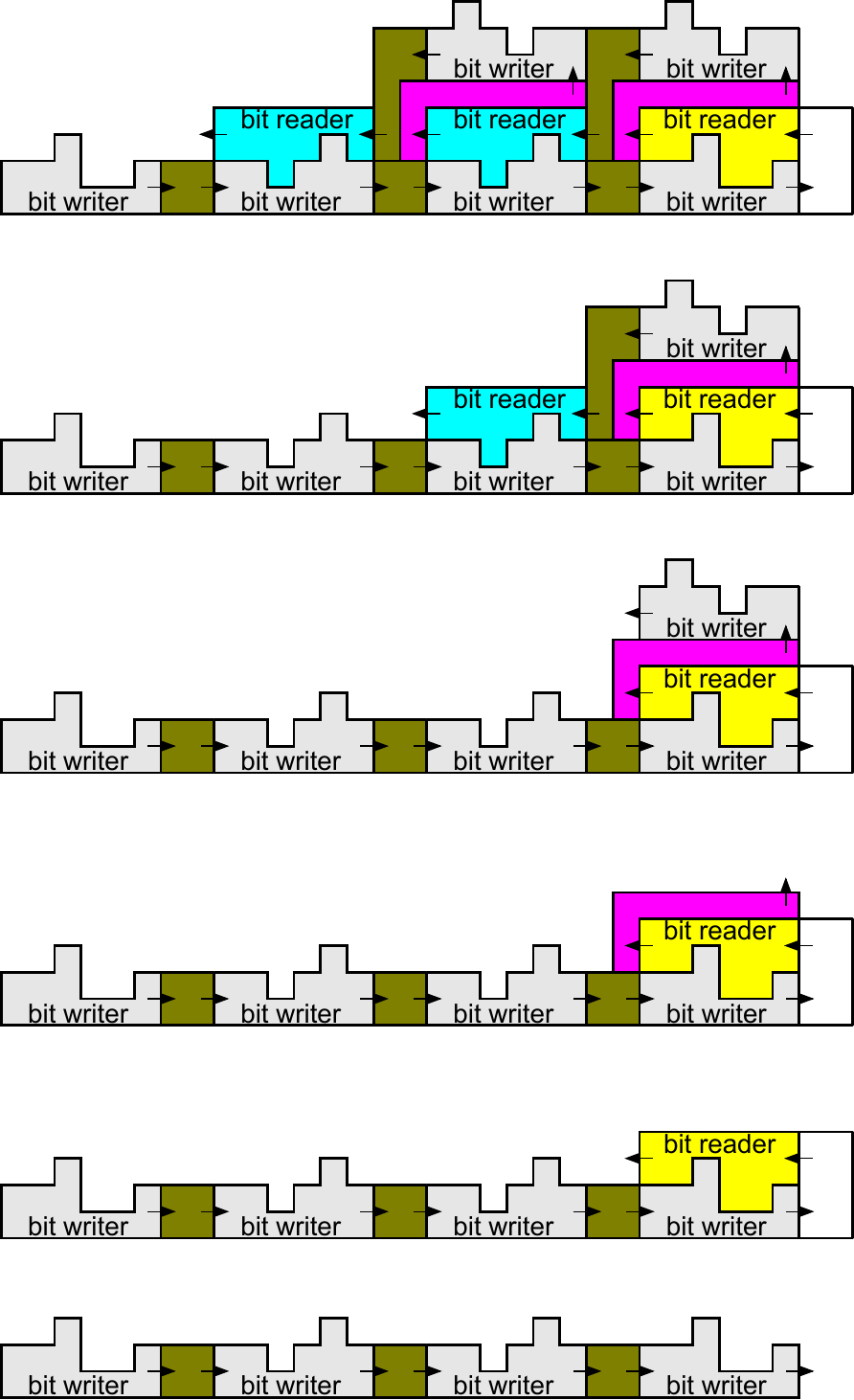}
\caption{High-level schematic view of a zig-zag Turing machine and the bit-reading/writing
gadgets that make up each row of the simulation.  The bottom
shows the seed row, consisting of bit-writer gadgets separated by spacers.
Then, depicted as consecutive upward figures, the second row begins its growth.
Yellow/blue portions depict locations of bit-reader gadgets (for $0$ and $1$,
respectively), which grow pink paths upward after completing in order to grow
bit writer gadgets (grey), and then gold spacers back down to the point where
the next bit reader can grow.} \label{fig:bit-reading-writing-scheme1}
\end{center}
\end{figure}

Each cell of the Turing machine's tape is encoded by a series of bit-reader
gadgets that encode in binary the symbol in that cell and, if the read/write
head is located there, what state the machine is in.  Additionally, as each
cell is read by the row above, the necessary information must be geometrically
written above it so that the next row can read it.  See
Figure~\ref{fig:bit-reading-writing-scheme1} for an example depicting a
high-level schematic without showing details of the individual polyominoes.
Figure~\ref{fig:bit-reading-writing-scheme2} shows the same system after two
rows have completed growth.

\begin{figure}[htp]
\begin{center}
\vspace{-15pt}
\includegraphics[width=3.0in]{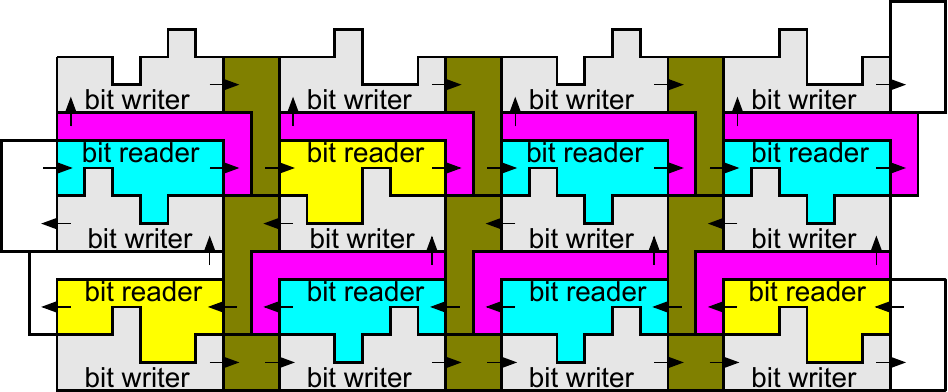}
\caption{High-level schematic view of a zig-zag Turing machine and the bit-reading/writing gadgets that make up the first two rows of simulation.}
\label{fig:bit-reading-writing-scheme2}
\vspace{-10pt}
\end{center}
\end{figure}

For a more specific example that shows the placement of individual, actual
polyomino tiles as well as the order of their growth, see
Figure~\ref{fig:plus-sign-gadgets}. Note that the simulation of a zig-zag
Turing machine can be performed by horizontal or vertical growth, and in any
orientation.

\begin{figure}[htp]
\begin{center}
\includegraphics[width=5.5in]{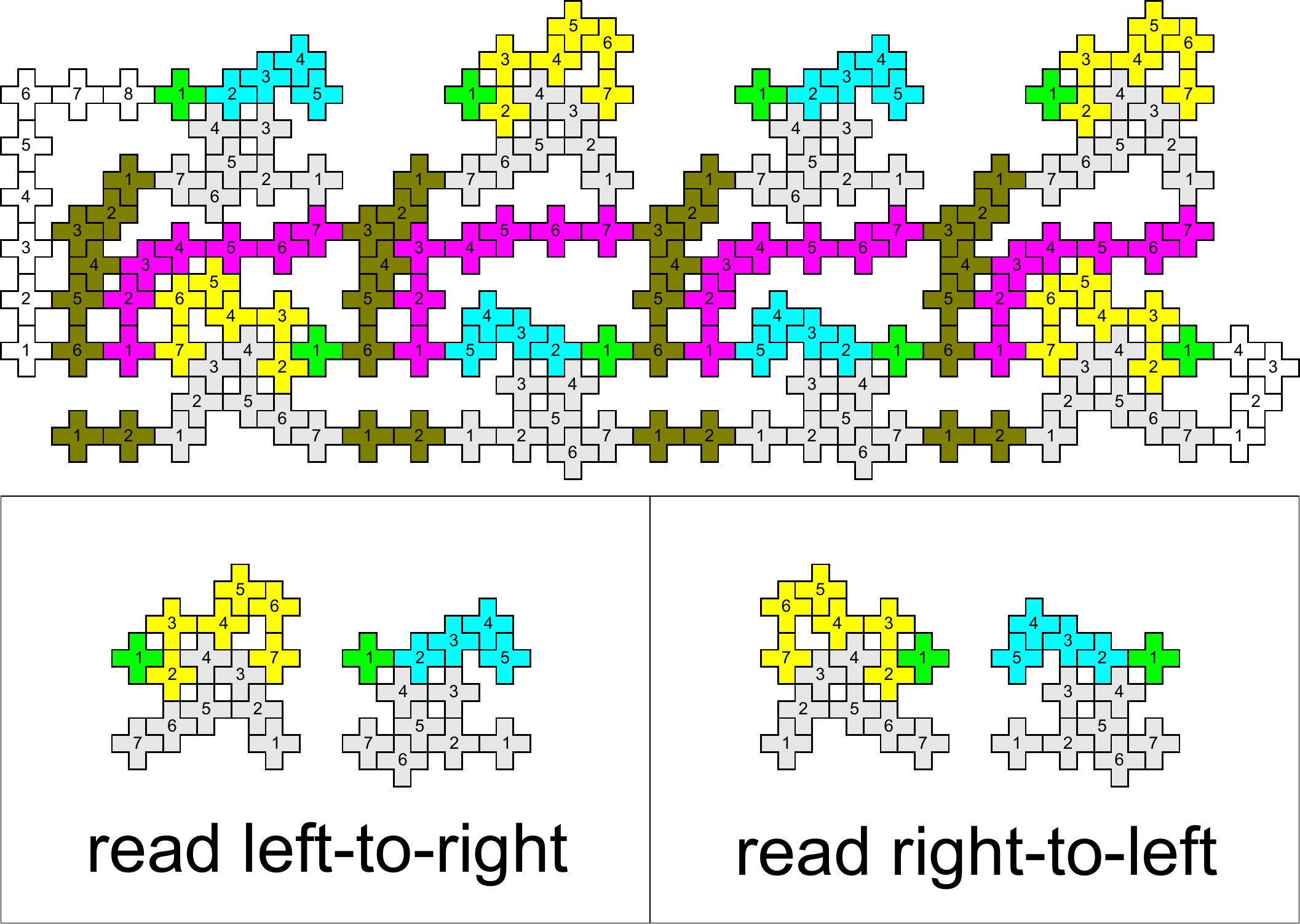}
\caption{The system of Figure~\ref{fig:bit-reading-writing-scheme1} after two
rows of the zig-zag simulation have been completed (omitting the output bit writer
gadgets of the second row), implemented with ``plus-sign'' polyominoes.  The
bottom left shows $0$ and $1$ bit writer and reader combinations, with the
writer having grown from right to left and the reader from left to right.
The bottom right shows the same, but with growth directions reversed.  Grey
tiles represent bit-writer gadgets.  Green tiles represent the beginning of bit-reader
gadgets that are common to either bit; yellow represents the path that
can grow to signify a $0$ bit being read, and blue a $1$ bit.  Other colors
correspond to those for the gadgets used in
Figure~\ref{fig:bit-reading-writing-scheme1}, with numbers corresponding to the
growth order of the tiles in each gadget.} \label{fig:plus-sign-gadgets}
\end{center}
\end{figure}

\section{Technical Lemmas}\label{sec:tech-lemmas}

\subsection{Grids of Polyominoes}

As mentioned above, in order to show that all polyominoes of size greater than 2 are universal, we show that a bit-reading gadget can be constructed with these polyominoes. In this section we show two lemmas about single-shaped polyTAM systems that will aid in the construction of bit-writers used to show that any polyomino of size greater than $2$ can be used to define a single-shape polyTAM system capable of universal computation. Throughout this section, any mention of a polyTAM system refers to a single-shape system.

The following lemma says that if a polyomino $P$ can be translated by a vector $\vec{v}$ so that no pixel positions of the translated polyomino overlap the pixel positions of $P$, then for any integer $c \neq 0$, no pixel positions of $P$ translated by $c\cdot \vec{v}$ overlap the pixel positions of $P$. The proof of Lemma~\ref{lem:chain} can be found in~\cite{OneTile}; the statement of the lemma has been included for the sake of completeness.

\begin{lemma}
\label{lem:chain}
Consider a two-dimensional, bounded, connected, regular closed set $S$, i.e., $S$ is equal to the topological closure of its interior points.
Suppose $S$ is translated by a vector $v$ to obtain shape $S_{v}$, such that $S$ and $S_{v}$ have disjoint interiors. Then the shape $S_{c*v}$ obtained by translating $S$ by $c*v$ for any integer $c \neq 0$ and $S$ have disjoint interiors.
\end{lemma}

Informally, the following lemma says that any polyomino gives rise to a polyTAM system that can produce an infinite ``grid'' of polyominoes.

\begin{lemma} \label{lem:grids}
Given a polyomino $P$. There exists a directed, singly seeded, single-shape tile system $\mathcal{T} = (T, \sigma)$ (where the seed is placed so that pixel $p \in P$ is at location $(0,0)$ and the shape of tiles in $T$ is $P$) and vectors $\vec{v}, \vec{w} \in \mathbb{Z}^2$, such that $\mathcal{T}$ produces the terminal assembly $\alpha$, which we refer to as a \emph{grid}, with the following properties. (1) Every position in $\alpha$ of the form $c_1\vec{v} + c_2\vec{w}$, where $c_1, c_2 \in \mathbb{Z}$, is occupied by the pixel $p$, and (2) for every $c_1, c_2\in \mathbb{Z}$, the position in $\Z^2$ of the form $c_1\vec{v} + c_2\vec{w}$ is occupied by the pixel $p$ for some tile in $\alpha$.
\end{lemma}

\begin{proof}

Let $P$ be a polyomino, consisting of $n$ pixels. Let $D_x$ be the biggest difference between two $x$-values
of pixels in $P$, and let $(p_1,p_2)$ and $(q_1,q_2)$ be two pixels in $P$ that attain this difference such that $p_1\leq q_1$.
Then define $\vec{v}=(q_1-p_1 + 1, q_2-p_2)=(D_x+1, q_2 - p_2 )$.

Furthermore, let $d_y$ be the biggest difference between two $y$-values of pixels in $P$ that are in the same column. In other words,
$d_y = \max\{|y_1 - y_0| \text{ such that } (x, y_0) \text{ and } (x,y_1) \text{ are pixels in } P\}$. Then let $\vec{w}$ be the vector $(0,d_y + 1)$.
For a depiction of the vectors $\vec{v}$ and $\vec{w}$ as well as the distances $D_x$ and $d_y$, see Figure~\ref{fig:polyomino-dimensions}.

\begin{figure}[htp]
\begin{center}
\includegraphics[width=3.0in]{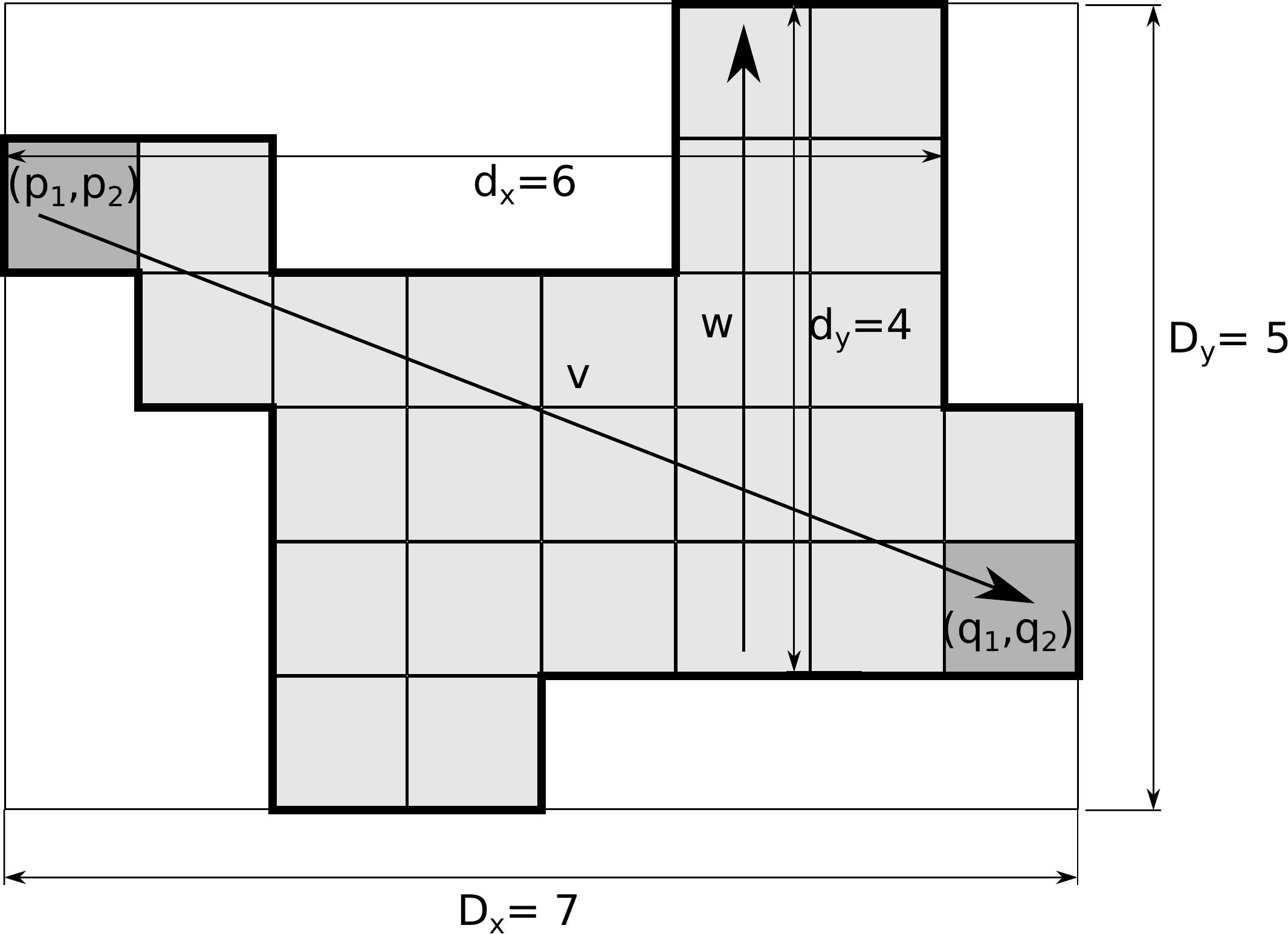}
\caption{Dimensions of an example polyomino}
\label{fig:polyomino-dimensions}
\end{center}
\end{figure}

Now define $P[i,j]=p+i\cdot \vec{v}+j\cdot \vec{w}$ for $i, j \in \Z^2$. Here, $p$ acts as a distinguished pixel that we use as a reference point as discussed above.  Then, for two polyominoes $P[i,j]$ and $P[k,l]$, we say that these polyominoes are \emph{neighboring} if $i=k$ and $|j-l|=1$, or $j=l$ and $|i-k|=1$.

{\it Claim:} $P[i,j]$ for all $i,j\in \Z^2$ defines a grid of non-overlapping polyominoes such that any two neighboring polyominoes $P[i,j]$ and $P[k,l]$ contain pixels with a shared edge. See Figure~\ref{fig:polyomino-grid} for an example of such a grid of polyominoes.

\begin{figure}[htp]
\begin{center}
\includegraphics[width=3.0in]{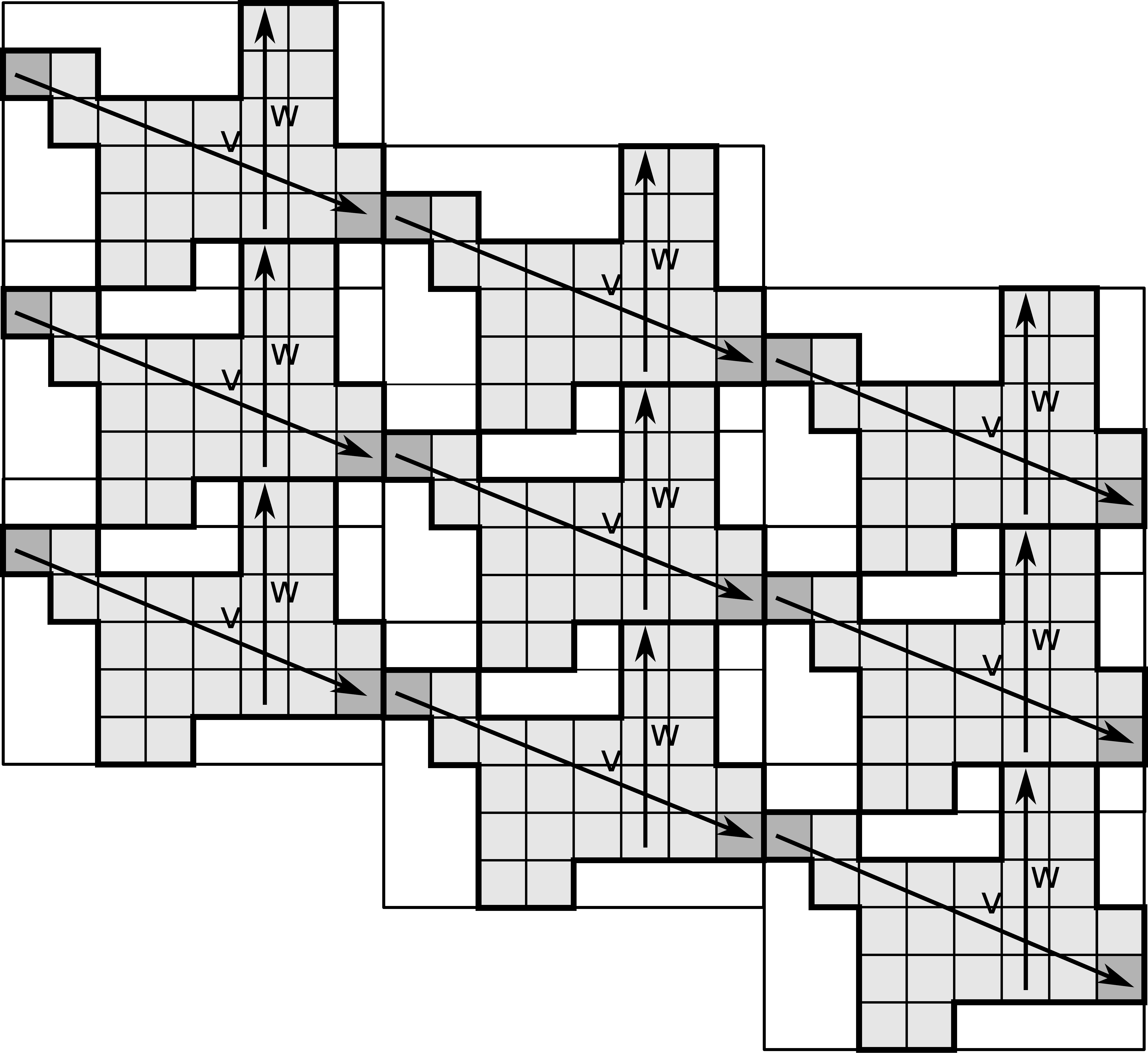}
\caption{A lattice that can be formed from the polyomino in Figure~\ref{fig:polyomino-dimensions} by translating using the vectors $\vec{v}$ and $\vec{w}$.}
\label{fig:polyomino-grid}
\end{center}
\end{figure}

Informally, this grid is formed by finding extremal horizontal points of $P$, a furthest left and furthest right, and using those as the connection points between horizontal neighbors in the grid (which may cause some vertical shifting as well).  Then, the vertical adjacencies are formed by simply sliding one copy of $P$ down toward another copy, with both in the same horizontal alignment, until they meet.  The point of first contact determines their neighboring edges.  In this way, an infinite grid can be formed.

To prove this claim, first note that $P[i,j]$ and $P[k,l]$ are disjoint if and only if $i\neq k$ or $j\neq l$.  In particular, if $i\neq k$, then the pixels $p$ of $P[i,j]$ and $P[k,l]$ are separated by at least a horizontal distance of $D_x+1$, and hence $P[i,j]$ and $P[k,l]$ cannot overlap; and if $i=k$ and  $j\neq l$, then pixels with the same $x$-coordinate are separated by at least $d_y$, and so they cannot overlap since $d_y$ was chosen to be the maximum distance of $y$-values of pixels in $P$ that are in the same column.

Finally, suppose that $P[i,j]$ and $P[k,l]$ are neighbors in the grid. First, suppose that $j=l$ and $|i-k|=1$. We will consider the case where $j=l$ and $i=k + 1$; the case where $i=k - 1$ is similar. Then, $P[i,j]$ has a leftmost pixel at $(p_1 + i\cdot (q_1-p_1 + 1), p_2 + i\cdot (q_2- p_2) + j\cdot (d_y+ 1))$, which equals $(p_1 + (k+ 1)\cdot (q_1-p_1 + 1), p_2 + (k+1)\cdot (q_2- p_2) + l\cdot (d_y+ 1)) = (q_1 + 1 + k\cdot (q_1-p_1 + 1), q_2 + k\cdot (q_2- p_2) + l\cdot (d_y+ 1))$. Therefore, the rightmost pixel of $P[k,l]$ located at $(q_1 + k\cdot (q_1-p_1 + 1), q_2 + k\cdot (q_2- p_2) + l\cdot (d_y+ 1))$ is $1$ unit to the left of the pixel of $P[i,j]$ located at $(p_1 + i\cdot (q_1-p_1 + 1), p_2 + i\cdot (q_2- p_2) + j\cdot (d_y+ 1))$. Hence, $P[i,j]$ and $P[k,l]$ contain pixels with a common edge.

Now suppose that $i=k$ and $|j-l|=1$. We will consider the case where $i=k$ and $j=l + 1$; the case where $j=l - 1$ is similar. In this case, let $(x, y)$ and $(x,y + d_y)$ be pixels of $P$. Then, $P[i,j]$ contains a pixel, $p$ say, at $(x + i\cdot (q_1-p_1 + 1), y + i\cdot (q_2- p_2) + j\cdot (d_y+ 1))$, which equals $(x + k\cdot (q_1-p_1 + 1), y + k\cdot (q_2- p_2) + (l + 1)\cdot (d_y+ 1)) = (x + k\cdot (q_1-p_1 + 1), (y + d_y) + 1 + k\cdot (q_2- p_2) + l\cdot (d_y+ 1))$. Therefore, $p$ is $1$ unit above the pixel of $P[k,l]$ located at $(x + k\cdot (q_1-p_1 + 1), y + d_y + k\cdot (q_2- p_2) + l\cdot (d_y+ 1))$.

To finish the proof, note that for any polyomino $P$ we can give a directed, singly seeded tile system $\mathcal{T} = (T, \sigma)$ with single shape $P$ by letting $T$ consist of a single tile and assigning a single glue to the appropriate edges of a polyomino tile such that the binding of two polyomino tiles enforces a shift by $\vec{v}$ or $\vec{w}$.

\end{proof}

Note that using $D_x$ and $D_y$ or using $d_x$ and $d_y$ (instead of $D_x$ and $d_y$) does not work;
in both cases it is easy to come up with examples for which there is an overlap.

Let $p, p' \in P$ be distinct pixels in the polyomino $P$ at positions $(x,y)$ and $(x',y')$ respectively, let $\vec{r}=(x-x',y-y')$, and let $\vec{v}, \vec{w}$ be as defined in Lemma~\ref{lem:grids}. Then, if there exists $c_1,c_2 \in \mathbb{Z}$ such that $(x,y)+c_1\vec{v} + c_2\vec{w} = (x',y')$, we say that the polyomino which occupies $(x',y')$ is $\vec{r}$-shifted with respect to (or relative to) the polyomino at $(x,y)$.  If a polyomino at position $(x',y')$ is $0$-shifted with respect to a polyomino at $(x,y)$, we say that the polyomino at position $(x',y')$ is \emph{on grid} with the polyomino at $(x,y)$.  If a polyomino is not on grid with  a polyomino at $(x,y)$, we say that the polyomino is \emph{off grid}. Henceforth, if we do not mention the tile which another tile is shifted in respect to, assume that the tile is shifted with respect to the seed.

For the remainder of this section, for a polyomino $P$, we let $V \subset \mathbb{Z}^2$ denote the set of vectors such that $\vec{r} \in V$ provided that there exists some directed, singly seeded system, single shape $\mathcal{T}' = (T',\sigma')$ with shape given by $P$ whose terminal assembly $\alpha'$ contains an $\vec{r}$-shifted polyomino tile, and we let $B$ denote the subset of vectors $B \subset V$ such that $\vec{b} \in B$ provided that there exists a directed, singly seeded system, single shape $\mathcal{T} = (T,\sigma)$ with shape given by $P$ such that the terminal assembly $\alpha$ of $\mathcal{T}$ consists of exactly two tiles: the seed tile $\sigma$ and a $\vec{b}$-shifted tile. $B$ can be thought of as the set of vectors such that the polyomino $P$ and a copy of $P$ shifted by a vector in $B$ are non-overlapping and contain pixels that share a common edge. We can think of $B$ as a set of basis vectors for $V$ in the following sense. If $\vec{r}\in V$, then $\vec{r}$ can be written as a linear combination of shifts in $B$. The following lemma is a more formal statement of this fact.

\begin{lemma}\label{lem:vect-v}
For any vector $\vec{r}\in V$, $\vec{r}=\Sigma c_i\vec{b_i}$ for some $c_i \in \mathbb{Z}$ and $\vec{b_i} \in B$.
\end{lemma}

\begin{proof}
This follows from the fact that if $\vec{r}\in V$, then there exists some directed, singly seeded system $\mathcal{T} = (T,\sigma)$ which contains an $\vec{r}$-shifted polyomino tile $A$. Then, there must be a path of neighboring polyomino tiles from the seed tile $S$ to $A$. Starting from $S$, each consecutive tile along this path to $A$ must be a $\vec{b}$-shifted tile for some $\vec{b}$ in $B$, and the sum of these vectors is $\vec{r}$.
\end{proof}

For a rectangle $R$ and a tile $T$, we say that $T$ \emph{lies in the southeast \emph{(respectively northwest)} corner} of $R$ iff the south and east edges of the bounding rectangle of the tile $T$ lie on the south and east edges of $R$.
Let $V$ contain every linear combination $\Sigma c_i\vec{b_i}$ for $c_i \in \mathbb{Z}$ and $\vec{b_i} \in B$. The next two lemmas formalize this notion. For $c\in \Z$ and $\vec{b}\in B$, the following lemma shows how if $\vec{r} = c\cdot \vec{b}$, we can give a system $\mathcal{T}_{\vec{r}}$ that contains an $\vec{r}$-shifted tile. Furthermore, the properties given in the lemma statement ensure that if we have two such system, $\mathcal{T}_{\vec{r}_1}$ and $\mathcal{T}_{\vec{r}_2}$, corresponding to two shift vectors $\vec{r}_1$ and $\vec{r}_2$, $\mathcal{T}_{\vec{r}_1}$ and $\mathcal{T}_{\vec{r}_2}$ can be ``concatenated'' to give a system that contains an $(\vec{r}_1 + \vec{r}_2)$-shifted polyomino tile. See Figure~\ref{fig:simpleShiftCases} for schematic depictions of the properties given in the following lemma.
\vspace{-14pt}
\begin{figure}[htp]
\centering
    \subfloat[\label{subfig:simpleShiftCase1}]{%
      \raisebox{0.3\height}{\includegraphics[width=.5in]{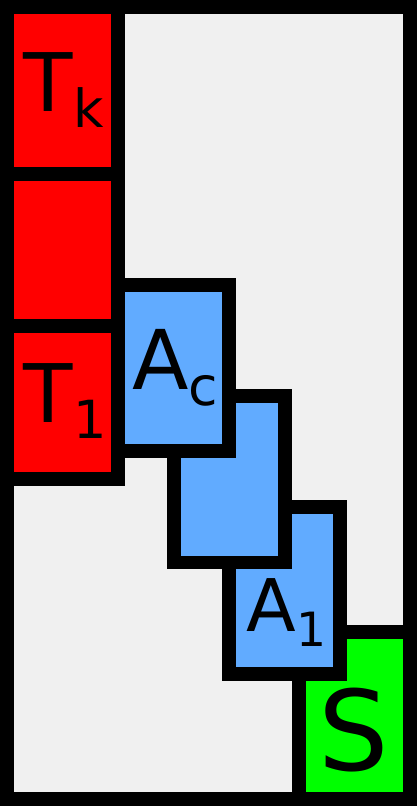}}
    }
    \quad\quad\quad
    \subfloat[\label{subfig:simpleShiftCase2}]{%
       \raisebox{0.15\height}{\includegraphics[width=.5in]{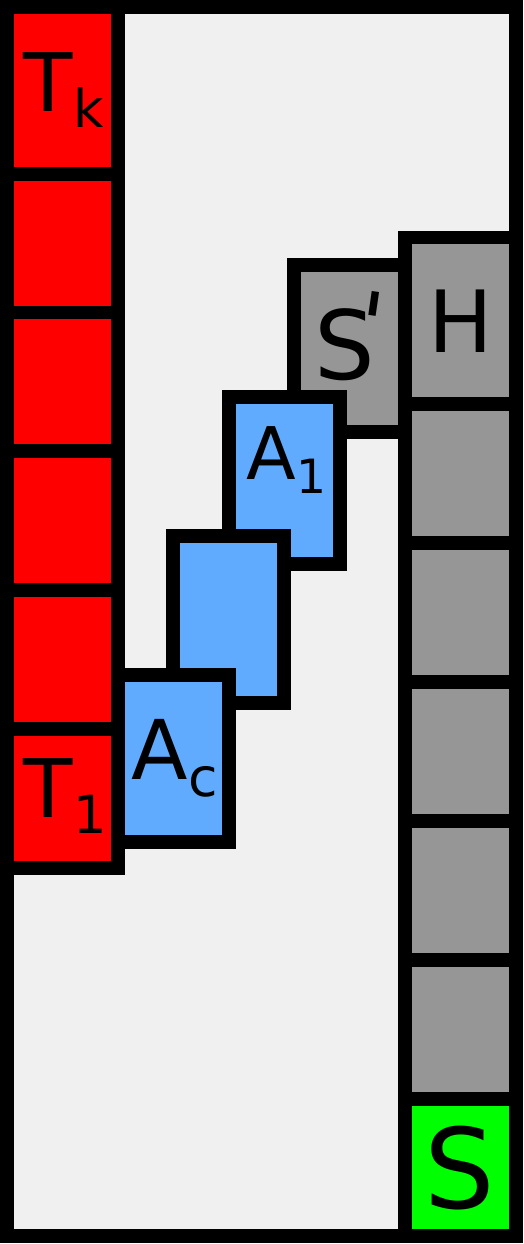}}
    }
    \quad\quad\quad
    \subfloat[\label{subfig:simpleShiftCase3}]{%
      \includegraphics[width=.5in]{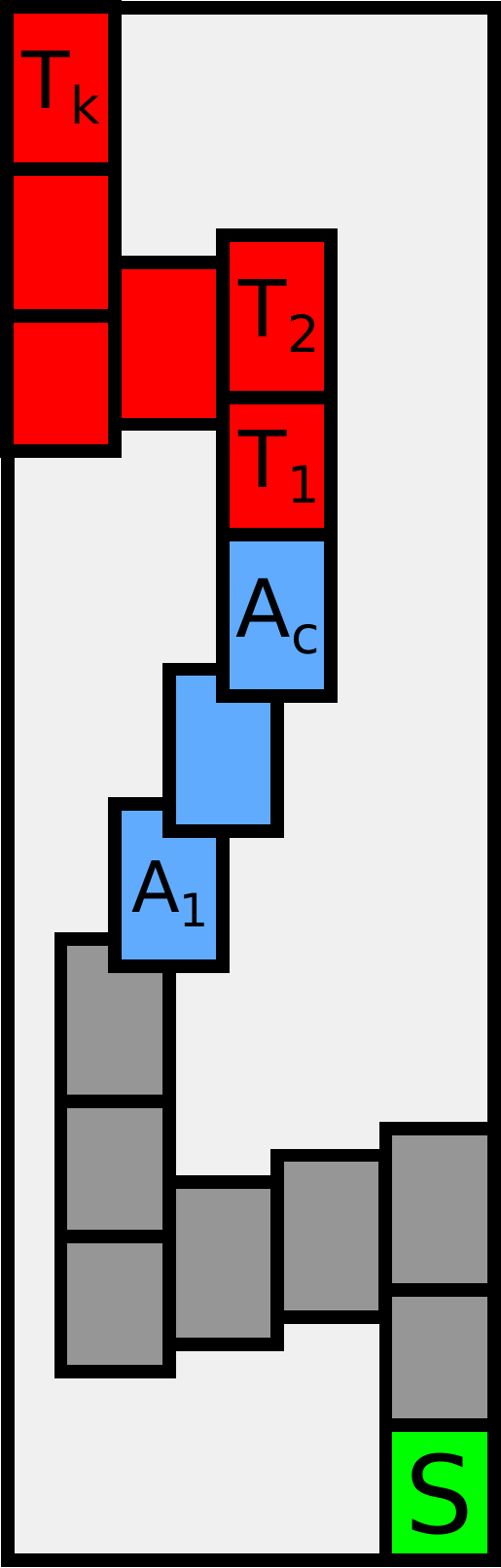}
    }
    \quad\quad\quad
    \subfloat[\label{subfig:simpleShiftCase4}]{%
      \raisebox{0.25\height}{\includegraphics[width=.55in]{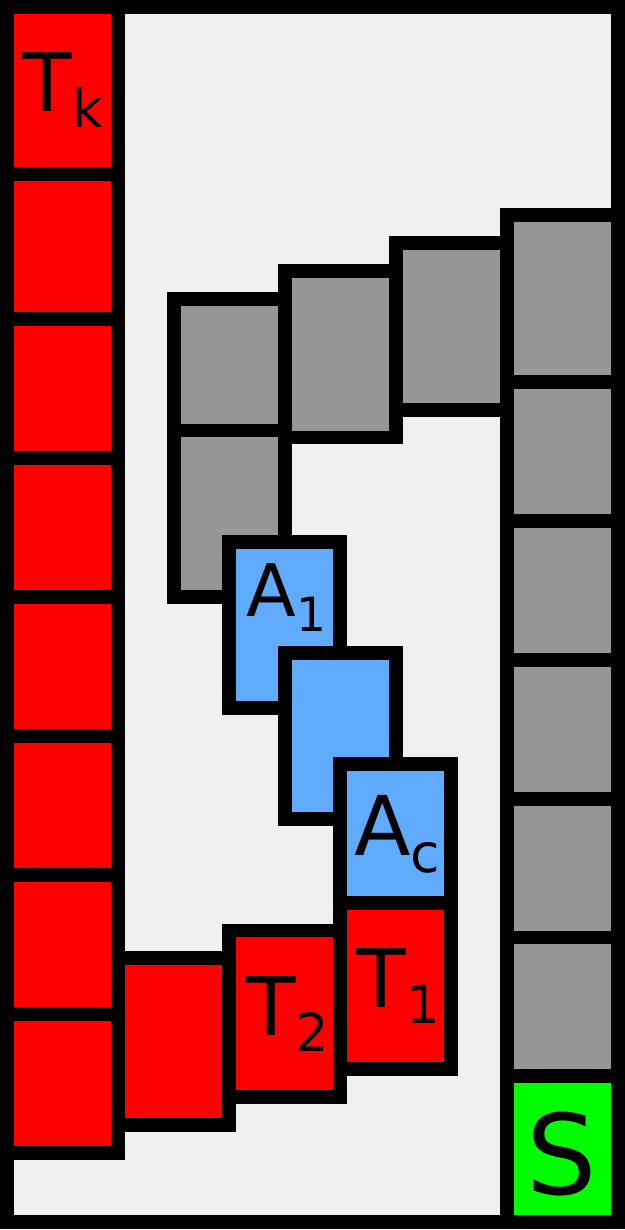}}
    }
    \caption{A depiction of the properties given in Lemma~\ref{lem:simpleShiftCases} for four different cases. Each of the small rectangles (red, green, blue or gray in color) serves as the bounding rectangle of some polyomino. Hence, $\alpha$ is contained in the union of all of the regions bounded by these rectangles. The seed tile $S$ of $\alpha$ is contained in the green rectangle. As assembly proceeds from the seed the first $r$-shifted polyomino tile is $A_c$. From this point on in the assembly process, the rectangles are on grid with $A_c$.}
    \label{fig:simpleShiftCases}
  \end{figure}

\begin{lemma}\label{lem:simpleShiftCases}
Let $\vec{r} = c\cdot \vec{b}$ for any $c \in \mathbb{Z}$ and $\vec{b} \in B$. Then there exists some directed, singly seeded system $\mathcal{T} = (T, \sigma)$ with all tiles shaped $P$ such that the terminal assembly $\alpha$ of $\mathcal{T}$ contains an $\vec{r}$-shifted tile. Moreover, the system $\mathcal{T}$ and assembly $\alpha$ have the following properties.
\begin{enumerate}
\item There is a single assembly sequence that yields $\alpha$,
\item For some $m,n\in \N^2$, $\alpha$ is contained in an $m\times n$ rectangle $R$ and the seed tile $S$ lies in the southeast corner of $R$, and
\item the last tile, $A$, to attach to $\alpha$ lies in the northwest corner of $R$.
\end{enumerate}
\end{lemma}

\begin{proof}
By possibly negating $\vec{b}$, we can assume that $c\geq0$ without loss of generality. Let $\vec{b} = (b_1, b_2)$, let $\vec{v}$ and $\vec{w}$ be vectors given by Lemma~\ref{lem:grids}, and let the dimensions of the bounding rectangle of $P$ be $p\times q$ so that the height is $p$ and the width is $q$. Then consider the following four cases: (a) $b_1 < 0$ and $b_2 \geq 0$ (shown in Subfigure~\ref{subfig:simpleShiftCase1}), (b) $b_1 < 0$ and $b_2 < 0$ (shown in Subfigure~\ref{subfig:simpleShiftCase2}), (c) $b_1 \geq 0$ and $b_2 \geq 0$ (shown in Subfigure~\ref{subfig:simpleShiftCase3}), and (d) $b_1 \geq 0$ and $b_2 < 0$ (shown in Subfigure~\ref{subfig:simpleShiftCase4}). In each of these cases, we will define tiles that give rise to the system $\mathcal{T}$ with unique terminal assembly $\alpha$ such that Properties 1-3 hold.

In Case (a), starting from the seed tile $S$ we can assign appropriate glues to $S$ and a tile $A_1$ such that $A_1$ and no other tile can attach to $S$, and such that $A_1$ is $\vec{b}$-shifted. Similarly, we can assign appropriate glues to $A_2$ so that $A_2$ and no other tile can attach to $A_1$, and moreover, we can ensure that $A_2$ is $\vec{b}$-shifted relative to $A_1$.  This must be possible by the existence of $\vec{b}$ and Lemma~\ref{lem:chain}. Then, $A_2$ will be $2\cdot \vec{b}$-shifted relative to $S$. After attaching $c$ total $\vec{b}$-shifted (relative to the previously attached tile), the last $\vec{b}$-shifted tile $A_c$ is attached. This assembly $\alpha_c$ of $\mathcal{T}$ is such that (1) $\alpha_c$ has a single assembly sequence, and (2) $\alpha_c$ contains a tile $A_c$ that is $c\cdot \vec{b}$-shifted relative to $S$. Now, using shifts by $\vec{v}$ and $\vec{w}$ of the grid, we can define tiles so that a path of tiles from $A_c$ can assemble that begins growth by attaching a tile $T_1$ to the left of $A_c$ (according to the vector $\vec{v}$), and then grows a vertical path of tiles $T_2, \ldots, T_k$ with each successive tile placed at locations that are translations of the previous tile by $\vec{w}$. In other words, using $\vec{w}$, we define tiles that form a vertical path of tiles above $T_1$ such that each tile of this path is on grid with $T_1$, and therefore, $c\cdot \vec{b}$-shifted relative to $S$. By defining sufficiently many such tiles, we can ensure that this vertical path places a final tile $T_k$ ($=A$ in the lemma statement) such that Property 3 holds. Note that Properties 1 and 2 hold by virtue of our tile definitions. The tiles denoted in Case (a) are depicted in Subfigure~\ref{subfig:simpleShiftCase1}.

In Case (b), we define tiles that allow for a vertical path of tiles to grow above $S$ with each successive tile placed at locations that are translations by $\vec{w}$ of the previous tile's locations. The top tile, denoted by $H$ say, of this vertical path can be defined so that it allows for the attachment of a tile $S'$ to its left such that $S'$ is on grid with $S$. To obtain the tile location of $S'$, we can translate $H$'s tile locations by $-\vec{v}$. Now, we can choose this vertical path so that the vertical distance from the south edge of the bounding rectangle of $S'$ to the south edge of $R$ is greater than $c\cdot p + p$. These tiles are schematically depicted as the gray rectangles in Subfigure~\ref{subfig:simpleShiftCase2}. The $c\cdot p$ summand ensures that we can now define tiles that form a path starting from $S'$ consisting of $c$ tiles such that each successive tile in the path is $\vec{b}$-shifted relative to the previous tile in the path, and that no portion of any tile belonging to this path lies below the south edge of $R$. These tiles are schematically depicted as the blue rectangles in Subfigure~\ref{subfig:simpleShiftCase2}. Let $A_c$ denote the final tile of this path. Then, the $p$ summand ensures that we can define a tile $T_1$ that can attach to the left of $A_c$ that is on grid with $A_c$. Finally, we can now define tiles that allow for the formation of a vertical path of tiles $T_2, \ldots, T_k$ starting from $T_1$ and ending at a tile $T_k$ ($=A$ in the lemma statement).
By defining sufficiently many such tiles, we can ensure that this vertical path places $A$ such that Property 3 holds.
These tiles are schematically depicted as the red rectangles in Subfigure~\ref{subfig:simpleShiftCase2}.
Once again, we've defined our tiles such that Properties 1 and 2 hold.

Cases (c) and (d) are similar to Case (b). The key is to first define tiles for $\mathcal{T}$ that form a path of tiles ending with a tile $S'$ such that
(1) no portion of any tile of this path lies to the east (respectively south) of the line defined by the east (respectively south) edge of the bounding rectangle for $S$, (2) we can define tiles that form a path $\{A_i\}_{i=1}^c$ such that $A_1$ attaches to $S'$ and each successive tile in the path is $\vec{b}$-shifted relative to the previous tile in the path, and (3) a third and final path of tiles $\{T_i\}_{i=1}^k$ can be defined such that $T_1$ attaches to $A_c$, each $T_i$ is on grid with $A_c$ (and therefore $c\vec{b}$-shifted relative to $S$), and $T_k$ is the northernmost and westernmost tile of the assembly consisting of all of the tiles three aforementioned paths.

\end{proof}

As previously mentioned, the properties given in Lemma~\ref{lem:simpleShiftCases} ensure that if we have two such systems, $\mathcal{T}_{\vec{r}_1}$ and $\mathcal{T}_{\vec{r}_2}$ corresponding to two shift vectors $\vec{r}_1$ and $\vec{r}_2$, $\mathcal{T}_{\vec{r}_1}$ and $\mathcal{T}_{\vec{r}_2}$ can be ``concatenated'' to give a system that contains an $(\vec{r}_1 + \vec{r}_2)$-shifted polyomino tile. The following lemma formalizes this notion of ``concatenation''.

\begin{lemma}\label{lem:genBitWrite}
Let $\vec{r} = \Sigma_{i=0}^n c_i\vec{b_i}$ for any $c_i \in \mathbb{Z}$ and $\vec{b_i} \in B$. Then there exists some directed, singly seeded system $\mathcal{T} = (T,\sigma)$ with all tiles shaped $P$ such that the terminal assembly $\alpha$ of $\mathcal{T}$ contains an $\vec{r}$-shifted polyomino. Moreover, the system $\mathcal{T}$ and assembly $\alpha$ have the following properties.
\begin{enumerate}
\item There is a single assembly sequence that yields $\alpha$,
\item For some $m,n\in \N^2$, $\alpha$ is contained in an $m\times n$ rectangle $R$ and the seed tile $S$ lies in the southeast corner of $R$, and
\item the last polyomino tile, $A$, to attach in the system lies in the northwest corner of $R$.
\end{enumerate}
\end{lemma}

\begin{proof}
This follows by applying Lemma~\ref{lem:simpleShiftCases} to each of the summands of $\vec{r} = \Sigma c_i\vec{b_i}$.
\end{proof}

In Lemma~\ref{lem:genBitWrite}, we start with a seed tile in the southeast corner of a rectangular region $R$ and proceed to place an $r$-shifted tile in the northwest corner of the rectangle. Note that by using the techniques used to prove Lemma~\ref{lem:simpleShiftCases} and Lemma~\ref{lem:genBitWrite}, we can show analogous lemmas where the seed tile lies in any corner of $R$ and an $r$-shifted tile lies in the opposite corner.

Now we prove the main lemma for this section.  This lemma will allow us to construct bit-writing gadgets used in the construction given in Section~\ref{sec:main-result}. Intuitively, the lemma states that given a polyomino $P$ and vector $\vec{r}\in V$, we can define a polyTAM system that starts growth from a seed tile, $S$, in the southeast corner of a rectangle $R$, and without growing outside of $R$, places $\vec{r}$-shifted tiles on the west edge of $R$ in such a way that it is possible to then continue growth to the west of $R$. The possibility of continuing growth to the west of $R$ is formally stated as Property 4 in Lemma~\ref{lem:r-shifts}. This lemma will allow us to assemble a series of bit-writers while also resting assured that once these bit-writers have assembled, bit-reader assemblies can continue growth. It is helpful to see Figure~\ref{fig:r-shifts} for an overview of the properties given in the following lemma.

\begin{figure}[htp]
\begin{center}
\includegraphics[width=3.0in]{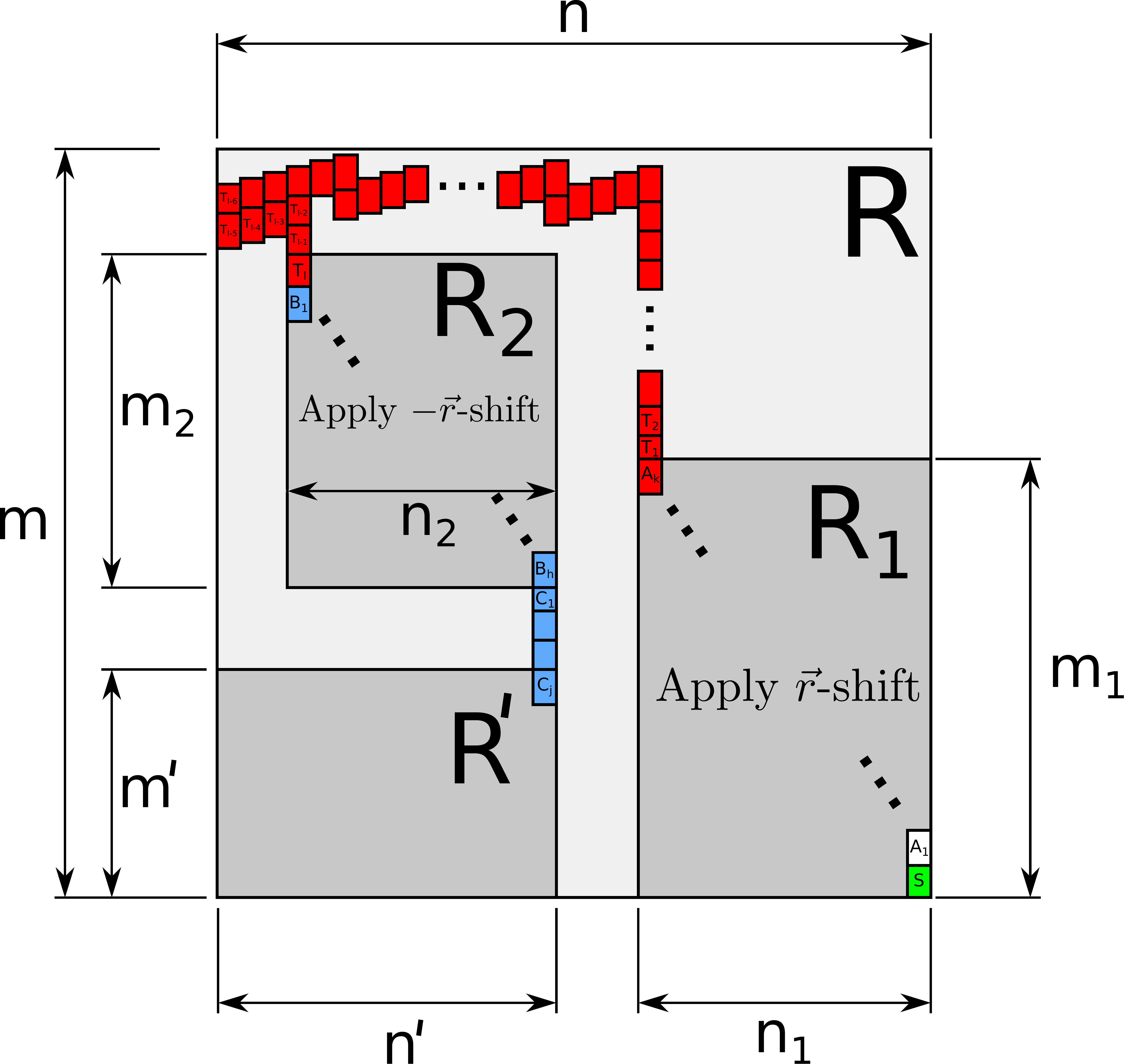}
\caption{The rectangular region $R_1$ contains tiles of $\alpha$ defined using Lemma~\ref{lem:genBitWrite} with vector $\vec{r}$. The red rectangular regions each serve as the bounding rectangle of an $\vec{r}$-shifted tile of $\alpha$. This path of tiles, $\{T_i\}_{i=1}^l$, places a final tile in the northwest corner of the region $R_2$. $R_2$ contains tiles of $\alpha$ defined using Lemma~\ref{lem:genBitWrite} with vector $-\vec{r}$-shifted. The blue rectangular regions each contain a single tile that is on grid with $S$. This path of tiles places a final tile in the northeast corner of the region $R'$. Note that by modifying the path of tiles $\{T_i\}_{i=1}^l$, we can make the dimension $m'$ and $n'$ of $R'$ arbitrarily large.}
\label{fig:r-shifts}
\end{center}
\end{figure}

\begin{lemma} \label{lem:r-shifts}
Let $P$ be a polyomino.  Let $\vec{r}$ be a vector in $V$.  Then there exists a directed, singly seeded system $\mathcal{T}=(T, \sigma)$ with all tiles of shape $P$ which produces $\alpha$ such that $\mathcal{T}$ and $\alpha$ have the following properties.
\begin{enumerate}
\item There is a single assembly sequence that yields $\alpha$,
\item for some $m,n\in \N^2$, $\alpha$ is contained in an $m\times n$ rectangle $R$ and the seed tile $S$ lies in the southeast corner of $R$,

\item for any tile $A$ of $\alpha$ such that the west edge of the bounding rectangle of $A$ lies on the west edges of $R$, $A$ is $\vec{r}$-shifted, and

\item for any $m'', n''\in \N^2$, we can choose $\mathcal{T}$ such that the last tile $L$ to attach to $\alpha$ lies on gird in the northeast corner of a rectangle $R'$ with dimensions $m'\times n'$ where $m' > m''$ and $n' > n''$.  Moreover, the south and west edges of $R'$ lie on the south and west edges of $R$, and no portion of any tile of $\alpha$ lies inside of $R'$ and outside of the bounding rectangle of $L$.
\end{enumerate}
\end{lemma}

\begin{proof}
We will define the tiles of $\mathcal{T}$ so that each property holds. First, let $\vec{v}$ and $\vec{w}$ be vectors given by Lemma~\ref{lem:grids}. It is with respect to these vectors that we can say whether or not a tile is ``on grid''. Now let $\mathcal{T}_1$ be the system given by Lemma~\ref{lem:genBitWrite} for the vector $\vec{r}$, and let the rectangular $m_1\times n_1$ region given by Lemma~\ref{lem:genBitWrite} be denoted by $R_1$. This region is depicted in Figure~\ref{fig:r-shifts}, and is where the $\vec{r}$-shift is applied. By defining tiles of $\mathcal{T}$ to be the same as tiles of $\mathcal{T}_1$ we obtain a seed $S$ and tiles $\{A_i\}_{i=1}^k$ such that $A_k$ is $r$-shifted relative to $S$, and lies in the northwest corner of $R_1$.

Now, we can use the vectors $\vec{v}$ and $\vec{w}$ to define a path of tiles $\{T_i\}_{i=1}^l$ such that $T_1$ attaches to $A_k$ (after adding an appropriate glue to the definition of $A_k$), each $T_i$ attaches to $T_{i-1}$ for $i>1$ and is on grid with $A_k$ (and hence remain $r$-shifted relative to $S$), and the last tile, $T_l$, of this path lies in the northwest corner of an $m_2\times n_2$ rectangular region $R_2$.
In Figure~\ref{fig:r-shifts}, each $T_i$ lies in a red rectangular regions; $R_2$ is also depicted in Figure~\ref{fig:r-shifts}.
Note that we can choose the tiles $\{T_i\}_{i=1}^l$ to be such that the west edges of exactly two tiles in $\{T_i\}_{i=1}^l$ lie on the west edge of $R$.
Moreover, we can choose this path of tiles $\{T_i\}_{i=1}^l$ such that if $\alpha_l$ is the subassembly of $\alpha$ consisting of all of the tiles $S$, $\{A_i\}_{i=1}^k$, and $\{T_i\}_{i=1}^l$, then no portion of any tile of $\alpha_l$ is contained inside of $R_2$ and outside of the bounding rectangle of $T_l$.

Notice that $m_2$ and $n_2$ can be made arbitrarily large by extending the path consisting of the tiles $\{T_i\}_{i=1}^l$. Then, for $m_2$ and $n_2$ sufficiently large, we can apply Lemma~\ref{lem:genBitWrite} to obtain a directed, singly seeded system $\mathcal{T}_2$ with a seed tile $S_2$ placed in the northwest corner of $R_2$ such that the terminal assembly of $\mathcal{T}_2$ contains a $-\vec{r}$-shifted tile in the southeast corner. Using $\mathcal{T}_2$, we can define tiles, $B_1$ through $B_h$, for $\mathcal{T}$ such that $B_1$ can attach to $T_l$, and $B_h$ is $-\vec{r}$-shifted relative to $T_l$, and lies in the southeast corner of $R_2$. Therefore, $B_h$ is on grid with $S$. That is, the locations of $B_h$ are the locations of $S$, only shifted by $c_1\vec{v} + c_2\vec{w}$ for some $c_1, c_2\in \Z$. Finally, note that we can choose the path consisting of tiles $\{T_i\}_{i=1}^l$ so that $|c_1|$ and $|c_2|$ are arbitrarily large. Then, using vertical shifts by $\vec{w}$, we can define tiles for $\mathcal{T}$ that assemble a vertical path of tiles $\{C_i\}_{i=1}^j$ such that $C_1$ can attach directly below and vertically aligned to $B_h$. Furthermore, $C_j$ ($=L$ in the statement of the lemma) is the last tile to attach in $\mathcal{T}$, and $C_j$ lies in the northeast corner of a rectangle $R'$ with dimensions $m'\times n'$ such that the south and west edges of $R'$ lie on the south and west edges of $R$. Again, we can choose the path consisting of tiles $\{T_i\}_{i=1}^l$ such that $m'$ and $n'$ are arbitrarily large, and such that no portion of any other tile of $\alpha$ lies inside $R'$ and outside the bounding rectangle of $C_j$. Hence for an appropriate choice of tiles $\{T_i\}_{i=1}^l$, Property 4 holds.

When defining tiles for $\mathcal{T}$ that assemble each of the aforementioned paths of tiles, we can ensure that Property 1 holds by giving unique glues that allow one and only one tile to attach at any given assembly step. Properties 2, 3, and 4 can be ensured by our choice of tiles $\{T_i\}_{i=1}^l$.

\end{proof}

As in Lemma~\ref{lem:genBitWrite}, in Lemma~\ref{lem:r-shifts}, we start with a seed tile in the southeast corner of a rectangular region $R$ and proceed to place $r$-shifted tiles on the west edge of the rectangle. Note that we can show analogous lemmas where the seed tile starts in any corner of a rectangle $R$, and $r$-shifted tiles are placed on a chosen opposite edge.

\section{All Polyominoes of Size at Least 3 Can Perform \\ Universal Computation at $\tau=1$}\label{sec:main-result}

We can now proceed to state our main result:
any polyomino $P$ of size at least three can be used for polyomino tile-assembly systems that are
computationally universal at temperature $1$.  Formally stated:

\begin{theorem}\label{thm:comp-univ-poly}
Let $P$ be a polyomino such that $|P| \geq 3$.  Then for every standard Turing
Machine $M$ and input $w$, there exists a TAS with $\tau=1$ consisting only of
tiles of shape $P$ that simulates $M$ on $w$.  \end{theorem}

It follows from the procedure outlined in Section~\ref{sec:Turing-overview} and the Lemmas of Section~\ref{sec:tech-lemmas} that in order to simulate an arbitrary Turing Machine by a
TAS consisting only of tiles of some polyomino shape $P$, it is sufficient to construct a
system consisting only of tiles of shape $P$ for which there exists a
bit-reading gadget, because the additional paths required for a zig-zag Turing
machine simulation are guaranteed to be producible by the lemmas of
Section~\ref{sec:tech-lemmas}.

To simplify our proof, we consider different categories of shapes of $P$ as separate cases, which first requires an additional definition.

\begin{definition}[Basic polyomino]
A polyomino $P$ is said to be a \emph{basic} polyomino if and only if for
every vector $\vec{x}$ modulo the polyomino grid for $P$, there exists a system
$\mathcal{T}$ containing only tiles with shape $P$ such that $\mathcal{T}$
produces $\alpha$ and $\alpha$ contains a $\vec{x}$-shifted polyomino.
Otherwise we call $P$ {\em non-basic}.  \end{definition}

Essentially, basic polyominoes are those which have the potential to grow paths that place tiles at any and all shift vectors relative to the grid.

Our proof consists of showing how to build bit-reader gadgets for each of the following cases based on the shape $P$:

\begin{enumerate}
    \item[{\bf (1)}]  $P$ has thickness 1 in one direction, i.e., it is an $m\times 1$ polyomino.
    \item[{\bf (2)}]  $P$ has thickness 2 in two directions, i.e., $d_x=d_y=2$.
    \item[{\bf (3)}]  $P$ is basic and has thickness at least 3 in one and at least 2 in the other direction.
    \item[{\bf (4)}]  $P$ is non-basic.
\end{enumerate}

The Lemmas of Section~\ref{sec:tech-lemmas} provide us with the basic facilities to build paths of tiles which occupy particular points while avoiding others.  By carefully designing the grids and offsets for the tiles of each polyomino $P$, we are able to construct the constituent paths of the bit-reading gadgets.  

Let $P$ be an arbitrary polyomino, with $|P|\geq 3$. Without loss of generality (as the following arguments all hold up to rotation), let the bounding
box of $P$ be of dimensions $m\times n$, with $m\geq n$, and $d_y$ be the
largest distance between two pixels of $P$ in the same column, and $d_x$
be the largest distance between two pixels of $P$ in the same row.
For ease of notation, we refer to the southernmost of all westernmost pixels of $P$ as $p_0=(0,0)$, and to all other pixels by their integer coordinates.

For any polyomino $P$, we know that tiles of
shape $P$ can produce a grid by Lemma~\ref{lem:grids}; throughout this
section, we simply refer to this as the grid (for $P$).  We also note that the
grid for a given $P$ may be slanted as in Figure~\ref{fig:polyomino-grid}, and
that the construction of the zig-zag Turing machine is simply slanted
accordingly.  If we say that a tile is $\vec{v}$-shifted for some vector
$\vec{v}$, we mean that it is off grid by the vector $\vec{v}$.

For all figures in this section, we use the same color conventions as in
Figure~\ref{fig:bit-gadget-definition}.  Thus, the green tiles in this
section represent the $t$ tile; as discussed in the caption of the figure, the
yellow and aqua tiles represent the two potential paths grown from $t$, while the
dark grey tiles represent tiles that prevent the growth of paths from $t$.  We
refer to these grey tiles as \emph{blockers}.  We use the convention that if a
path of yellow polyominoes grows, then a 0 is being read.  Similarly, if a path
of aqua polyominoes grows, then a 1 is being read.  Consequently, we call
polyominoes that prevent the growth of the aqua path {\em $1$-blockers} and
polyominoes that prevent the growth of the yellow path {\em $0$-blockers}.  In
addition, tiles of the same color are numbered in order to indicate the order
of their placement where the higher numbered tiles are placed later in the
assembly sequence.

Our proof consists of showing how to build bit-reader gadgets for each of the following cases based on the shape $P$:

\begin{enumerate}
    \item[{\bf (1)}]  $P$ has thickness 1 in one direction, i.e., it is an $m\times 1$ polyomino.
    \item[{\bf (2)}]  $P$ has thickness 2 in two directions, i.e., $d_x=d_y=2$.
    \item[{\bf (3)}]  $P$ is basic and has thickness at least 3 in one and at least 2 in the other direction.
    \item[{\bf (4)}]  $P$ is non-basic.
\end{enumerate}

\subsection{Case (1):  $P$ Is an $m\times 1$ Polyomino}
If $P$ is a straight line, and therefore $n=1$, we can simply use a scheme as illustrated in Figure~\ref{fig:comp_1by_case}.
In Figure~\ref{fig:comp_1by_case0}, a 0 bit is read as indicated by the
placement of the yellow polyomino.  Notice that the tile labeled 3 in
Figure~\ref{fig:comp_1by_case1} prevents the attachment of the aqua-colored
tile.  After the yellow tile attaches, a fuchsia tile attaches as shown in the
figure which allows for growth to continue.  Figure~\ref{fig:comp_1by_case1}
shows a similar scenario in the case that a 1 is read.  The key property of this
bit reader is that the yellow and aqua tiles have different offsets relative to the green tile, which is always possible if $P$ is a line of length $\ge 3$.
Note that the bit reader shown in Figure~\ref{fig:comp_1by_case} is a left-to-right bit-reading
gadget.  A right-to-left bit-reading gadget can be constructed in a similar
fashion.

\begin{figure}[htp]
\centering
  \subfloat[][A left-to-right bit-reading gadget reading a 0 bit.]{%
        \label{fig:comp_1by_case0}%
        \includegraphics[width=2.3in]{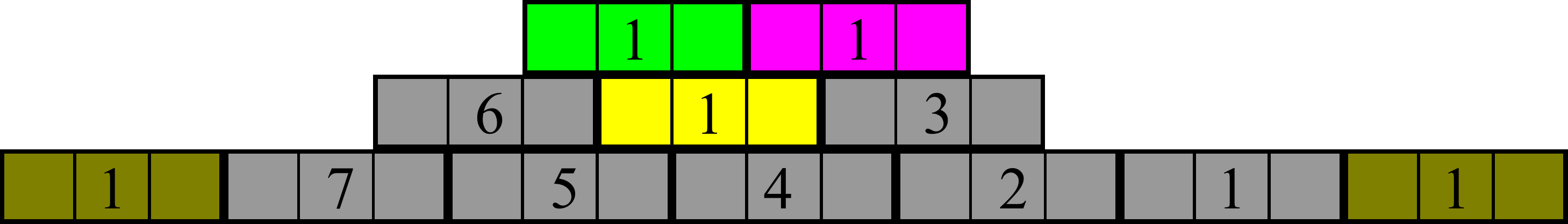}
        }%
        \quad
  \subfloat[][A left-to-right bit-reading gadget reading a 1 bit.]{%
        \label{fig:comp_1by_case1}%
        \includegraphics[width=2.3in]{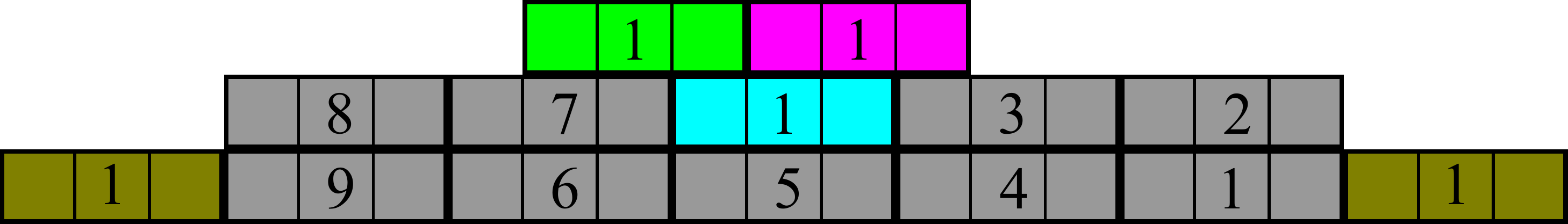}
        }%
        \quad
  \caption{The two different bit-reading schemes for an $m\times 1$ polyomino. Note that the bit reader in this figure proceeds from left to right.}
  \label{fig:comp_1by_case}
\end{figure}

The case for any $P$ which is an $n \times 1$ straight line can be handled in the same way.  Thus, we now only need to consider the cases $m\geq n\geq 2$. Because $P$ is connected,
this implies both $d_x\geq 2$ and $d_y\geq 2$.
Furthermore, we assume the that the grid constructed from
Lemma~\ref{lem:grids} using $P$ is created by attaching the southernmost pixel on the
eastern edge of $P$ to the northernmost pixel on the western side of the
$P$, as suggested by Figure~\ref{fig:polyomino-grid}.

\subsection{Case (2):  $P$ Is Such That $d_x=d_y=2$}
Before describing bit-reading constructions, we analyze the possible cases
for the shape of $P$; refer to Figure~\ref{tbl:SpecCompCs}.  Also, we note that $d_x=d_y=2$ implies that $P$ is basic.

First consider the situation in which $|P|$ is even.
If both $(1,0)$ and $(0,1)$ belong to $P$, the assumption $d_x=d_y=2$ implies
that $(1,1)$ must also belong to $P$, but no further pixels. Thus, $P$
is a $2\times 2$-square, which will be treated as {\bf Case (2a)}.
Now, without loss of generality consider the case that $(0,1)$ belongs to $P$, but $(1,0)$
does not. It follows from $d_x=d_y=2$ that $(1,1)$ belongs to $P$, as well as $(1,2)$.
This conclusion can be repeated until all pixels of $P$ are allocated. It follows
that $P$ is an even zig-zagging shape. This is shown as {\bf Case (2b)} in Figure~\ref{tbl:SpecCompCs}.

Now consider the case in which $|P|$ is odd.
If both $(1,0)$ and $(0,1)$ belong to $P$, $(1,1)$ cannot be part of $P$, and $P$ is an $L$-shape
consisting of three pixels. If without loss of generality $(0,1)$ belongs to $P$, but $(1,0)$
does not, we can conclude analogous to (2a) that
$P$ is an odd zig-zagging shape, shown as {\bf Case (2c)} in Figure~\ref{tbl:SpecCompCs};
this also comprises the case of an $L$-shape with three pixels.

\newcolumntype{M}{>{\centering\arraybackslash}m{\dimexpr.25\linewidth-2\tabcolsep}}
\begin{figure}[htp]
\centering
\begin{tabular}{| M | M | M|}
	\hline
	\includegraphics[scale=.13]{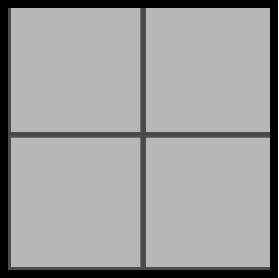}
     & \includegraphics[scale=.13]{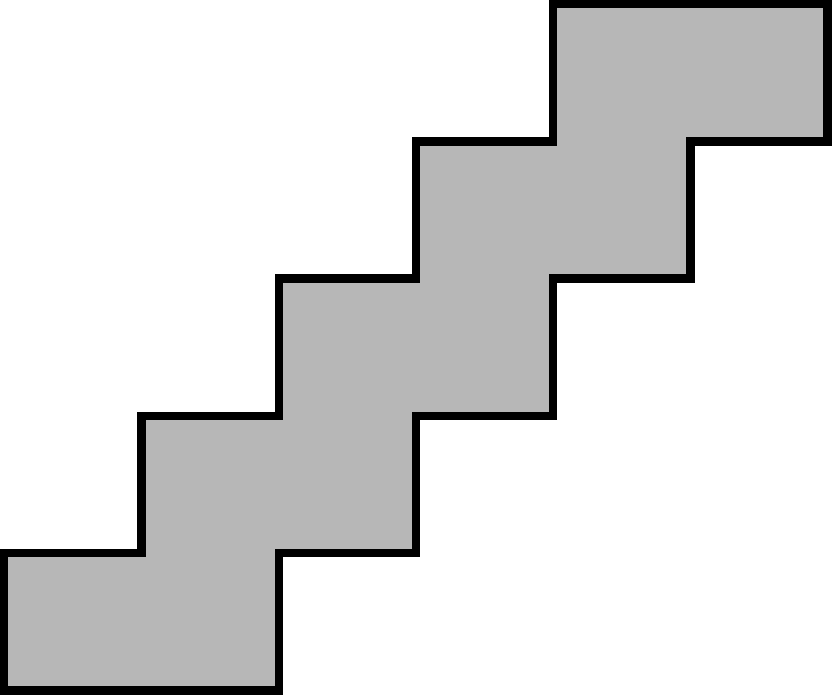}
     & \includegraphics[scale=.13]{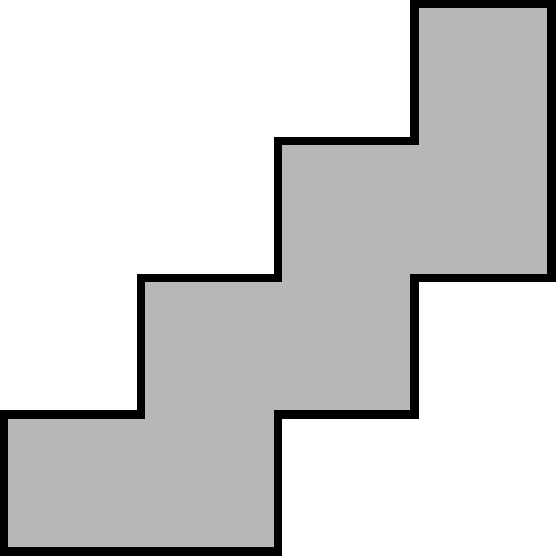}\\
	(2a) & (2b) & (2c)  \\\hline
\end{tabular}
\caption{The possible shapes in Case (A), when $d_x=d_y=2$.}
\label{tbl:SpecCompCs}
\end{figure}

Now we sketch the bit-reading schemes. As these cases are relatively straightforward, we simply refer
to the corresponding figures. Note that the logic of the arrangement is color coded:
the first polyomino we add to our tile set is the green polyomino along with the aqua and yellow polyominoes
that allow for an aqua tile to attach to the east of it in an on grid position
and a yellow tile to attach to the east of it shifted $(-1,-1)$ relative to the
polyomino grid.

\subsubsection{Case (2a)}
If $P$ is a $2 \times 2$ square we use the scheme shown in Figure~\ref{fig:SpecCompSQ}.

\begin{figure}[htp]
\centering
  \subfloat[][A left-to-right bit reader reading a ``0'' bit.]{%
        \label{fig:SpecCompSQ0}%
        \includegraphics[width=2.3in]{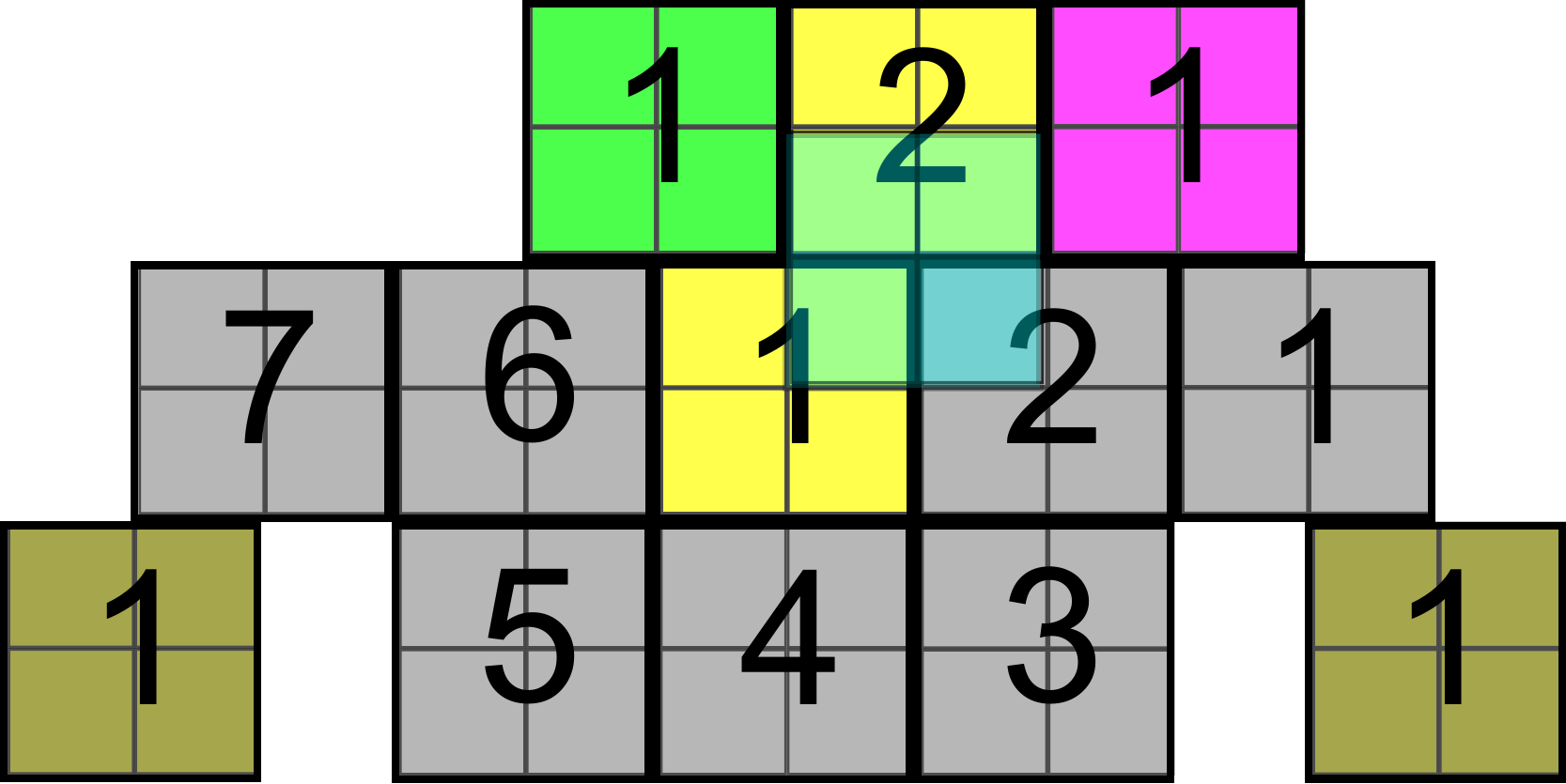}
        }%
        \quad
  \subfloat[][A left-to-right bit reader reading a ``1'' bit.]{%
        \label{fig:SpecCompSQ1}%
        \includegraphics[width=2.3in]{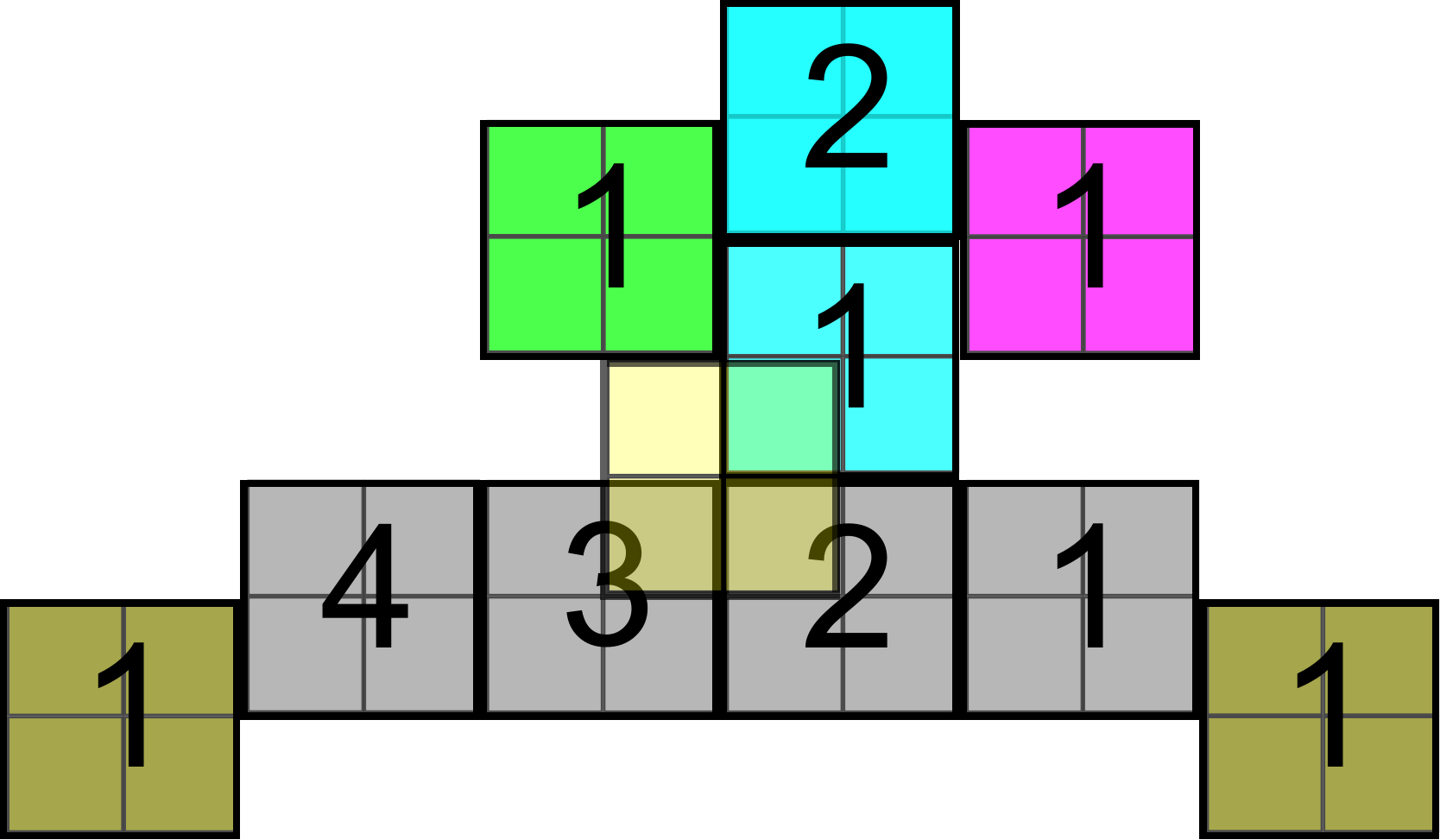}
        }%
        \quad
  \caption{A general bit-reading scheme for a left-to-right bit reader in case (A1), in which $P$ is a $2\times2$ square.}
  \label{fig:SpecCompSQ}
\end{figure}

\subsubsection{Case (2b)}
If $P$ is an even zig-zagging polyomino, as shown in part (A2) of Figure~\ref{tbl:SpecCompCs}, we use the bit-reading schemes shown in Figure~\ref{tbl:SpecCompS}.

\begin{figure}[htp]
\centering
\begin{tabular}{| M | M |}
	\hline
	\includegraphics[scale=.11]{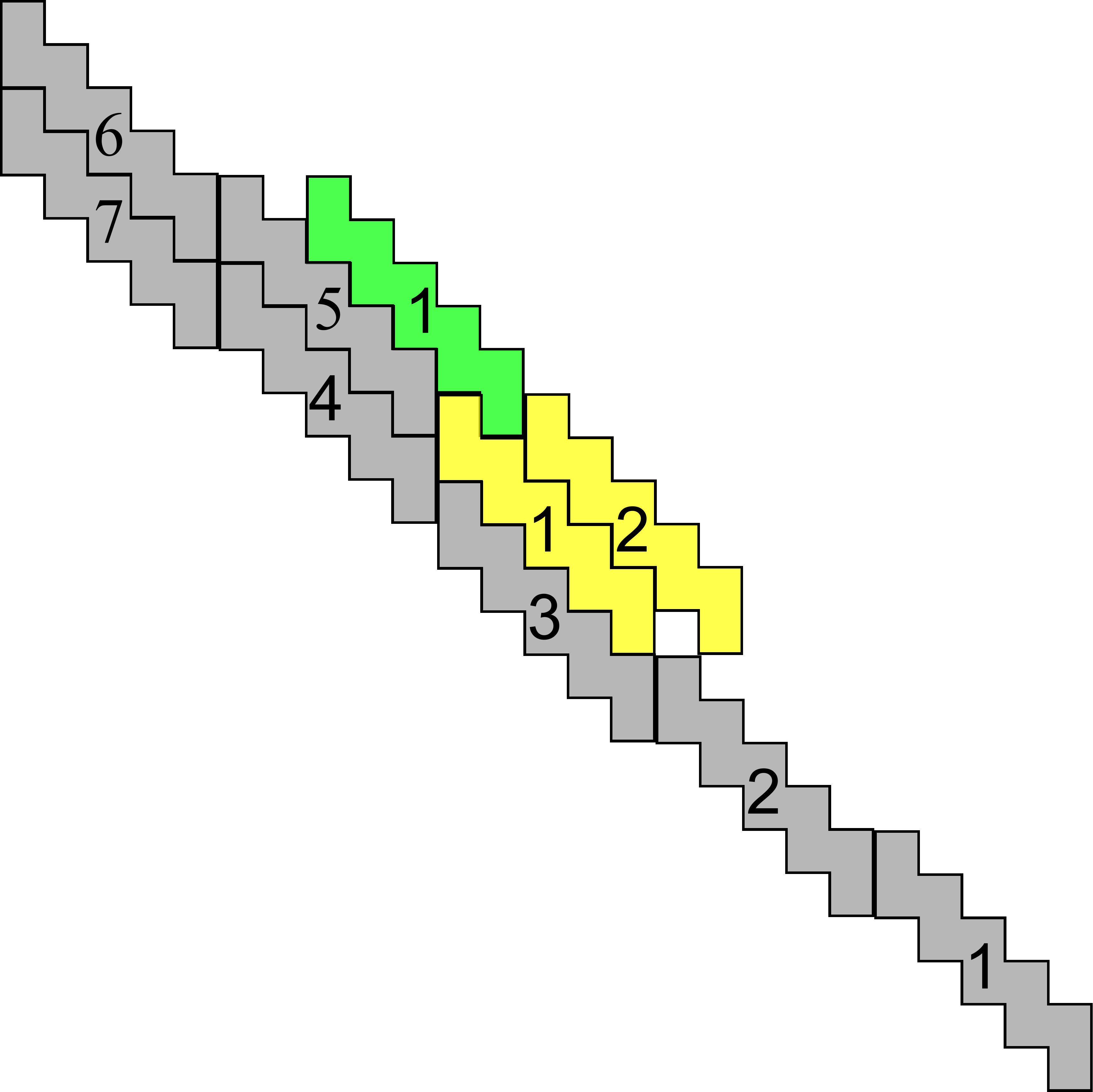}  & \includegraphics[scale=.11]{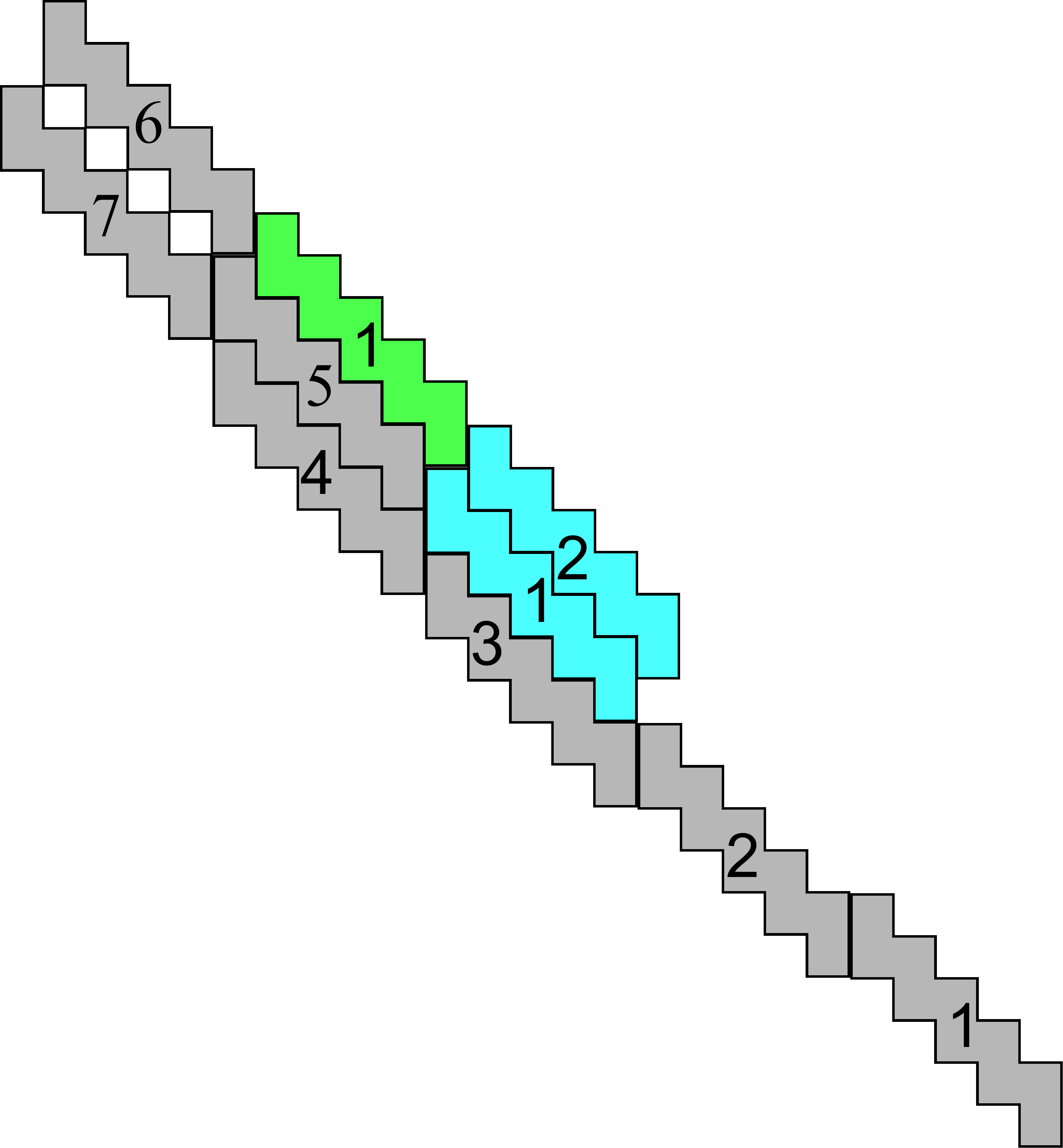} \\
	(a) & (b)  \\\hline
	\includegraphics[scale=.11]{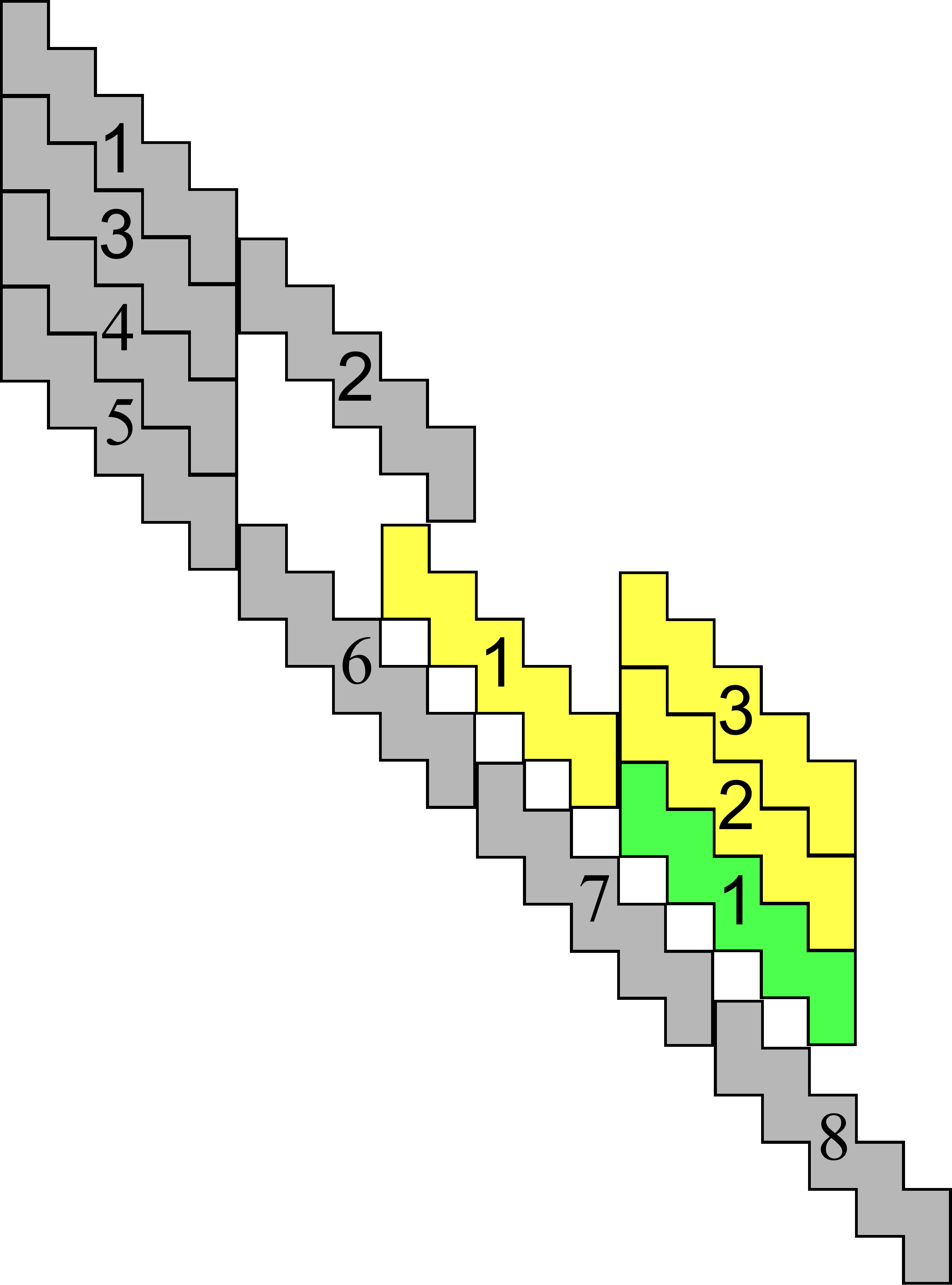}  & \includegraphics[scale=.11]{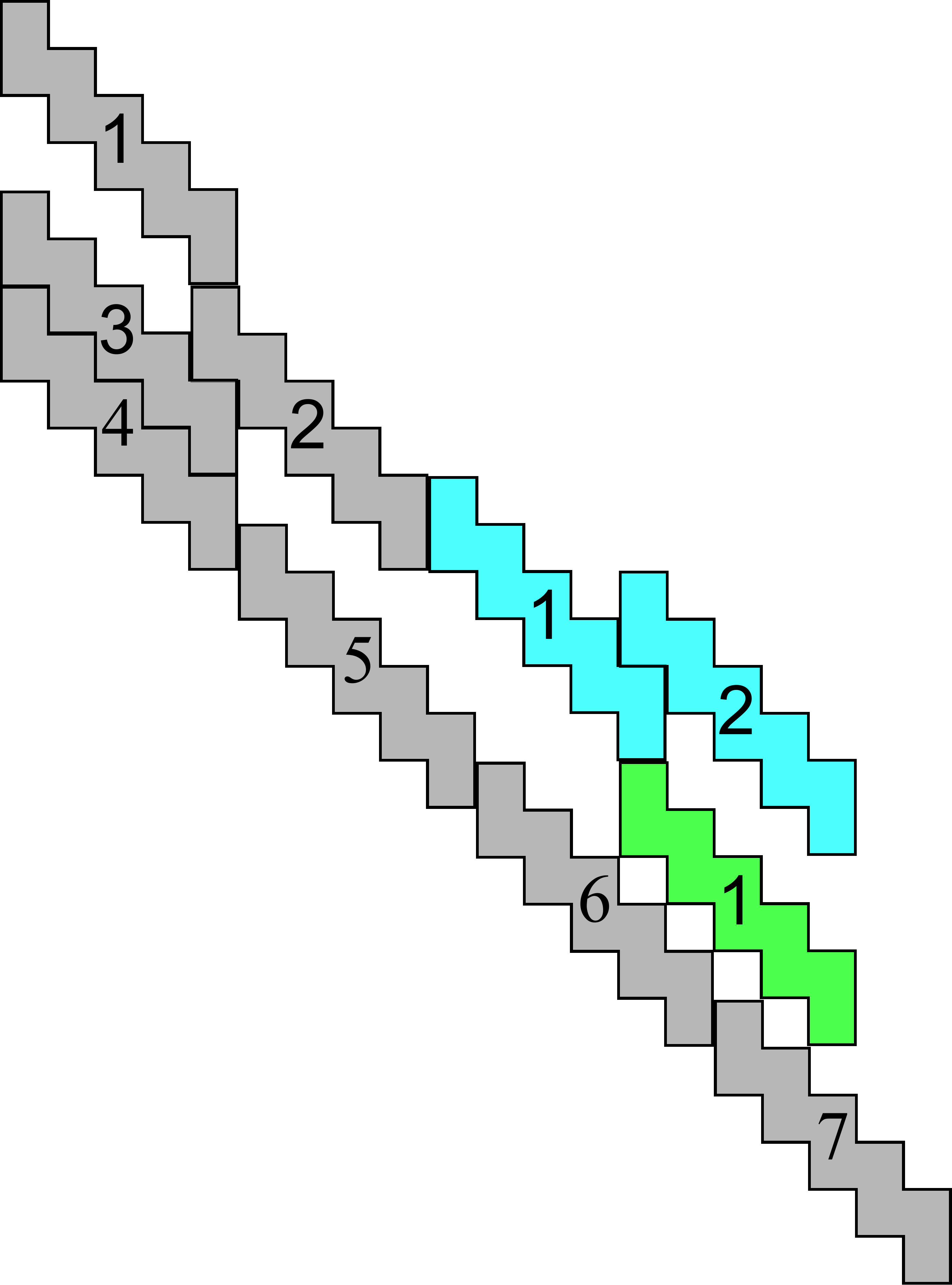} \\
	(c) & (d) \\\hline
\end{tabular}
\caption{The bit-reading schemes for Case (2b) of Figure~\ref{tbl:SpecCompCs}.}
\label{tbl:SpecCompS}
\end{figure}

\subsubsection{Case (2c)}
If $P$ is an odd zig-zagging polyomino, as shown in part (A3) of Figure~\ref{tbl:SpecCompCs}, we use the bit-reading schemes shown in Figure~\ref{tbl:SpecCompD}.

This concludes Case (2).

\begin{figure}[htp]
\centering
\begin{tabular}{| M | M |}
	\hline
	\includegraphics[scale=.13]{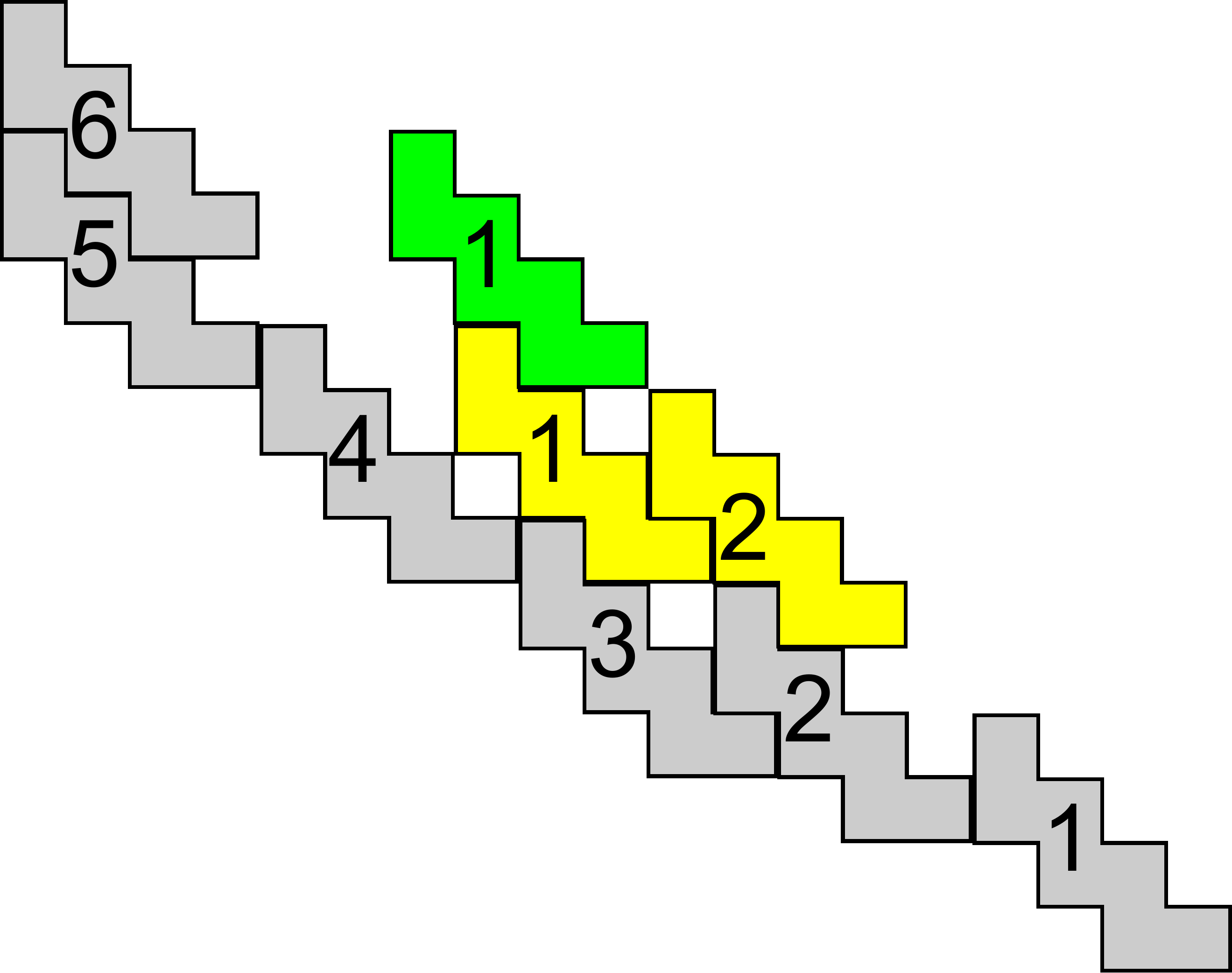}  & \includegraphics[scale=.13]{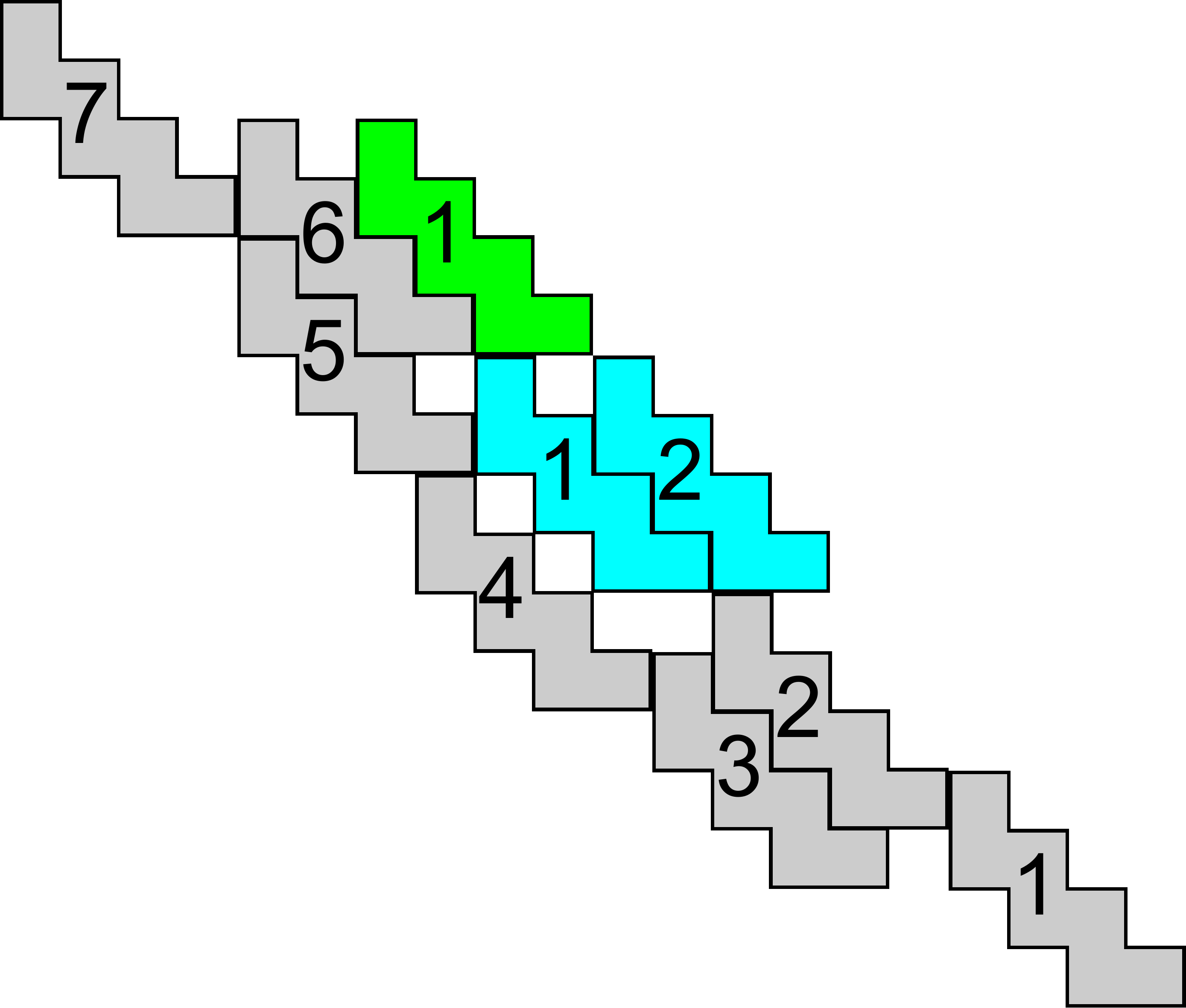} \\
	(a) & (b)  \\\hline
	\includegraphics[scale=.13]{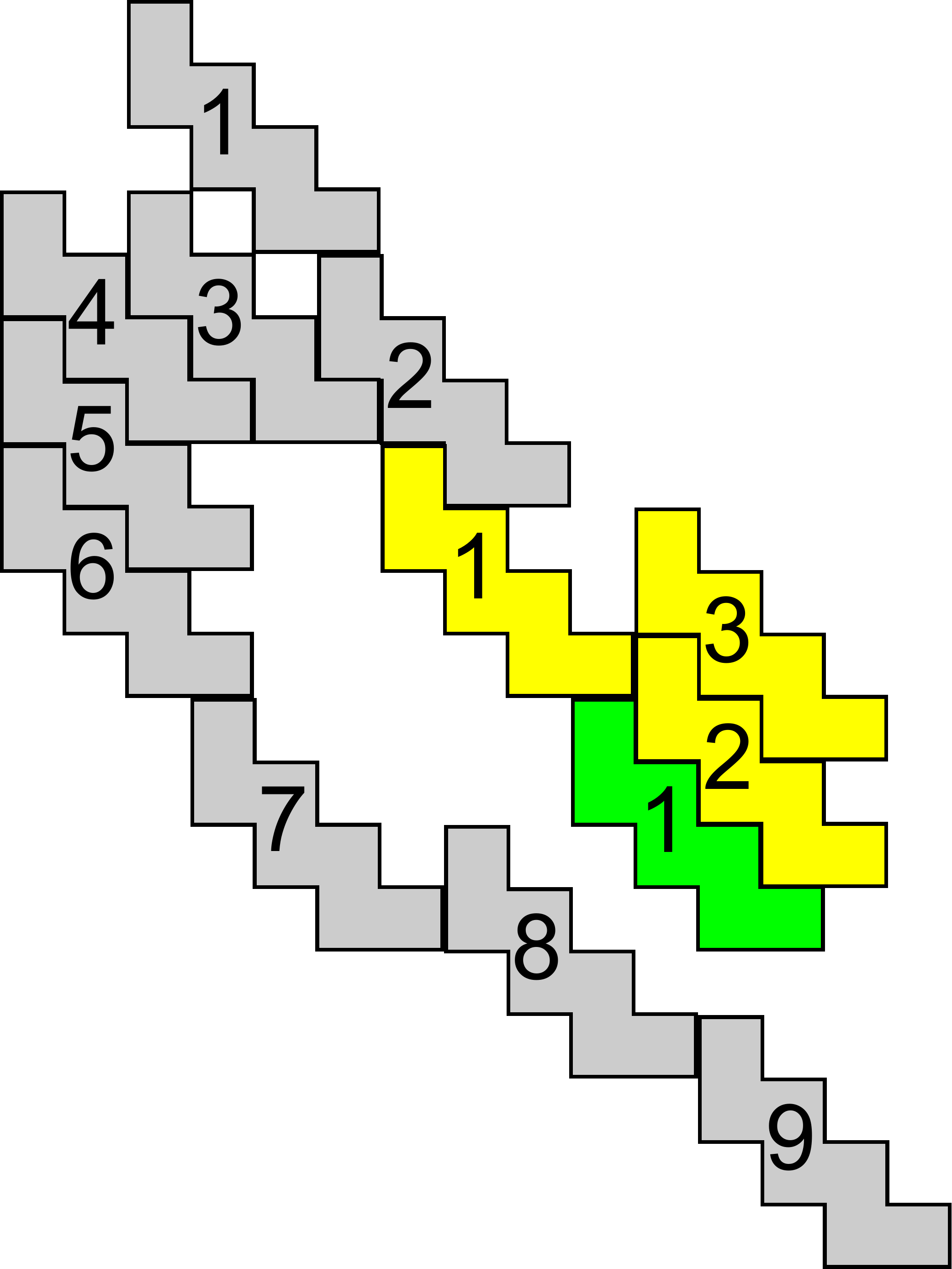}  & \includegraphics[scale=.13]{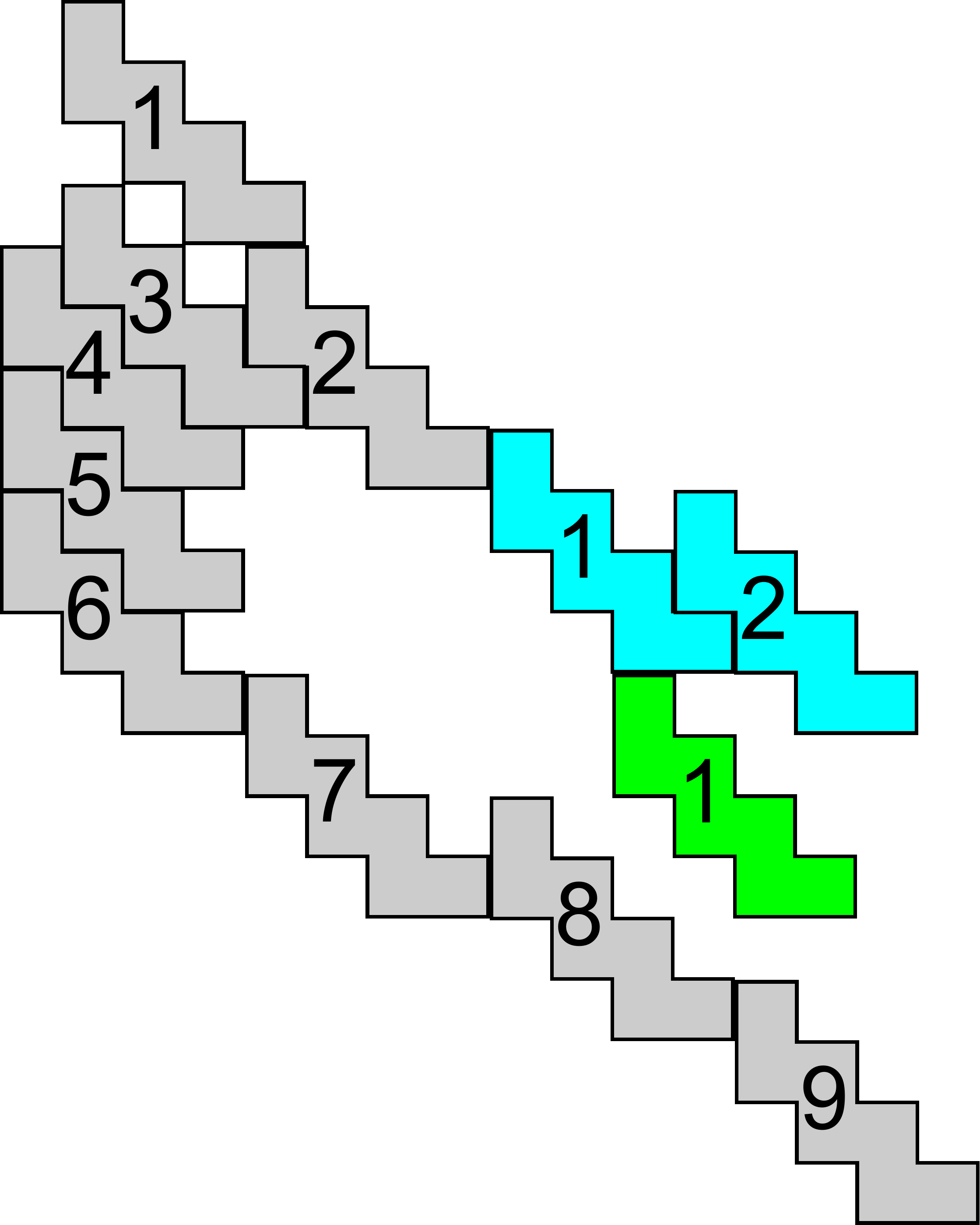} \\
	(c) & (d) \\\hline
\end{tabular}
\caption{The bit-reading schemes for Case (2c) of Figure~\ref{tbl:SpecCompCs}.}
\label{tbl:SpecCompD}
\end{figure}

\bigskip
\subsection{Case (3): $P$ Is Basic And Not In Case (1) or (2)}
We now describe how to construct a system that contains a bit-reading gadget in the case that $\max\{d_x,d_y\} \geq 3$ and $P$ is a basic polyomino.  This means that it is possible to construct a path using tiles of shape $P$ which place a tile at any possible offset in relation to the grid.  We will use this ability to place blockers and bit-reader paths exactly where we need them, with those locations specified throughout the description of this case.  Without loss of generality, assume that $\max\{d_x,d_y\}=d_y$.
\subsubsection*{Case (3) Overview} ~\label{sec:3over}
A schematic diagram showing the growth of the bit-reading gadget system we construct is shown in Figure~\ref{fig:GenCompSchematic}.  Note that the figure depicts what the bit-reader would look like if the grid formed by $P$ was a square grid.  In cases of a slanted grid (such as that shown in Figure~\ref{fig:polyomino-grid}), the bit-gadgets would be correspondingly slanted. Growth of the system begins with the seed as shown in Figure~\ref{fig:GenCompSchematic}.  From the seed, the system grows a path of tiles west (shown as a light grey path in the figure) to which one of the two bit writers attach (shown as dark grey in the figure).  Once one of the bit writers assembles, growth proceeds as shown in the schematic view until the other bit writers assemble.  Then growth continues upward to the next level (i.e. the seed row can be considered a ``zig'' row and the next row a ``zag'' row of the zig-zag Turing machine simulation) as shown in the figure until a green tile is placed. Depending on the bit writer gadget to the east of the green tile either a yellow path of tiles grows, indicating that a 0 has been read (as shown in the schematic view with the westernmost bit writer), or an aqua path of tiles grows, indicating that a 1 has been read (as shown in the schematic view with the easternmost bit writer).  Henceforth, we refer to the system described by the schematic view in Figure~\ref{fig:GenCompSchematic} as the bit-reading gadget.

\begin{figure}[htp]
\begin{center}
\includegraphics[width=4.0in]{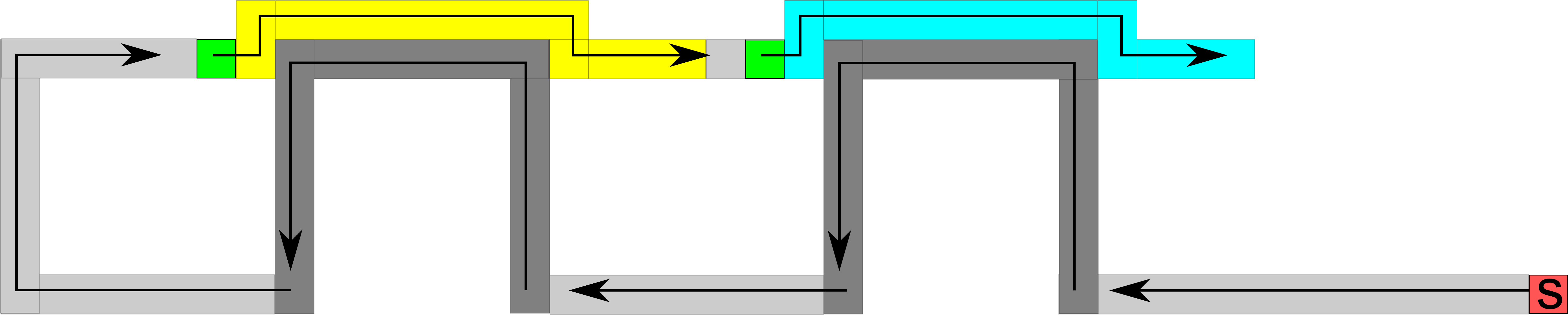}
\caption{A schematic diagram showing the growth of the bit-reading gadget system for Case (3).  Note that in the actual construction, a larger gap would exist between the two bit-readers to allow the path between them to first extend upward and create the necessary bit-writer, then come back down and continue growth of the next bit-reader.}
\label{fig:GenCompSchematic}
\end{center}
\end{figure}

The light grey tiles that compose the bit-reading gadget are easily constructed by placing glues on the polyomino $P$ so that they grow the paths shown in Figure~\ref{fig:GenCompSchematic} (again, modulo the slant of the particular grid formed by $P$), which are on grid with the seed, where the grid is formed following the technique used in the proof of Lemma~\ref{lem:grids}.  The construction of the other tiles is now described.  The green tile is constructed by placing a glue on its western side so that it attaches to the grey tiles on grid as shown in the schematic view.  Furthermore, glues are placed on the green tile and the first aqua tile so that the aqua tile attaches to the green tile in an on-grid manner. Glues are placed on the green tile and the first yellow tile so that the southern edge of the southernmost pixel on the east perimeter of the green tile attaches to the northern edge of the northernmost pixel on the western perimeter of the yellow tile (thus putting the yellow tile off grid).

\subsubsection*{Case (3) Bit-Writer Construction} \label{sec:3BW}

\begin{figure}[ht!]
\centering
\begin{tabular}{| M | M |}
	\hline
	\includegraphics[scale=.07]{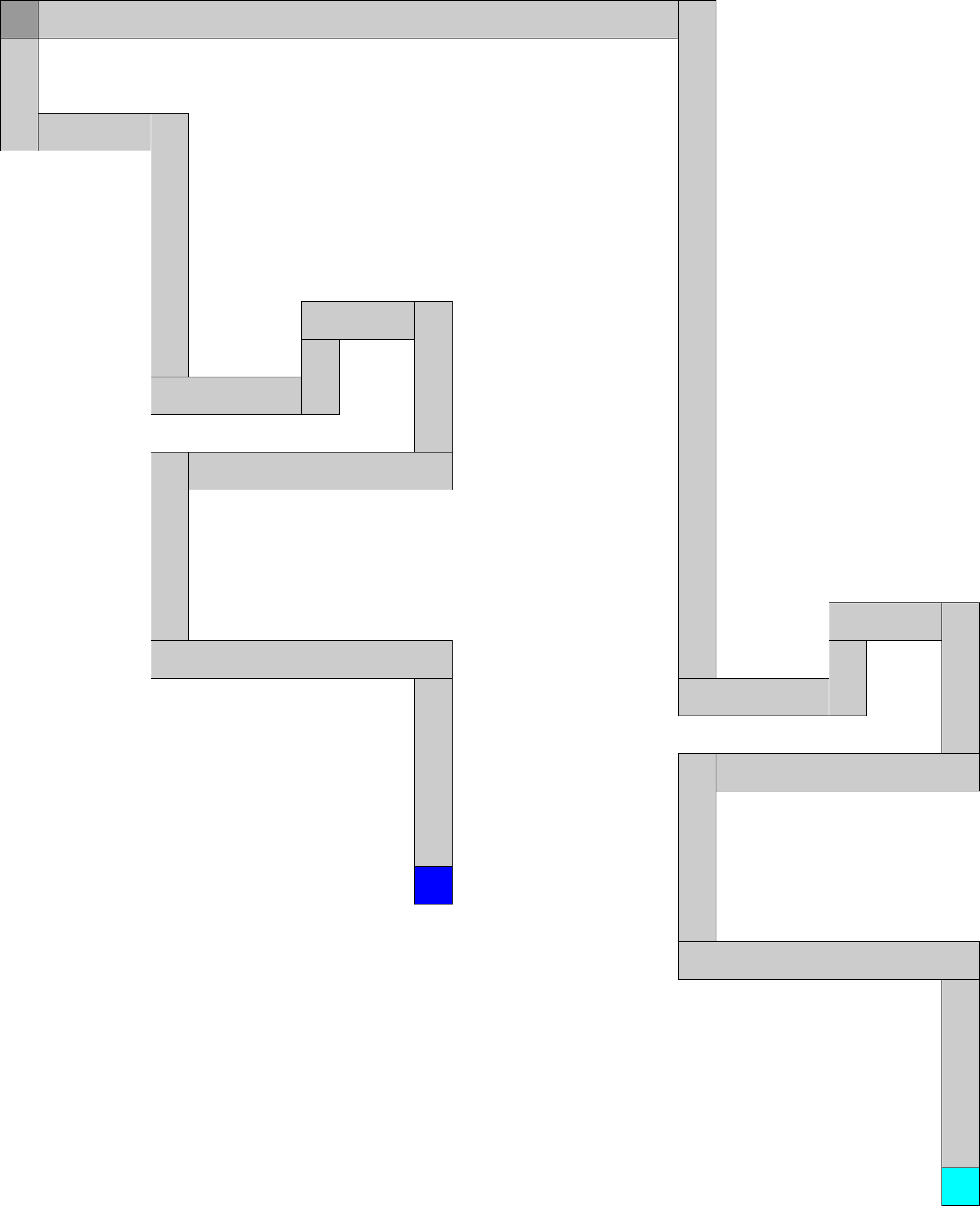}  & \includegraphics[scale=.07]{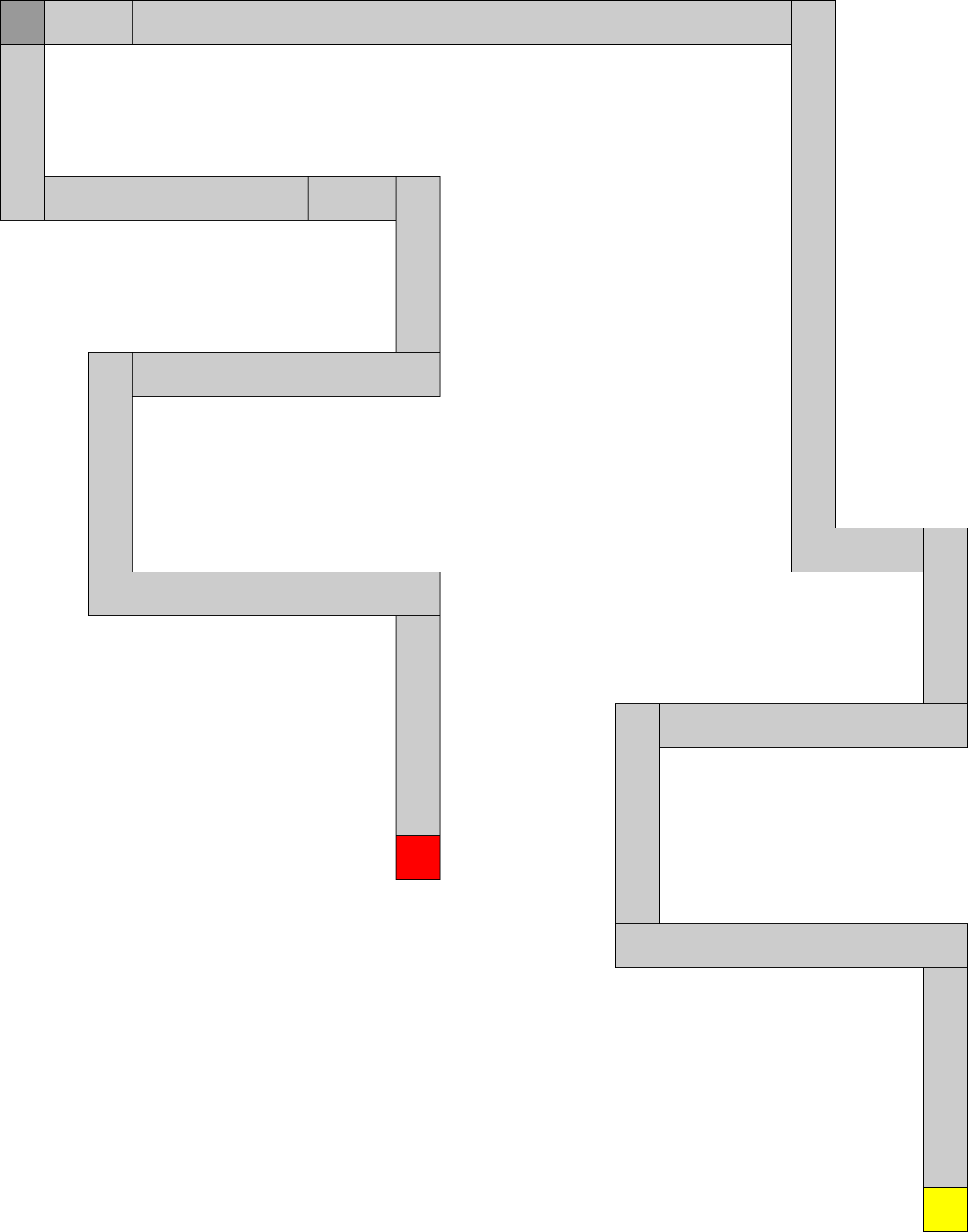} \\
	(a) & (b)  \\\hline
	\includegraphics[scale=.07]{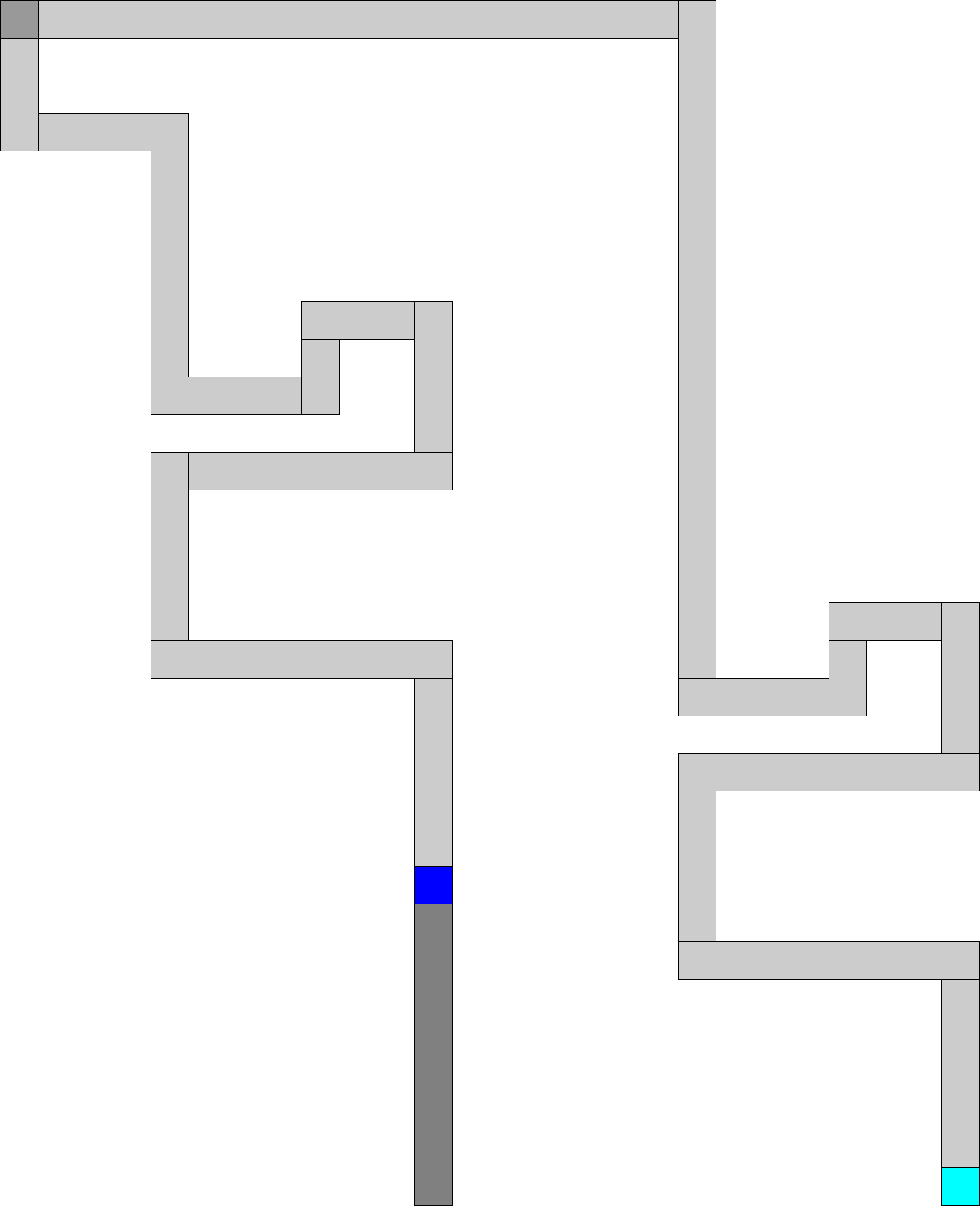}  & \includegraphics[scale=.07]{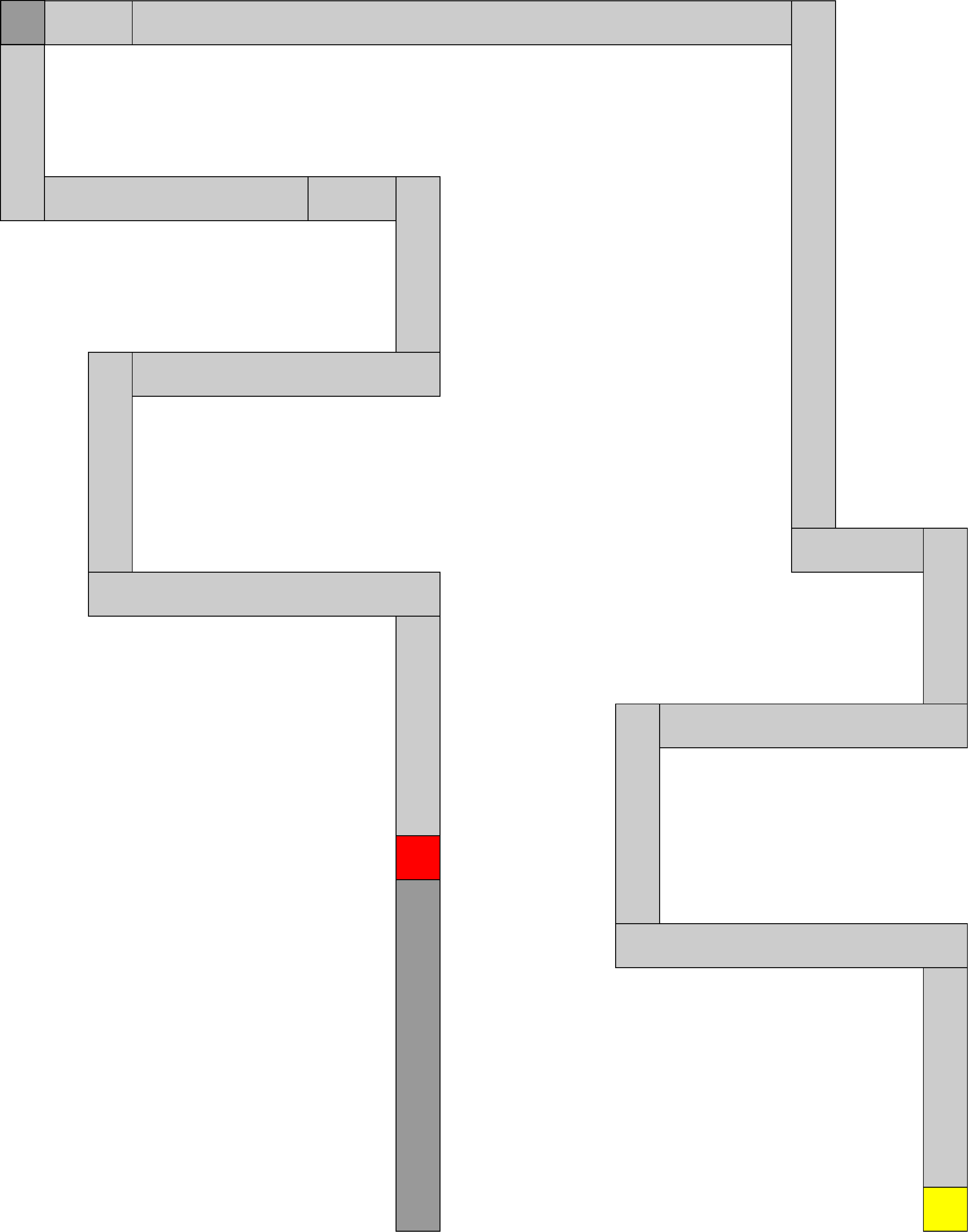} \\
	(c) & (d) \\\hline
    \includegraphics[scale=.07]{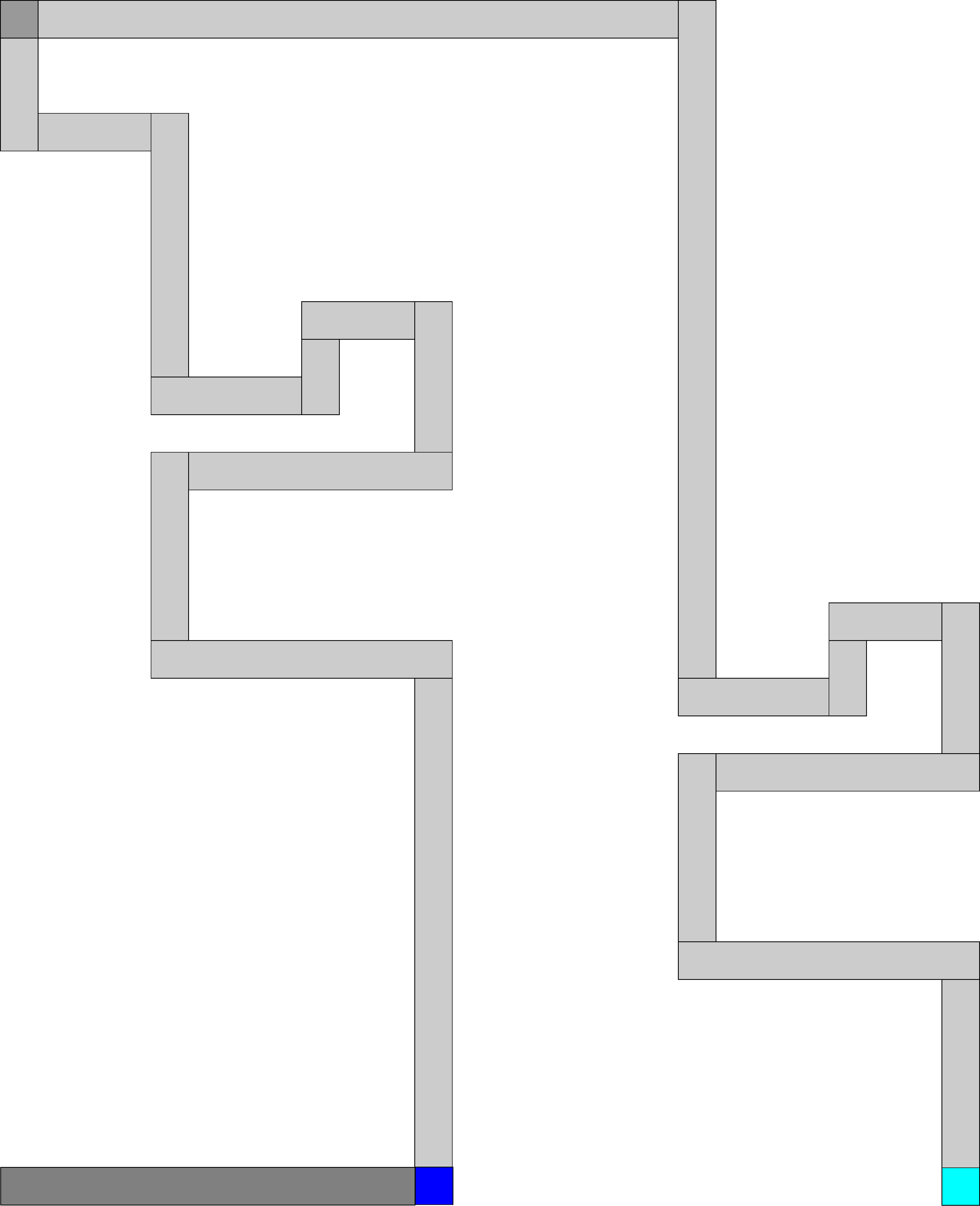}  & \includegraphics[scale=.07]{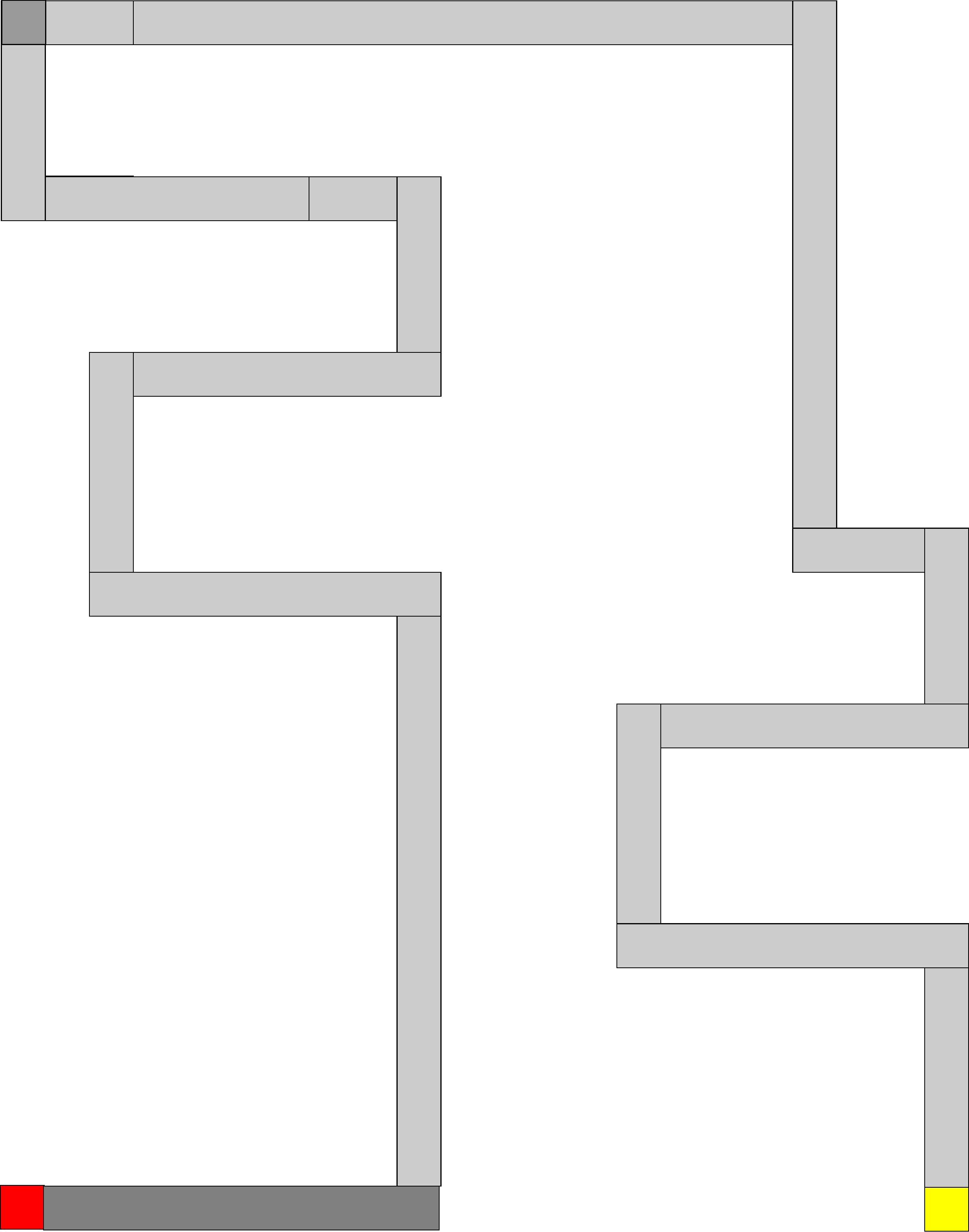} \\
	(e) & (f) \\\hline
    \includegraphics[scale=.07]{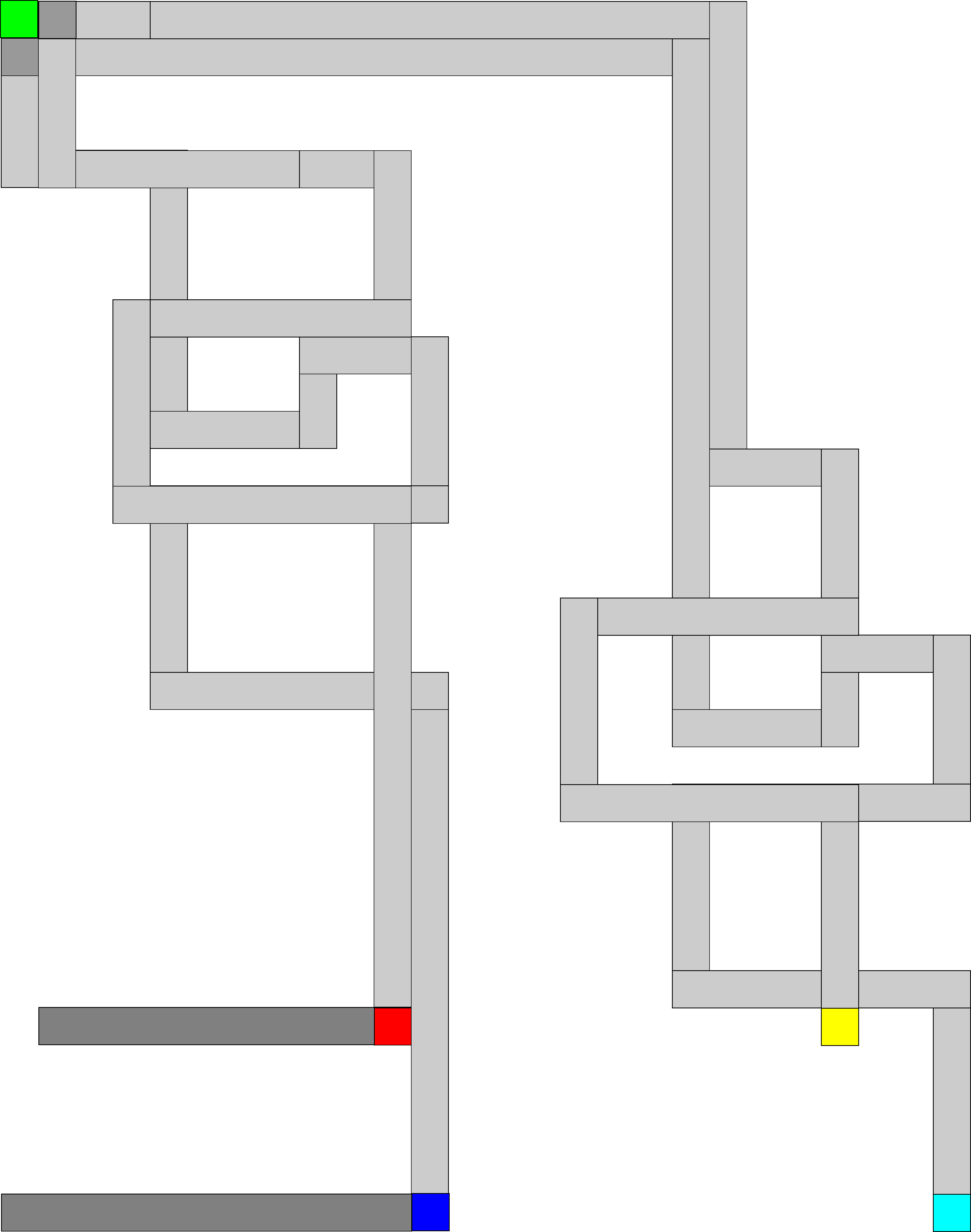}  & \includegraphics[scale=.07]{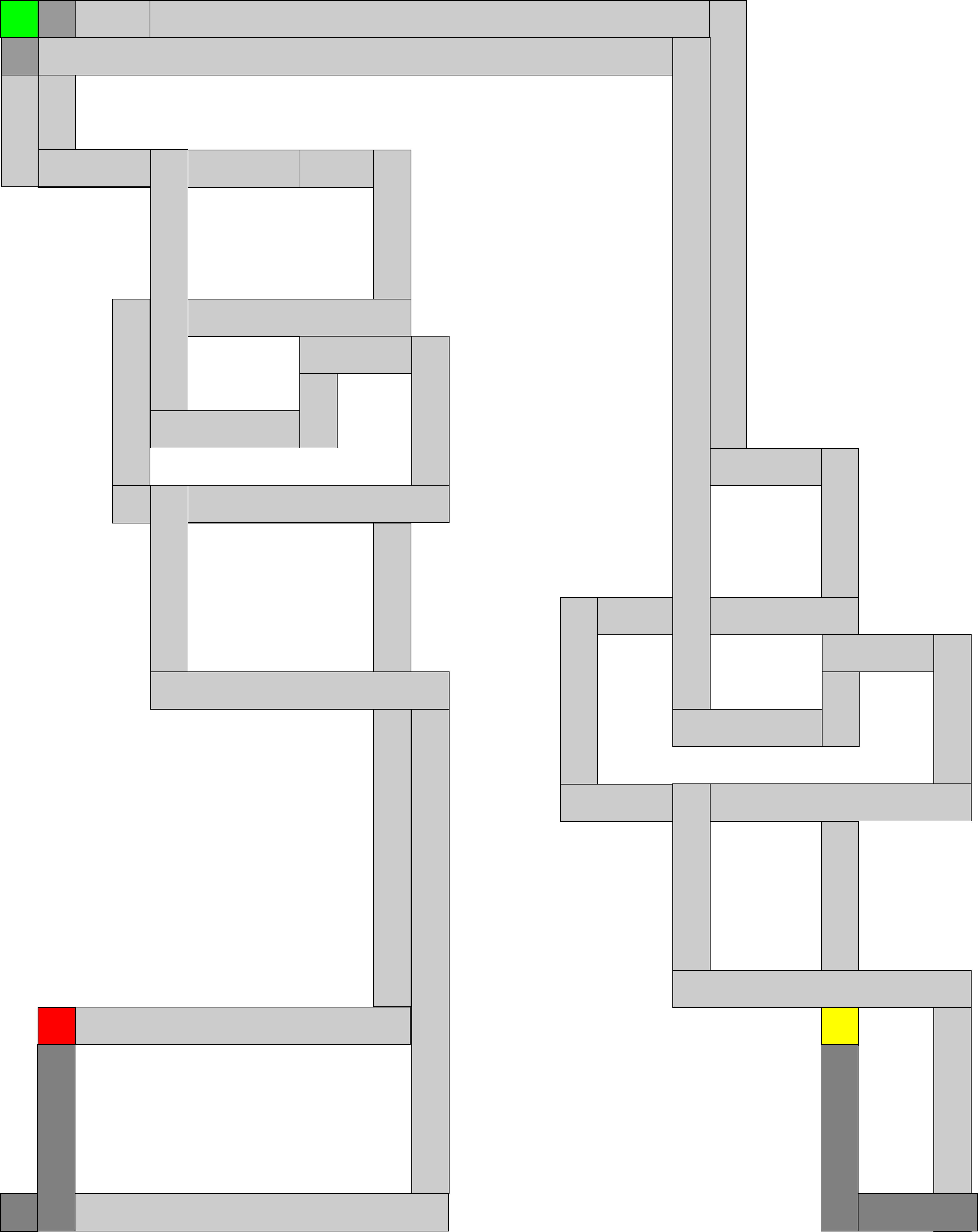} \\
	(g) & (h) \\\hline
    \includegraphics[scale=.07]{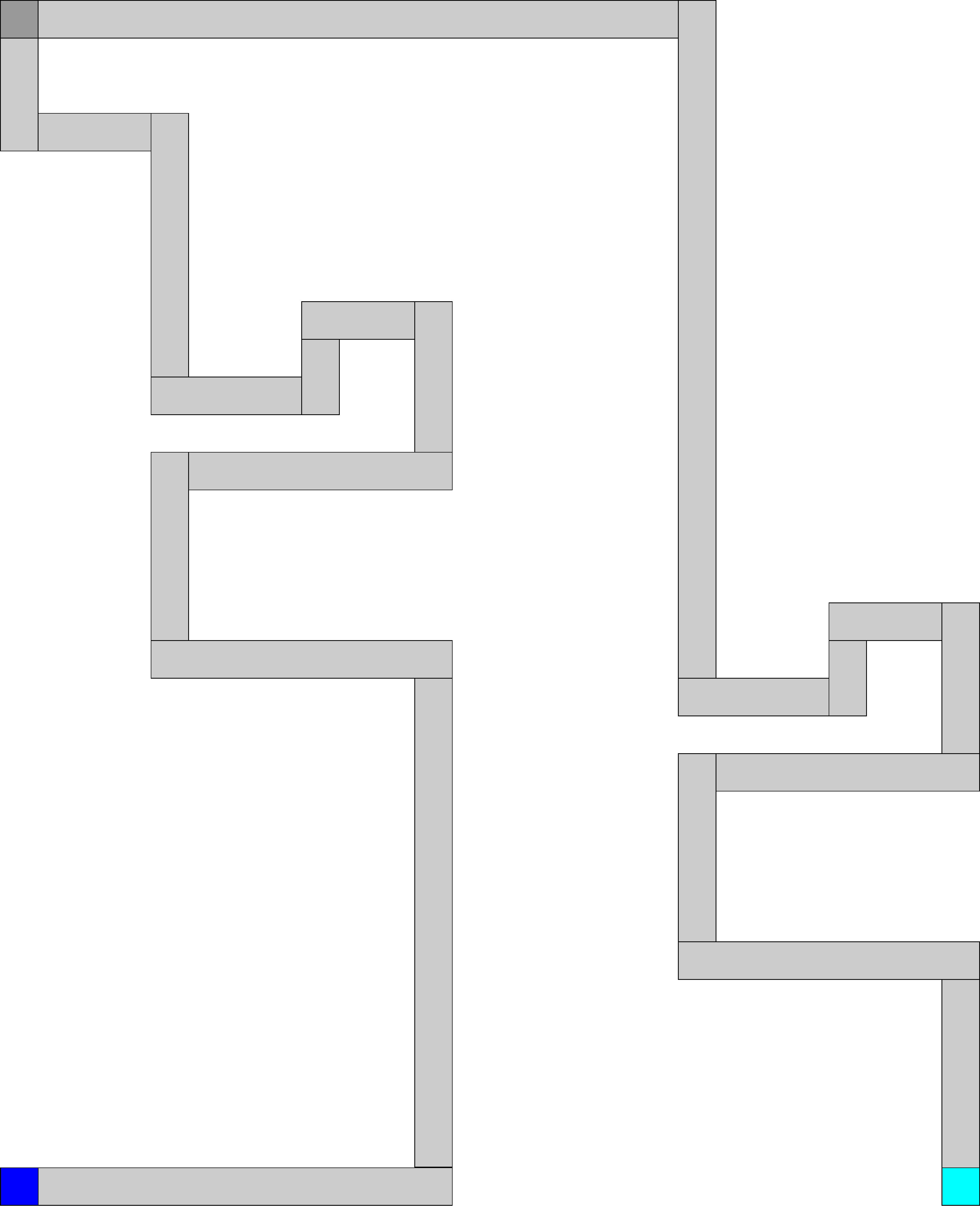}  & \includegraphics[scale=.07]{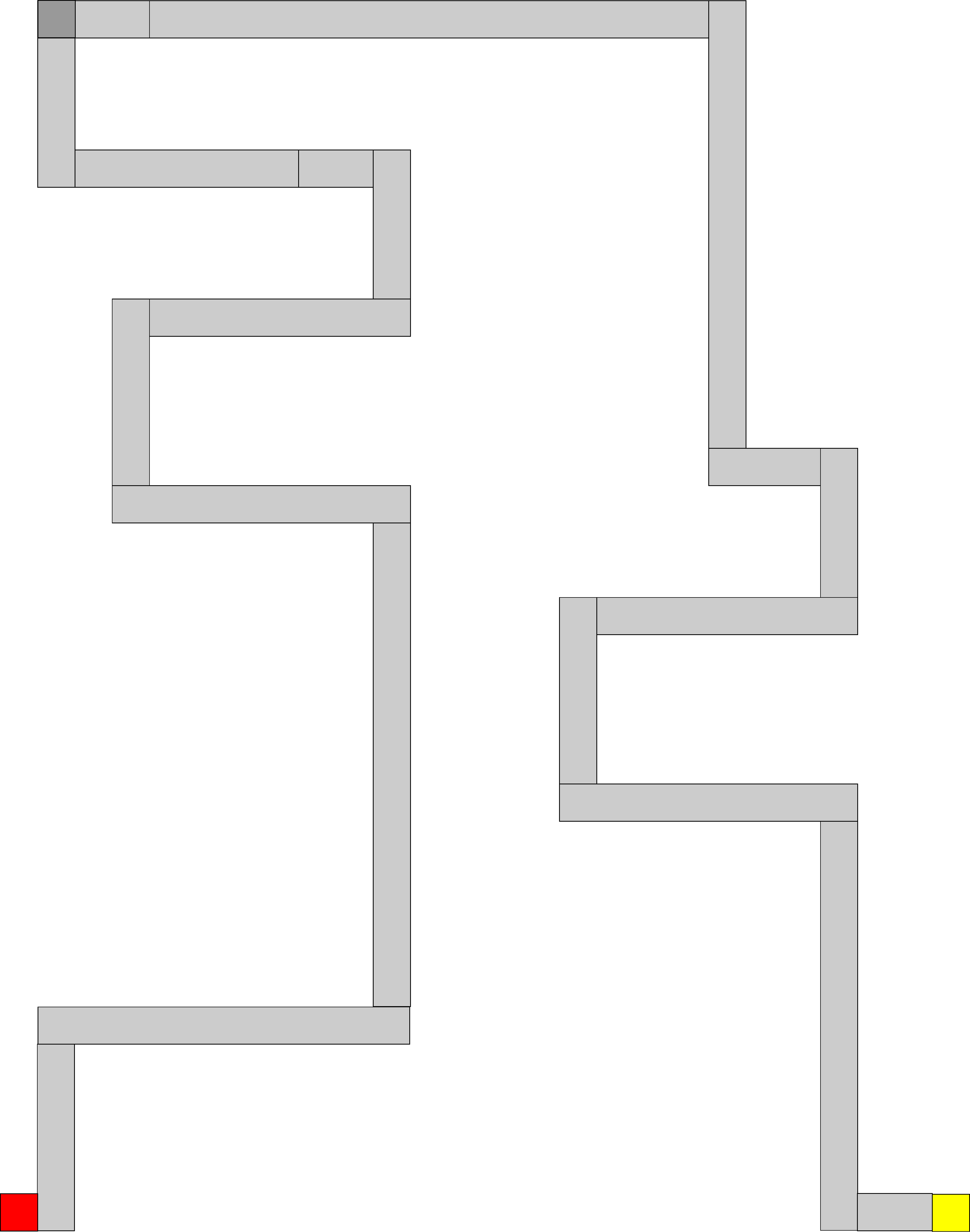} \\
	(i) & (j) \\\hline
\end{tabular}
\caption{A schematic representation of the construction of the bit writers.  The blockers are represented by dark grey squares in the northwest of the assemblies. The assembly which places the $0$-blocker begins growth from the aqua tile and the assembly which places the $1$-blocker begins growth from the yellow tile.  The last tile placed in the assembly that contains the $0$-blocker is shown in blue and the last tile placed by the assembly containing the $1$-blocker is shown in red.  Paths of new tiles that are added at each step are dark grey.}
\label{tbl:3write}
\end{figure}

First we describe the construction of the bit-writer subassemblies of the bit-reading gadget by describing the placement of the blockers in relation to the position of the green tile.  We will discuss how to create tile sets which can create the necessary sets of paths for the gadgets, and then the final tile set will simply consist of a union of those tile sets.  Suppose that the $0$-blocker is a $\vec{x_1}$-shifted polyomino and the $1$-blocker is a $\vec{x_2}$-shifted polyomino. (Recall that $P$ is a basic polyomino, and thus it is possible to build a path such that a blocker can be at any shift relative to the grid.)  We construct two separate systems, say $\mathcal{T}_{0B}$ and $\mathcal{T}_{1B}$ as described in Lemma~\ref{lem:r-shifts} so that the $0$-blocker and $1$-blocker, respectively, are the northernmost tiles on the western edges of the assemblies (shown in part (a) and (b) of Figure~\ref{tbl:3write}). We denote the assemblies produced by these systems as $\alpha_0$ and $\alpha_1$, respectively.  Next, extend the tile sets of the systems if needed so that the last tiles placed lie on grid in the same grid row as the seed as shown in part (c) and (d) of Figure~\ref{tbl:3write}.  In addition, extend the tile sets of the two systems (if needed) so that the last tile placed has pixels that lie in the same column as the westernmost pixel in the blocker or to the west of that column.  This is shown schematically in part (e) and (f) of the figure.  Now, place the green polyomino so that its bounding rectangle's southwest corner lies at the origin, and place the assemblies $\alpha_0$ and $\alpha_1$ constructed above so that the $1$-blocker and $0$-blocker lie relative to the green tile as described above (shown in part (g) of Figure~\ref{tbl:3write}).  Without loss of generality suppose that the seed of $\alpha_1$ lies to the southeast of the seed of $\alpha_0$ (as is the case in the figure).  Then we can extend the tile set of $\mathcal{T}_{0B}$ so that whenever $\alpha_1$ is placed as described above, the seeds of $\alpha_0$ and $\alpha_1$ lie at the same position, since both paths are on grid in those locations.  In addition, without loss of generality suppose that the tile placed last in $\alpha_1$ is further west than the last tile placed in $\alpha_0$.  Then we extend the tile set of $\mathcal{T}_{0B}$ so that the last tile placed in $\alpha_0$ is at the same position as the last tile placed in $\alpha_1$.  These two steps are shown in part (h) of the figure.  The construction of the bit writer gadgets is now complete and the schematic diagram of the completed bit writers is shown in parts (i) and (j) of Figure~\ref{tbl:3write}.

\subsubsection*{Case (3) Bit-Reader Construction}\label{sec:3BR}

\begin{figure}[htp]
\centering
  \subfloat[][The green tile is on grid and the yellow tile is $(-1,-1)$-shifted.]{%
        \label{fig:GenCompB01}%
        \includegraphics[width=2.3in]{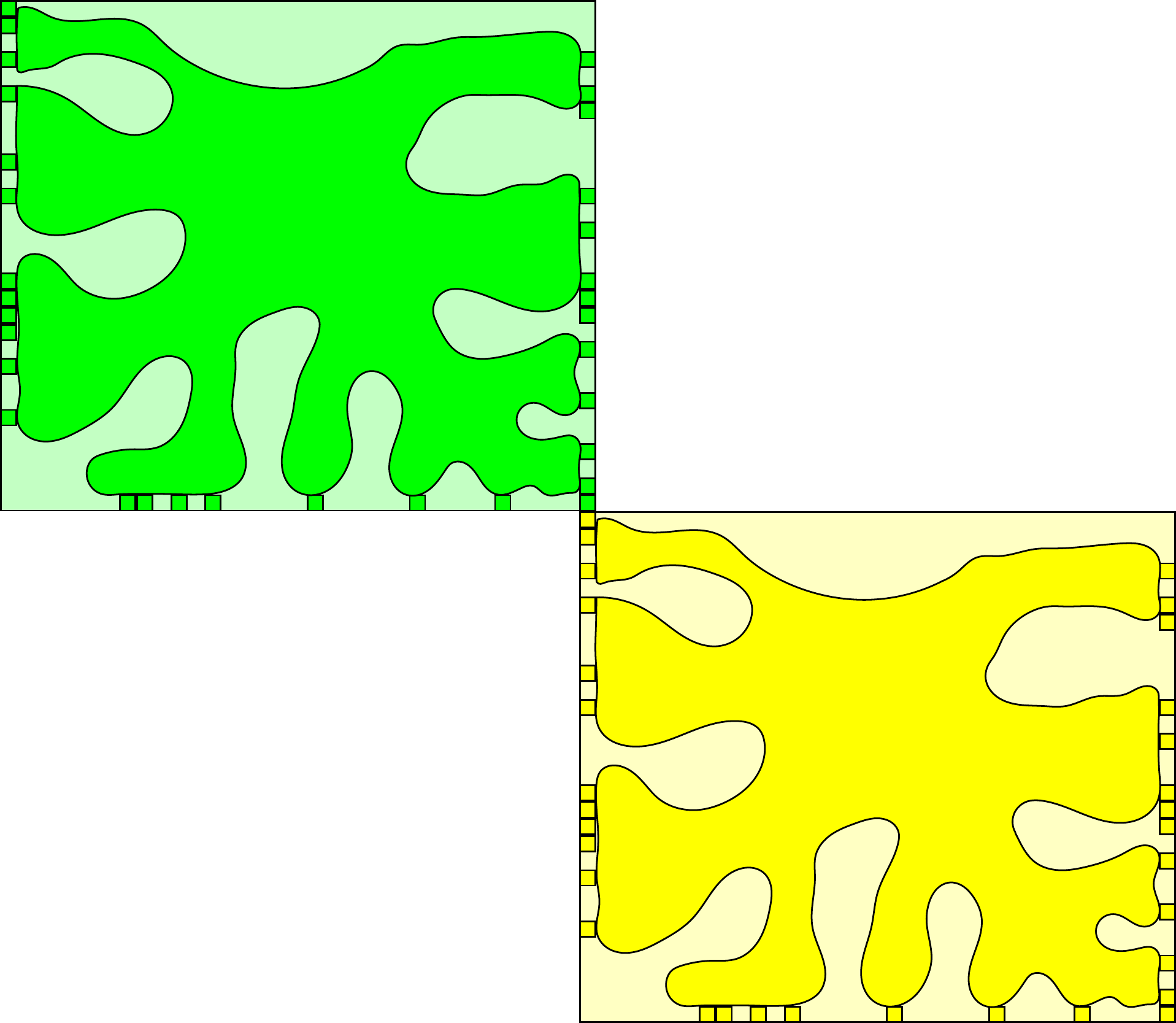}
        }%
        \quad
  \subfloat[][The $0$-blocker is placed so that the easternmost pixel on the north perimeter of the $0$-blocker overlaps the westernmost pixel on the south perimeter of the yellow tile.]{%
        \label{fig:GenCompB02}%
        \includegraphics[width=2.3in]{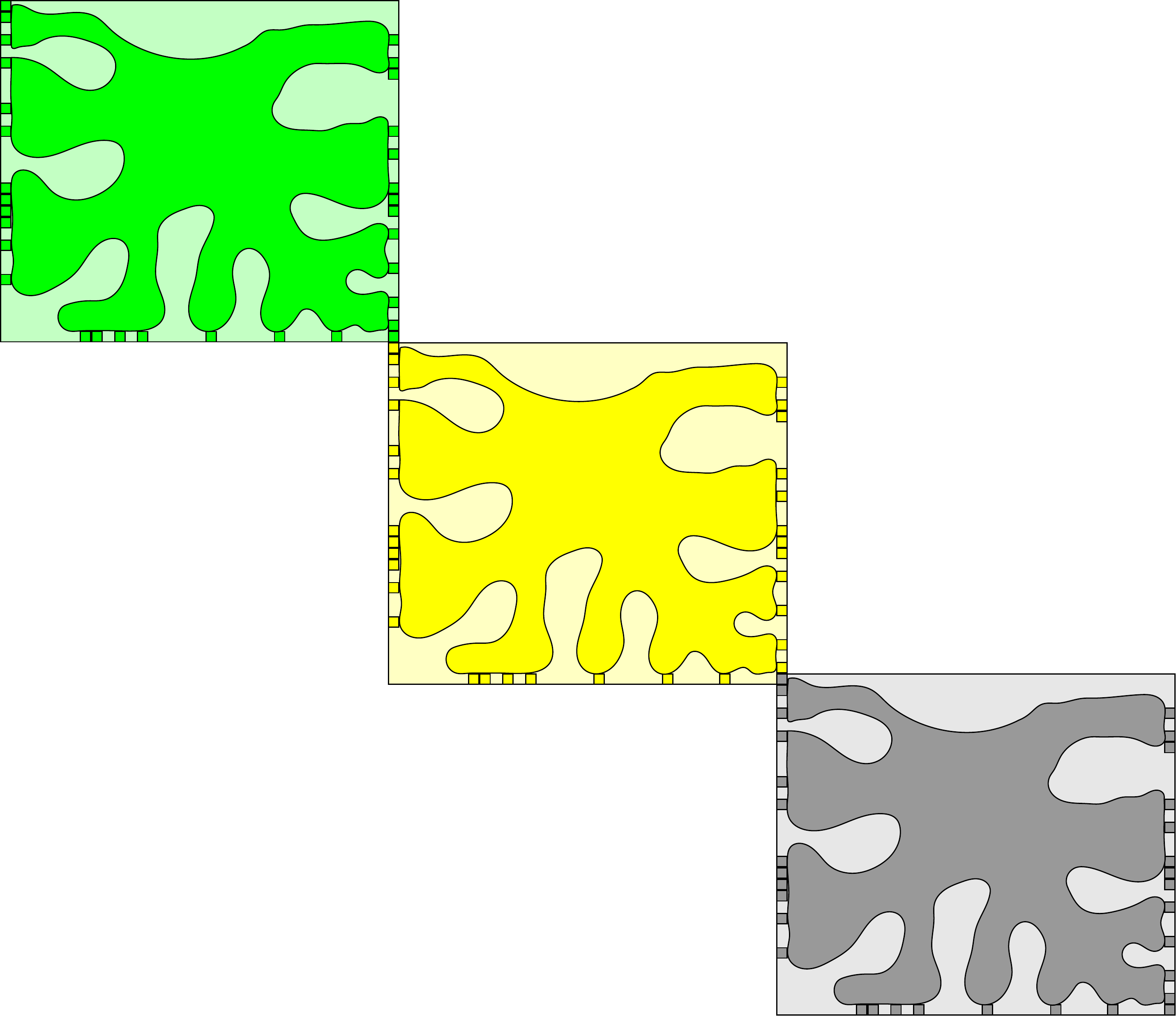}
        }%
        \quad
  \caption{Placement of the $0$-blocker, which blocks the yellow (i.e. $0$-reader) path.}
  \label{fig:GenCompB0}
\end{figure}

\begin{figure}[htp]
\centering
  \subfloat[][The green tile and aqua tile are both on grid.]{%
        \label{fig:GenCompB03}%
        \includegraphics[width=2.3in]{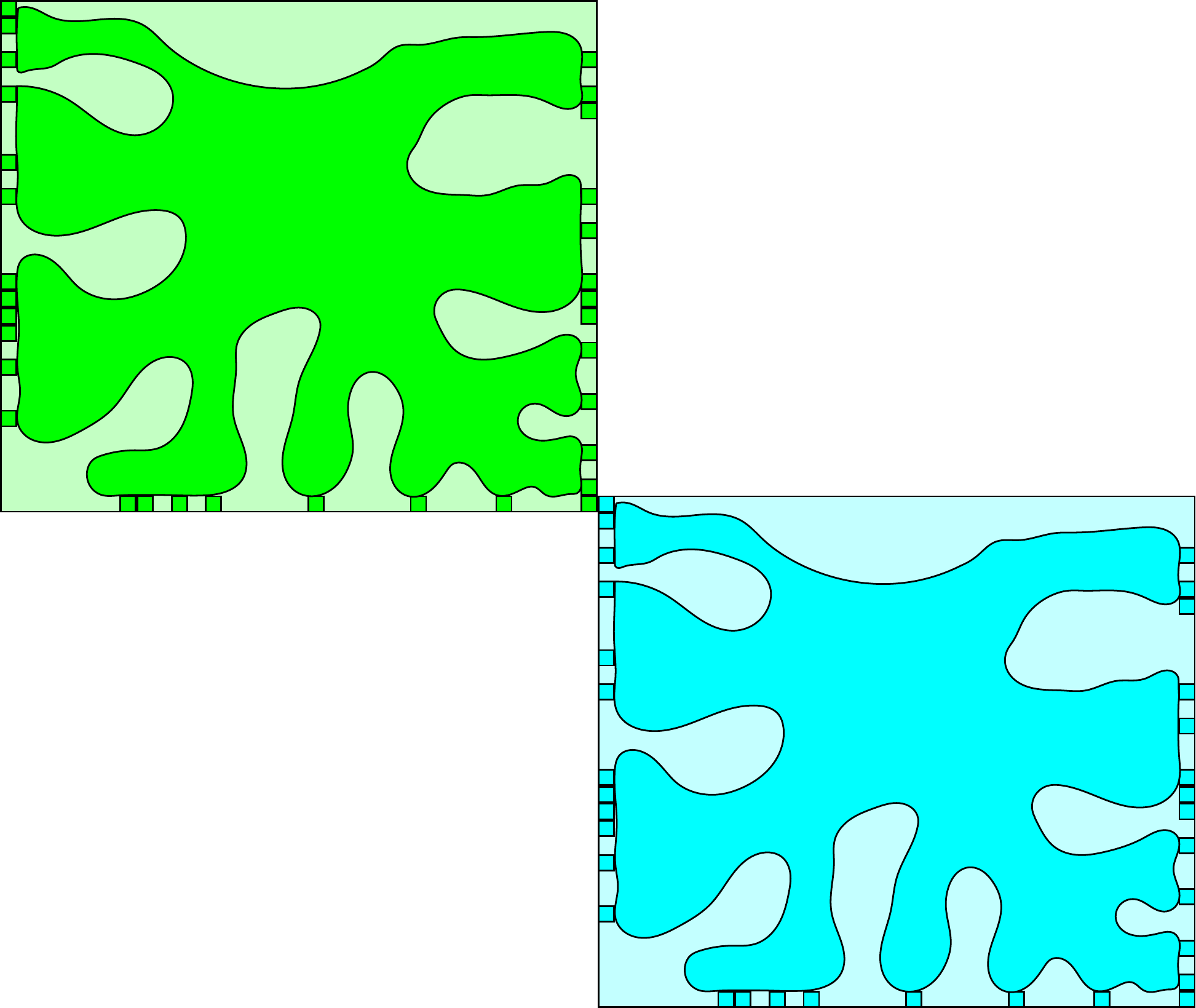}
        }%
        \quad
  \subfloat[][The $1$-blocker is placed so that the northernmost pixel on the western perimeter of the $1$-blocker overlaps the southernmost pixel on the east perimeter of the aqua tile.]{%
        \label{fig:GenCompB04}%
        \includegraphics[width=2.3in]{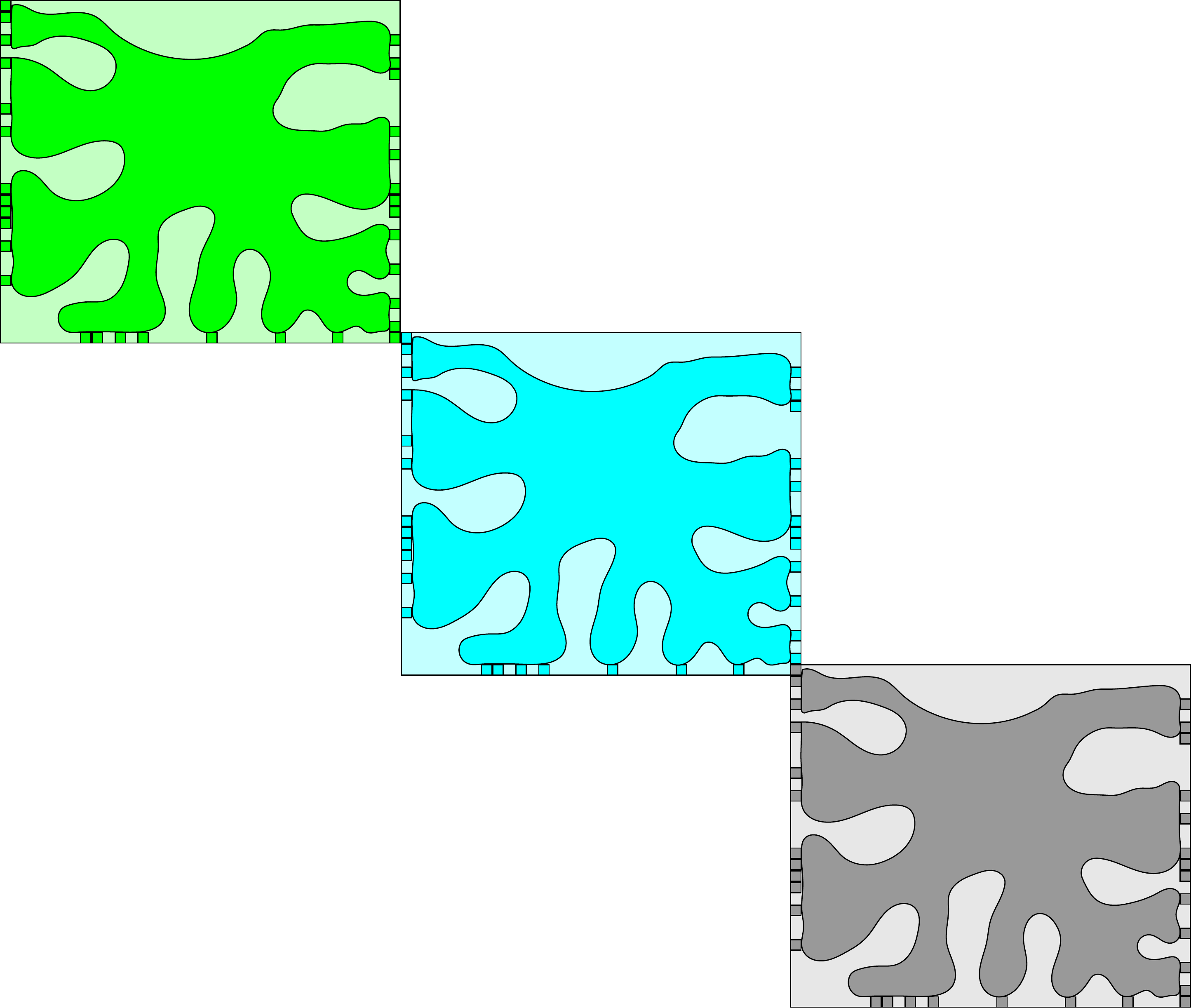}
        }%
        \quad
  \caption{Placement of the $1$-blocker, which blocks the aqua (i.e. $1$-reader) path.}
  \label{fig:GenCompB1}
\end{figure}

Figures~\ref{fig:GenCompB0} and \ref{fig:GenCompB1} show the placement of the $0$-blocker and $1$-blocker, respectively.
Figure~\ref{tbl:BitReadGenCase} shows how the glues are placed on the first and
second tiles in the yellow path (in the figure the second yellow tile is shown as an orange
tile for clarity) so that the second yellow tile binds to the first yellow tile
in the system.  In part (a) of Figure~\ref{tbl:BitReadGenCase}, an orange tile (representing the second tile to attach in the yellow path) is placed so that it now lays directly on top of the yellow tile.  The $d_y$ pixels which lie in the column with
the most pixels are shown as a red column in part (b) of the figure.  Notice that when
the orange tile is translated by the vector $(1,0)$ the $m$ red pixels on the
yellow tile now lay adjacent to the $m$ red pixels on the orange tile (see part
(c)).  Now, we shift the orange tile by the vector $(0,2)$ and make two
observations: (1) the bounding rectangle of the orange tile now no longer
overlaps the bounding rectangle of the grey tile, (2) the orange tile
has a pixel which lies adjacent to a pixel in the yellow tile and/or a
pixel which overlaps a pixel in the yellow tile as shown in part (d) of
the figure.  In the case that the orange tile contains pixels which overlap
pixels in the yellow tile, we translate the orange tile to the north until
no pixels overlap, but pixels lie adjacent to each in the two tile
(shown in part (e) of the figure).

We now have a configuration as shown in part (f) of
Figure~\ref{tbl:BitReadGenCase} in which there are not any overlapping pixels and
the yellow and orange tiles have pixels which lie adjacent to each other.
We can now place glues on the green, yellow and orange tiles so that they
assemble as shown with the yellow tile attaching to the green tile
and the orange tile attaching to the yellow tile.

\newcolumntype{N}{>{\centering\arraybackslash}m{\dimexpr.33\linewidth-2\tabcolsep}}
\begin{figure}[htp]
\centering
\begin{tabular}{| N | N | N |}
	\hline
	\includegraphics[scale=.175]{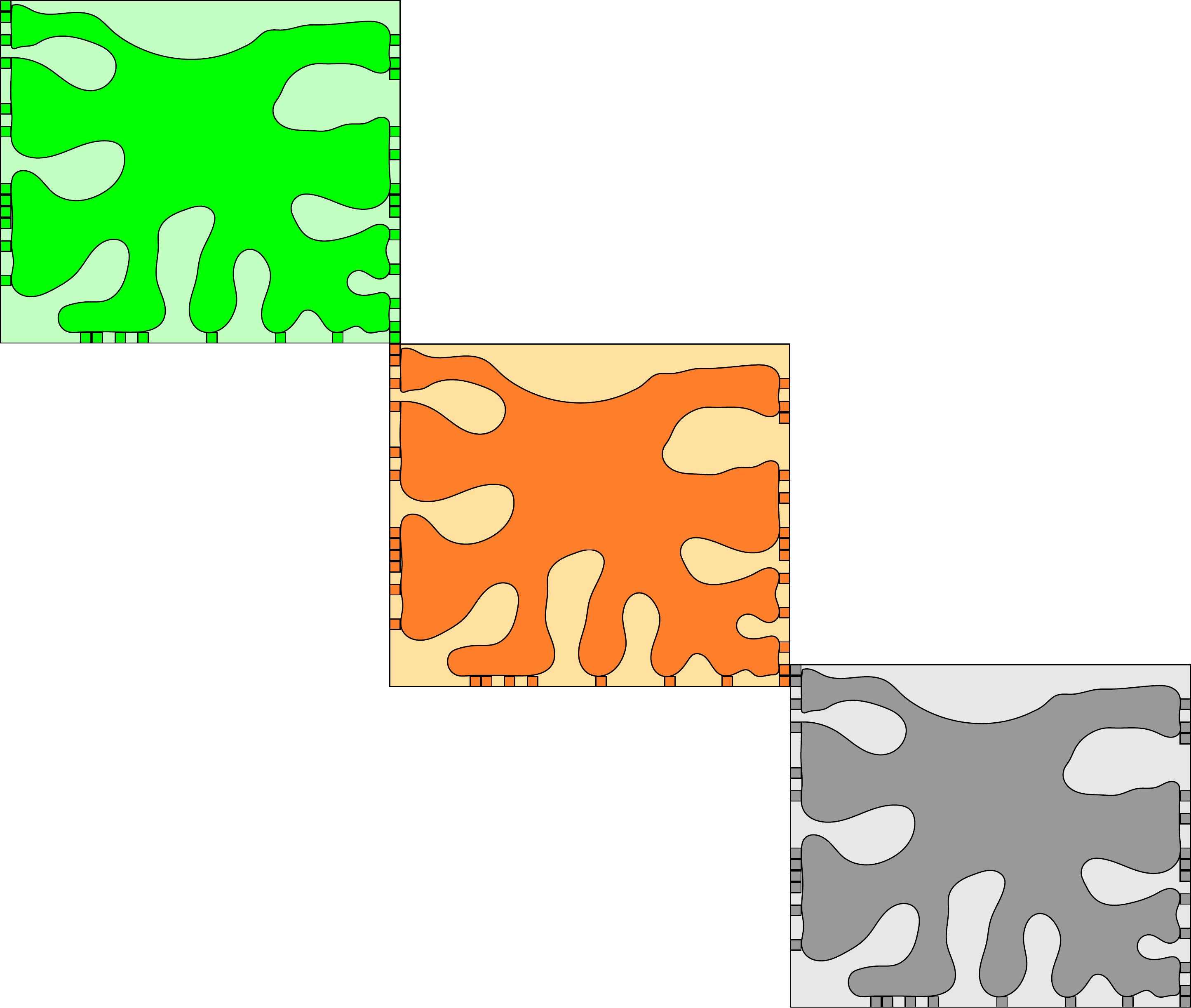}  & \includegraphics[scale=.175]{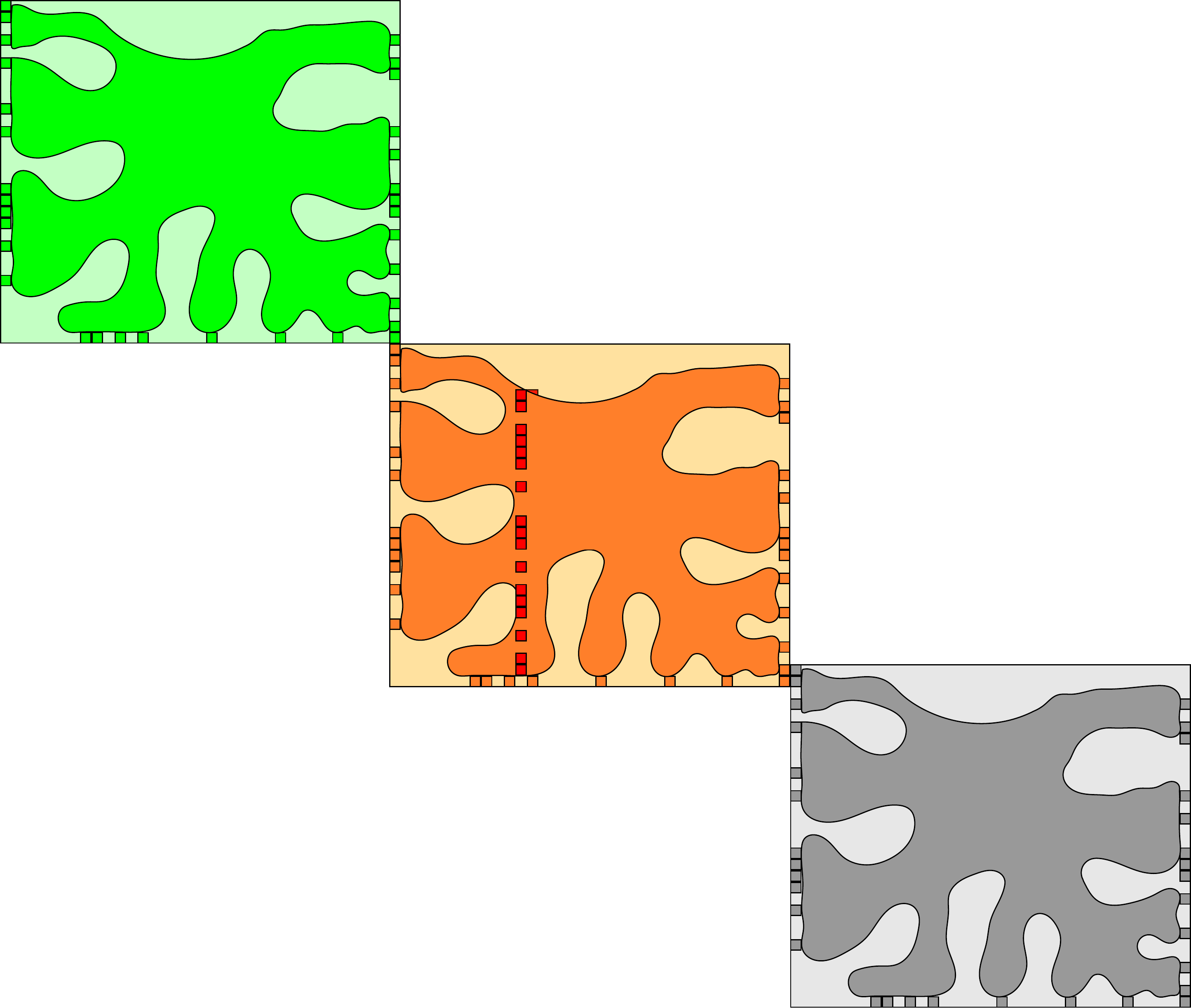} &
	\includegraphics[scale=.175]{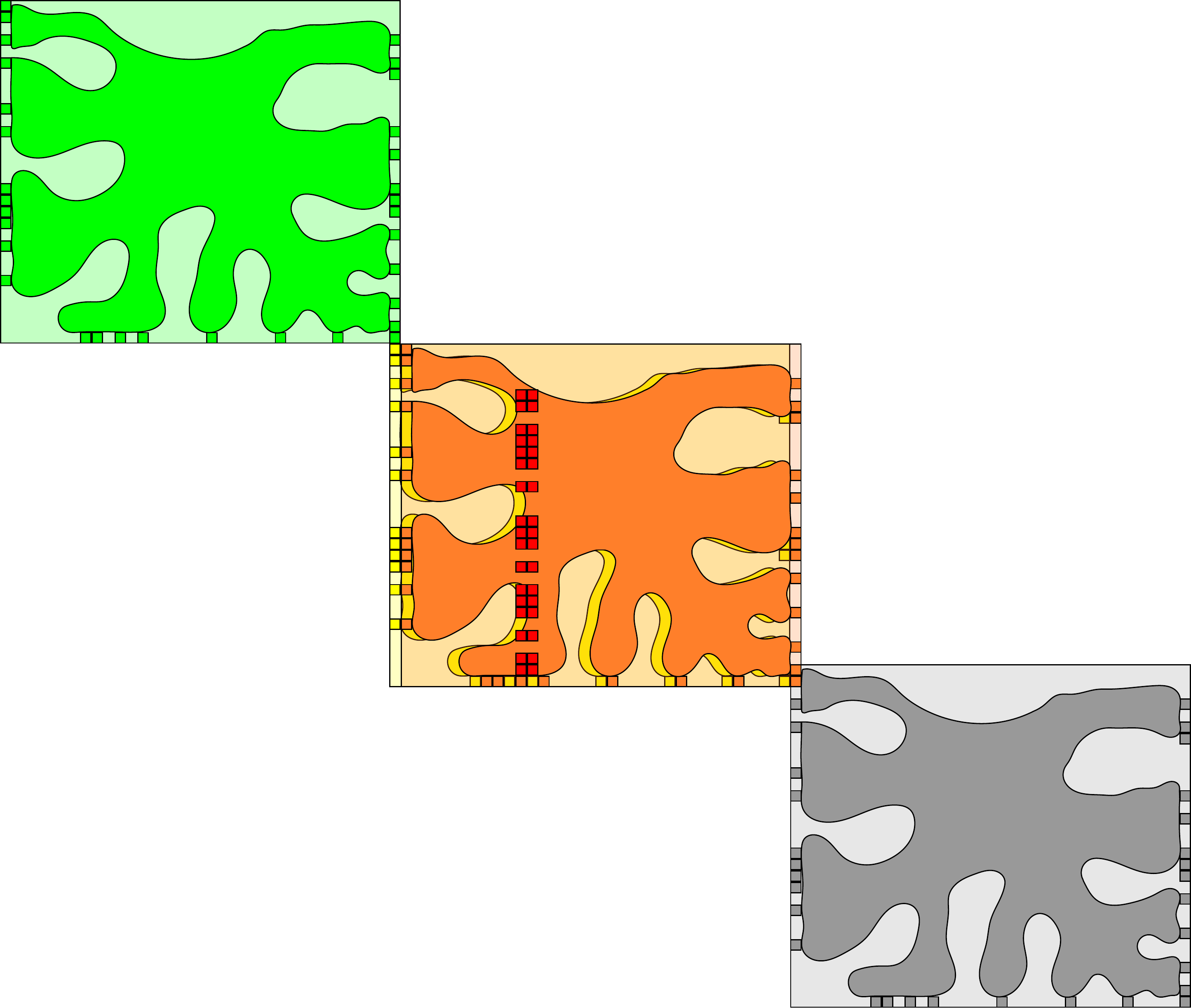}  \\
	(a) & (b) & (c) \\\hline
	\includegraphics[scale=.175]{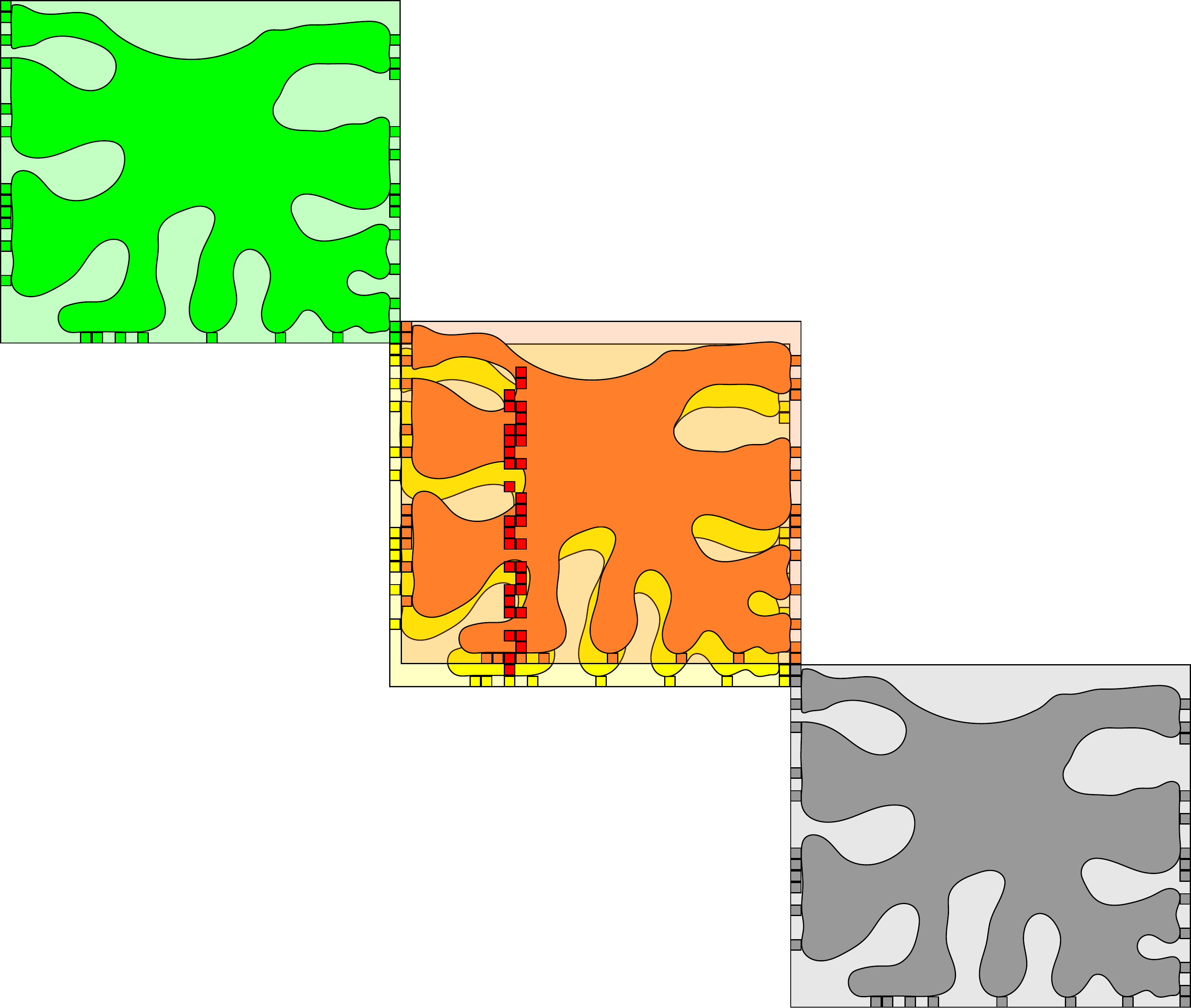}  & \includegraphics[scale=.175]{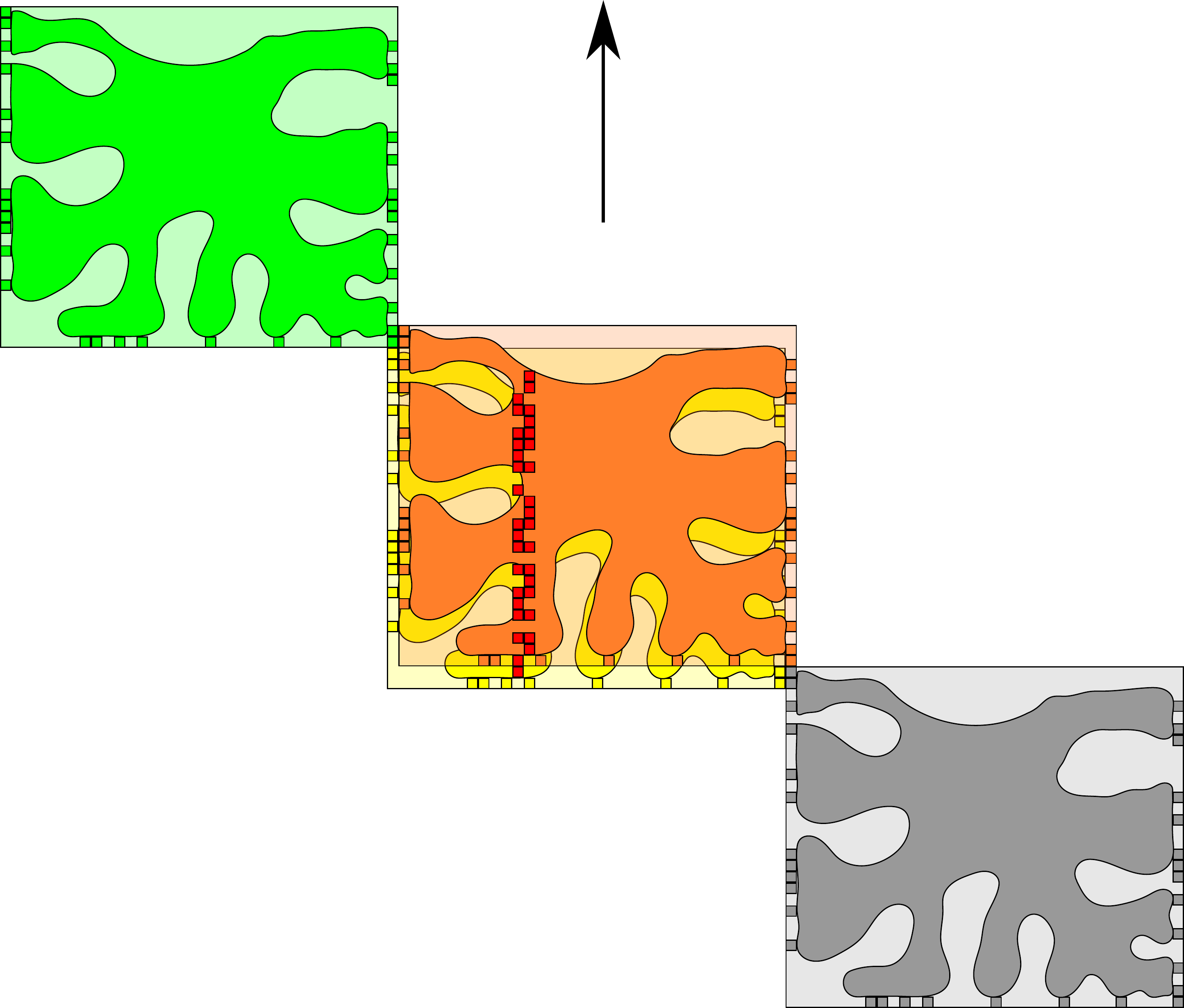} &
	\includegraphics[scale=.175]{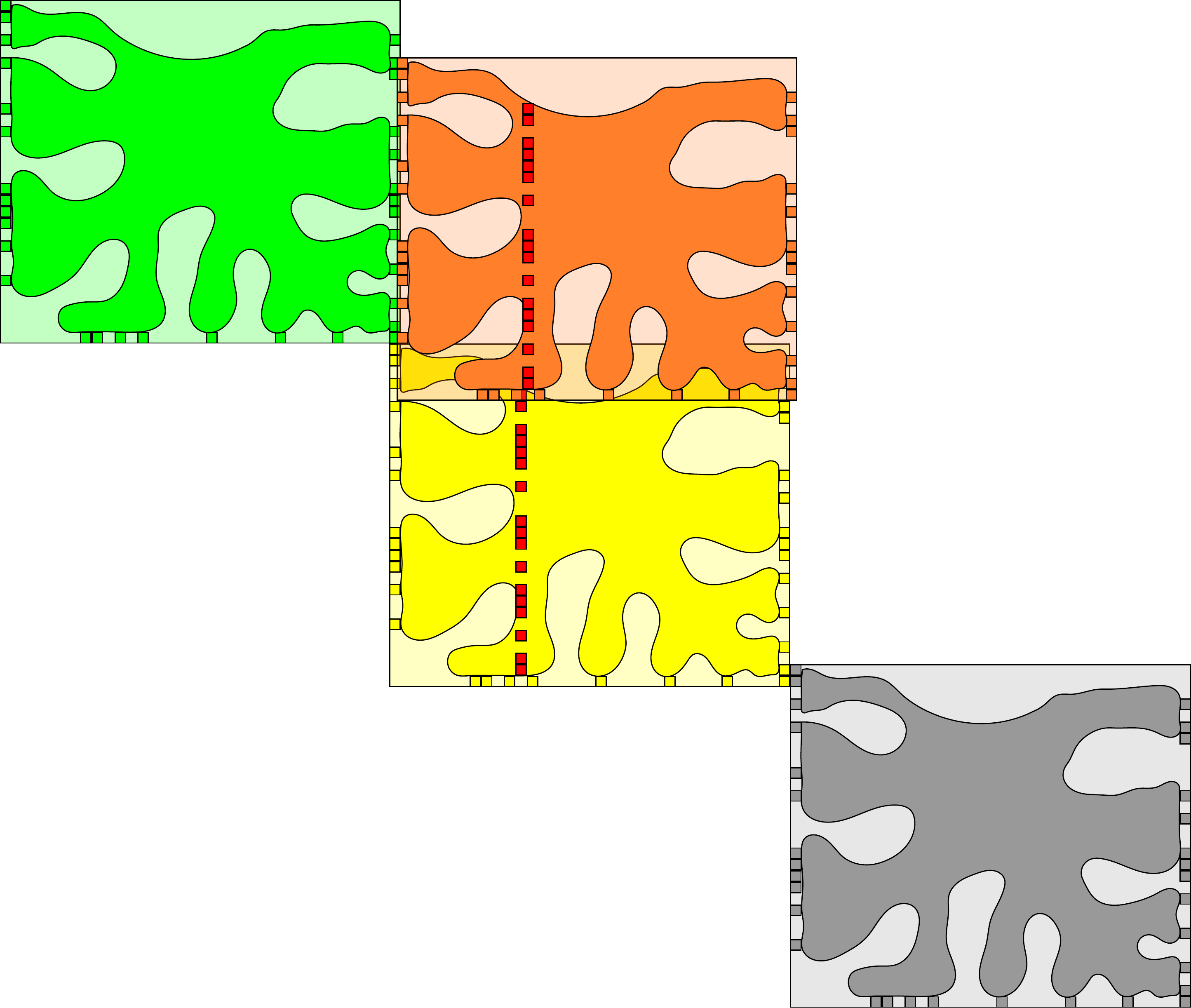}  \\
	(d) & (e) & (f) \\\hline
\end{tabular}
\caption{The steps involved in placing the grey blocker and yellow tiles so that an aqua tile is prevented from binding to the green tile, but allows for a yellow tile to bind to the green and then continue growth of the yellow path on the grid.  The orange tile represents the second tile of the yellow path (to make it easier to distinguish), and figures (a) through (e) show how it can be initially placed immediately on top of the first yellow tile, and then moved into a position which allows for correct binding.}
\label{tbl:BitReadGenCase}
\end{figure}

We now describe how glues are placed on the first and second tiles to assemble in the path of aqua colored tiles.  Figure~\ref{fig:BitReadGenCase11} shows how the $0$-blocker lies in relation to the aqua and green tiles.  Notice that a tile can attach to the north of the aqua tile without overlapping any pixels on other tiles.  Thus, the second tile to attach in the aqua path is placed to the north of the first aqua tile in the path such that it is on the grid.  This is shown in Figure~\ref{fig:BitReadGenCase12} where we use a purple tile to represent the second tile in the aqua path for clarity.  Consequently, we place glues on the first and second tiles to attach in the aqua path in a manner such that the second tile in the path binds on grid with respect to the first tile.

\begin{figure}[htp]
\centering
  \subfloat[][The green tile is on grid and the yellow tile is $(-1,-1)$-shifted.]{%
        \label{fig:BitReadGenCase11}%
        \includegraphics[width=2.3in]{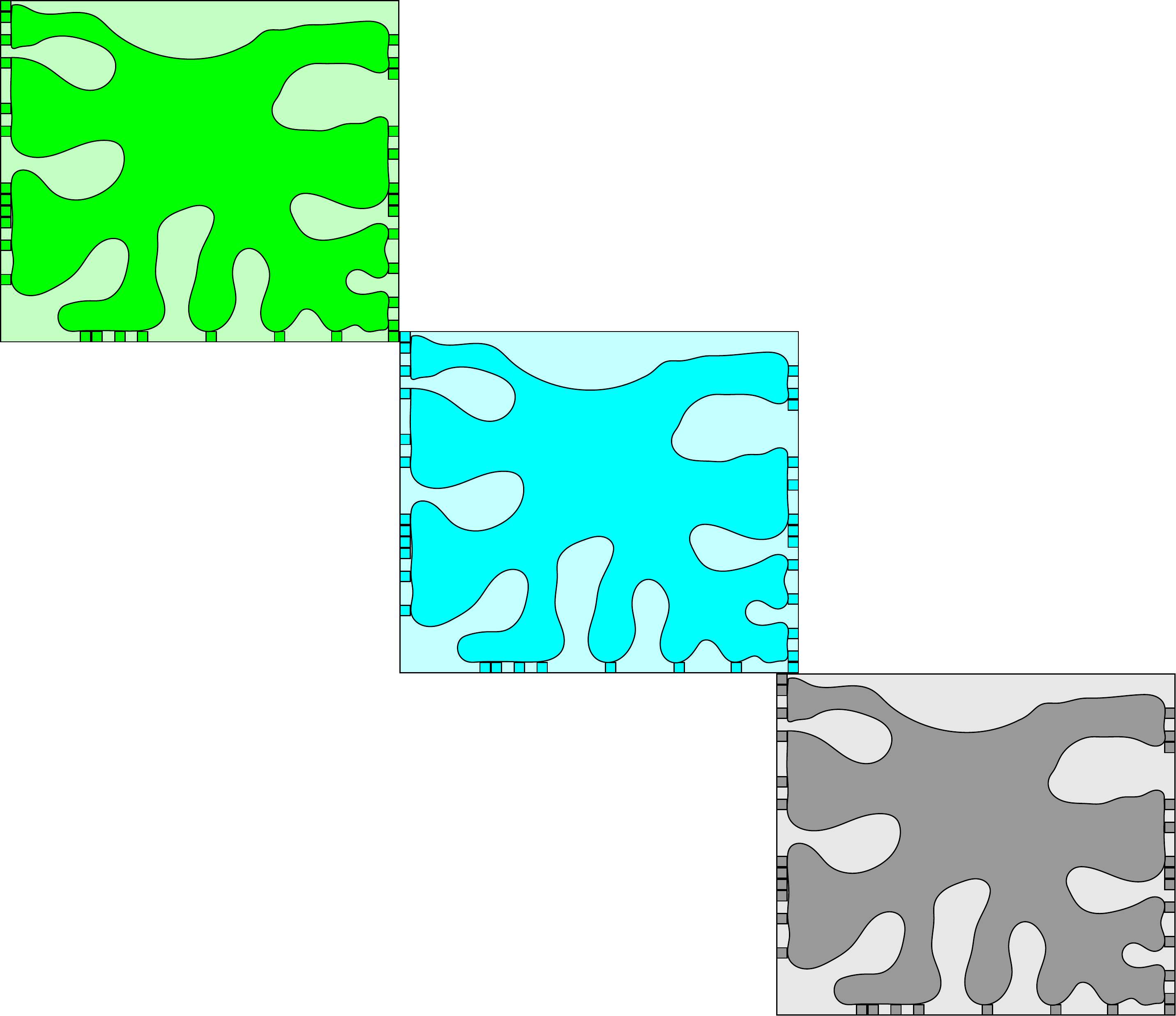}
        }%
        \quad
  \subfloat[][The $0$-blocker is placed so that the westernmost pixel on the north perimeter of the $0$-blocker overlaps the easternmost pixel on the south perimeter of the yellow tile.]{%
        \label{fig:BitReadGenCase12}%
        \includegraphics[width=2.3in]{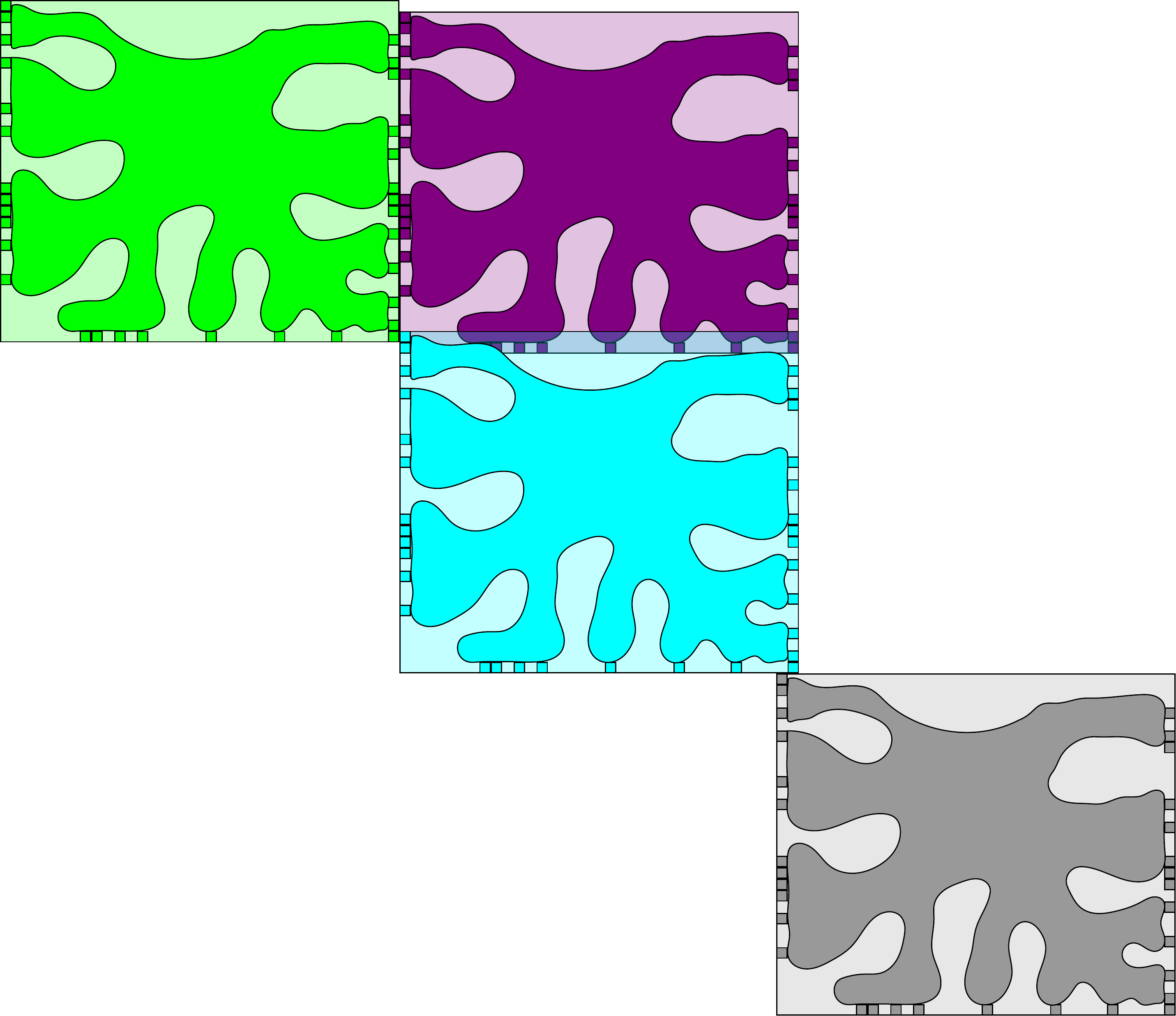}
        }%
        \quad
  \caption{Placement of the second tile in the aqua path (shown as purple to make it easier to distinguish).}
  \label{tbl:BitReadGenCase1}
\end{figure}

\subsubsection*{Case (3) Right-to-Left Bit-Reading Gadget Construction}
As the above sections describe how to build the left-to-right bit-reading gadget, we now construct the right-to-left bit-reading gadget by using mirrored versions of the arguments given above with a few small changes.  In the left-to-right bit-reading gadget we can always place the yellow tile so that it is a $(-1,-1)$-shifted polyomino.
Notice that this is not the case when the bit reader is growing to the west.  Thus we make the following changes to the argument above when constructing the right-to-left bit-reading gadget.  For convenience, we call the first tile to attach in the aqua path $t_a$ and the first tile to attach in the yellow path $t_y$.  To begin, we attach $t_a$ to the green tile so that the northernmost pixel on the east perimeter of $t_a$ attaches to the southernmost pixel on the western perimeter of the green tile via their east/west glues.  Say that this places the aqua tile so that it is an $(x_1,x_2)$-shifted polyomino.  Note that this means $t_a$ is not necessarily on grid since as noted above the grid we are using is formed by attaching the southernmost pixel on the east perimeter of $P$ to the northernmost pixel on the western perimeter of $P$.  Now, observe that this implies that we can also attach a $(x_1+1,x_2-1)$-shifted tile to the green tile (by the points that we used for their attachment at $(x_1,x_2)$).  We thus construct glues so that $t_y$ attaches to the green tile such that it is a $(x_1+1,x_2-1)$-shifted polyomino.  Now, we can construct the bit-writers as in Section~\ref{sec:3BW} with the blockers shifted in the following ways: (1) the $1$-blocker is shifted so that when it is placed its northernmost pixel on the east perimeter overlaps the southernmost pixel on the western perimeter of $t_a$, and (2) the $0$-blocker is placed so that its easternmost pixel on its north perimeter overlaps the westernmost pixel on the south perimeter of $t_y$.  We can then use the mirrored version of the construction in section~\ref{sec:3BR} to grow the rest of the path of tiles composing the yellow and aqua paths.

\subsubsection*{Case (3) Correctness of the Bit-Reading Gadget} ~\label{sec:3PC}

Let us now examine what our constructed system will assemble.  Growth will
start with the seed and then grow two bit-writer subassemblies consecutively.  For concreteness, suppose that $\alpha_1$ is grown first and then $\alpha_0$.  After $\alpha_0$ is assembled, a path of tiles will grow upward and over to place a green tile such that the green tile will be placed with its position relative to the grey tile as shown in part (a) of Figure~\ref{tbl:BitReadGenCase}.  It
then follows by the way we placed the green and yellow tiles and the location of $\alpha_0$ that the
yellow path will be able to assemble.  This will eventually lead to the placement of the second green tile, which is placed to the west of the second bit writer.  The relative placements of that green tile and the blocker of $\alpha_1$ ensure that an aqua path, and only an aqua path, will assemble.  This concludes the necessary demonstration of the correct growth of a bit-reader gadget.  (The full Turing machine simulation also includes bit-writers, designed as previously described, to output between bit-readers and the necessary zig-zag paths.)
\qed

\subsubsection*{Case (3) Example}
Figure~\ref{tbl:BitReadGenCaseE0} shows a concrete example of the steps outlined in Figure~\ref{tbl:BitReadGenCase}, and Figure~\ref{tbl:BitReadGenCaseE1} shows, using the same example, the steps outlined in Figure~\ref{tbl:BitReadGenCase1}.  In Figure~\ref{fig:GenComp0EX} and Figure~\ref{fig:GenComp1EX} we see how the placement of the polyominoes that prevent paths from growing can be combined with Lemma~\ref{lem:r-shifts} to create a bit-reading gadget.

\begin{figure}[htp]
\centering
\begin{tabular}{| M | M | M |}
	\hline
	\includegraphics[scale=.13]{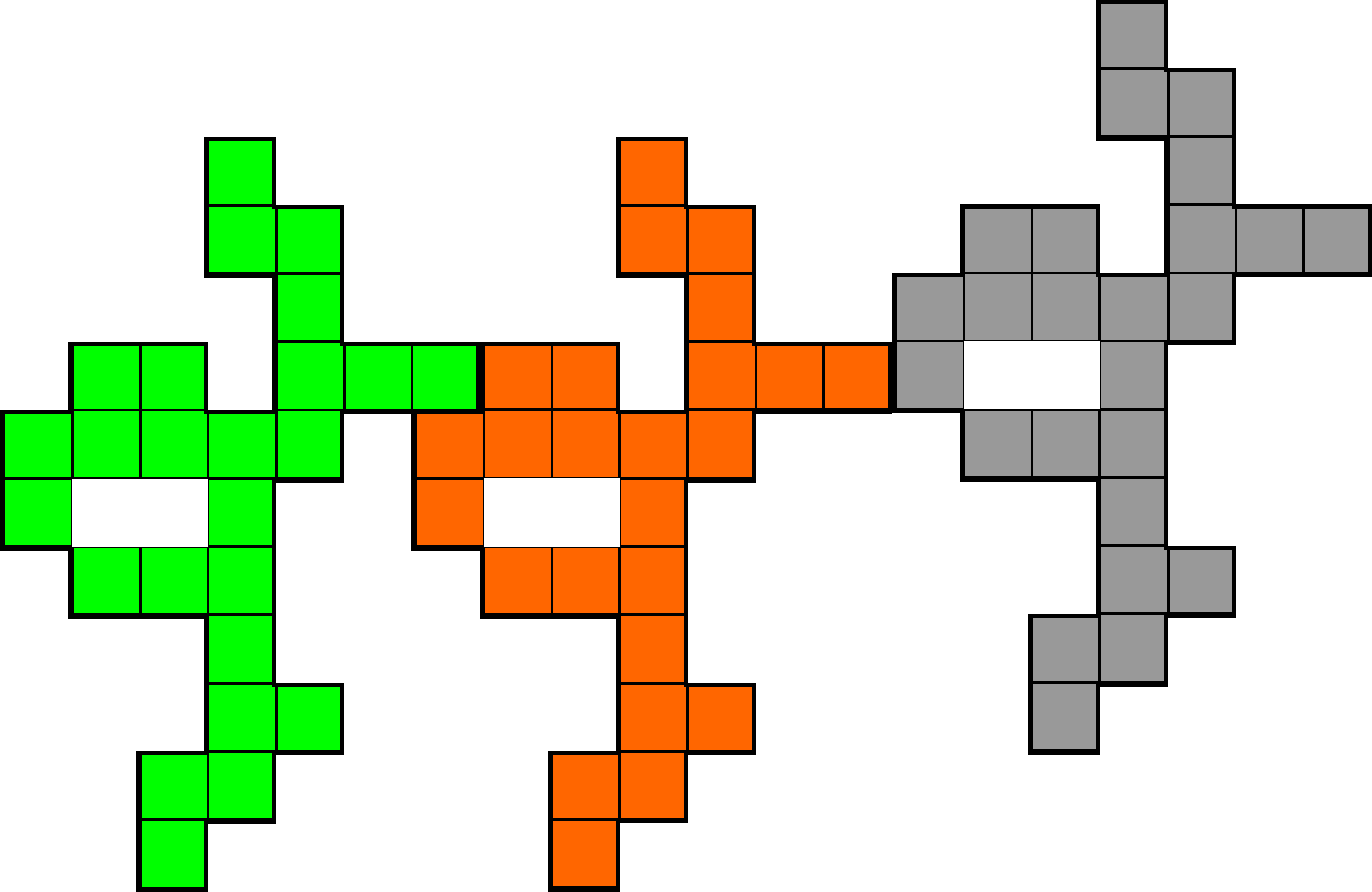}  & \includegraphics[scale=.13]{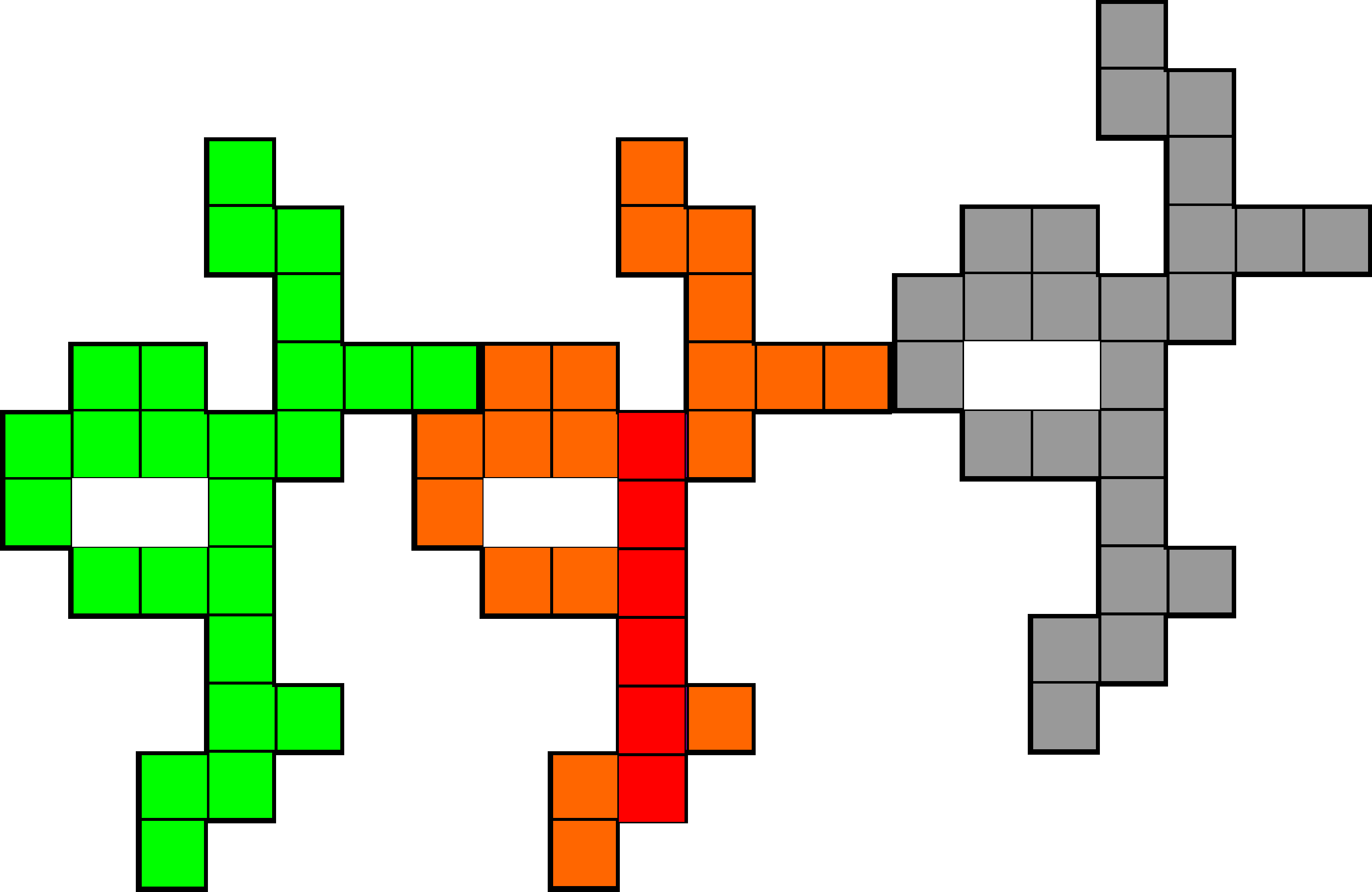} &
	\includegraphics[scale=.13]{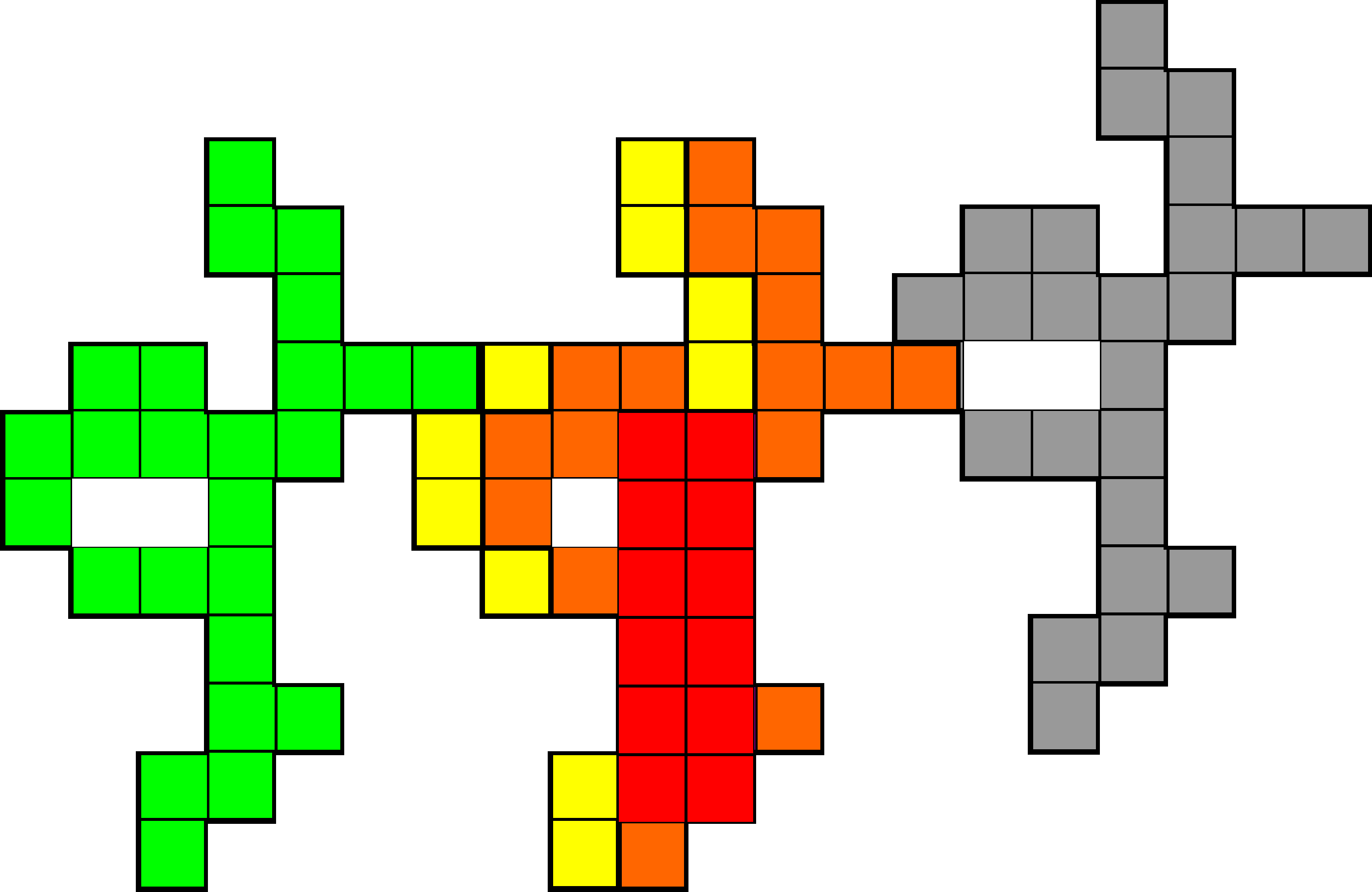}  \\
	(a) & (b) & (c) \\\hline
    \includegraphics[scale=.13]{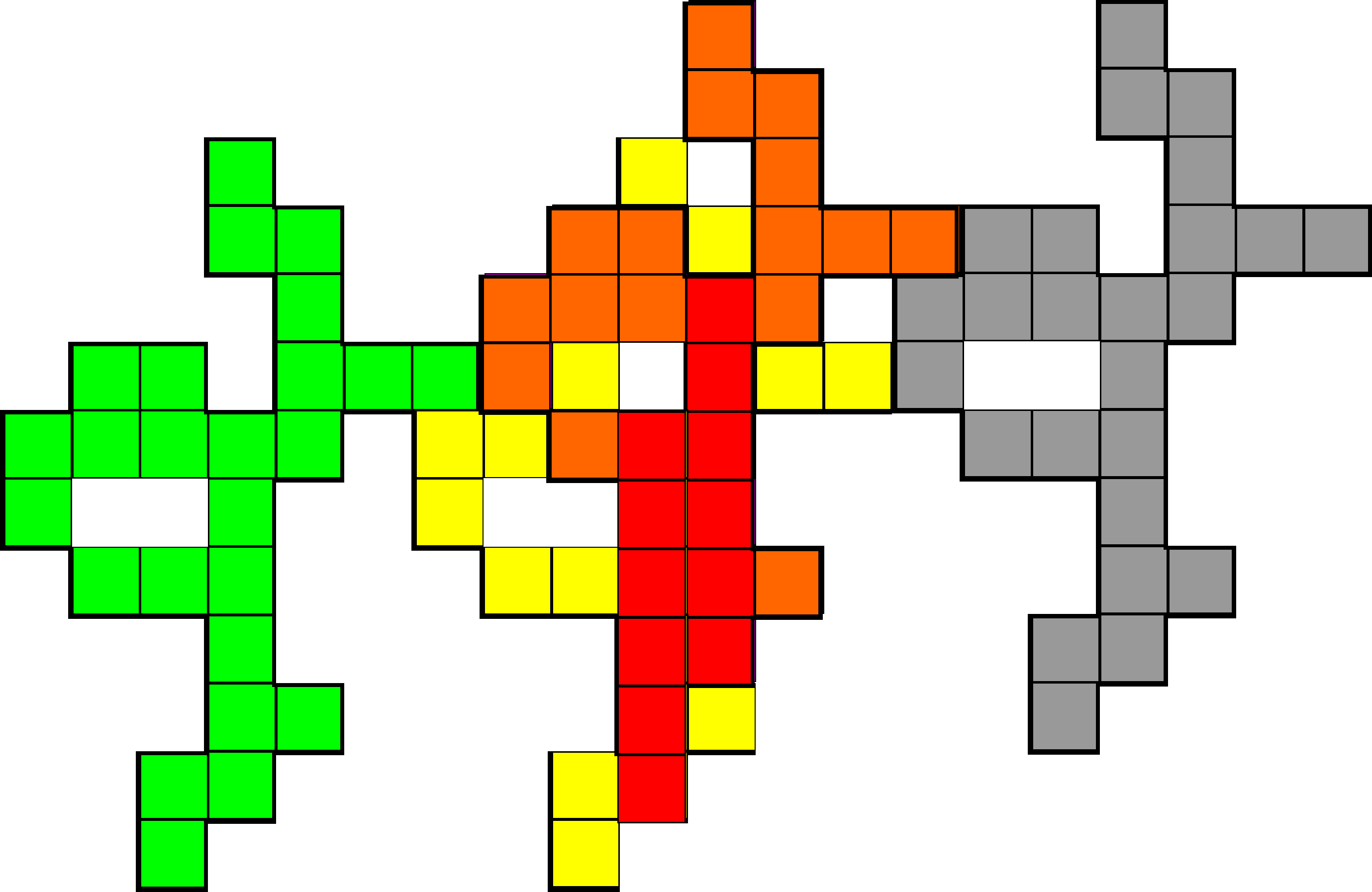}  & \includegraphics[scale=.13]{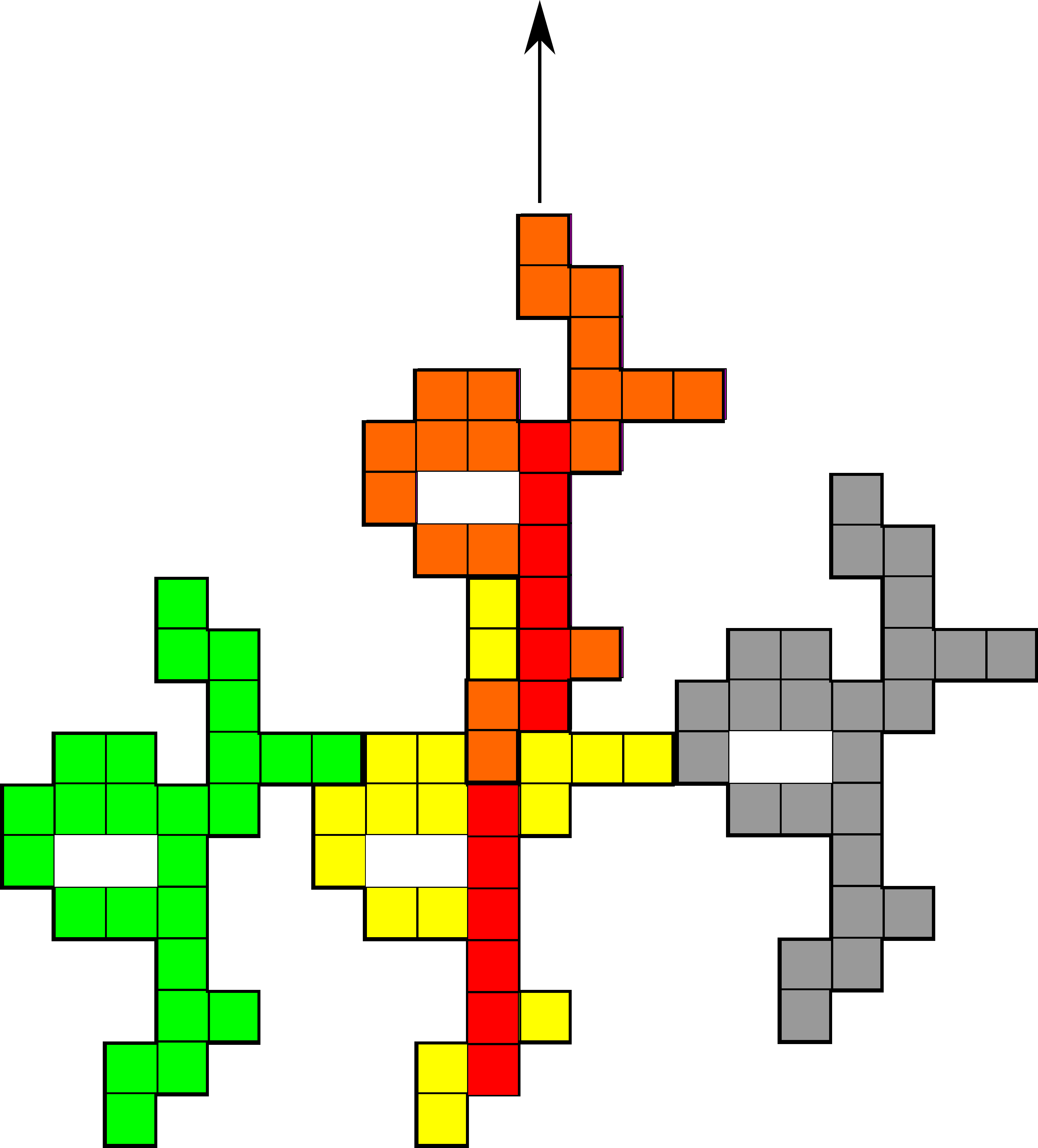} &
	\includegraphics[scale=.13]{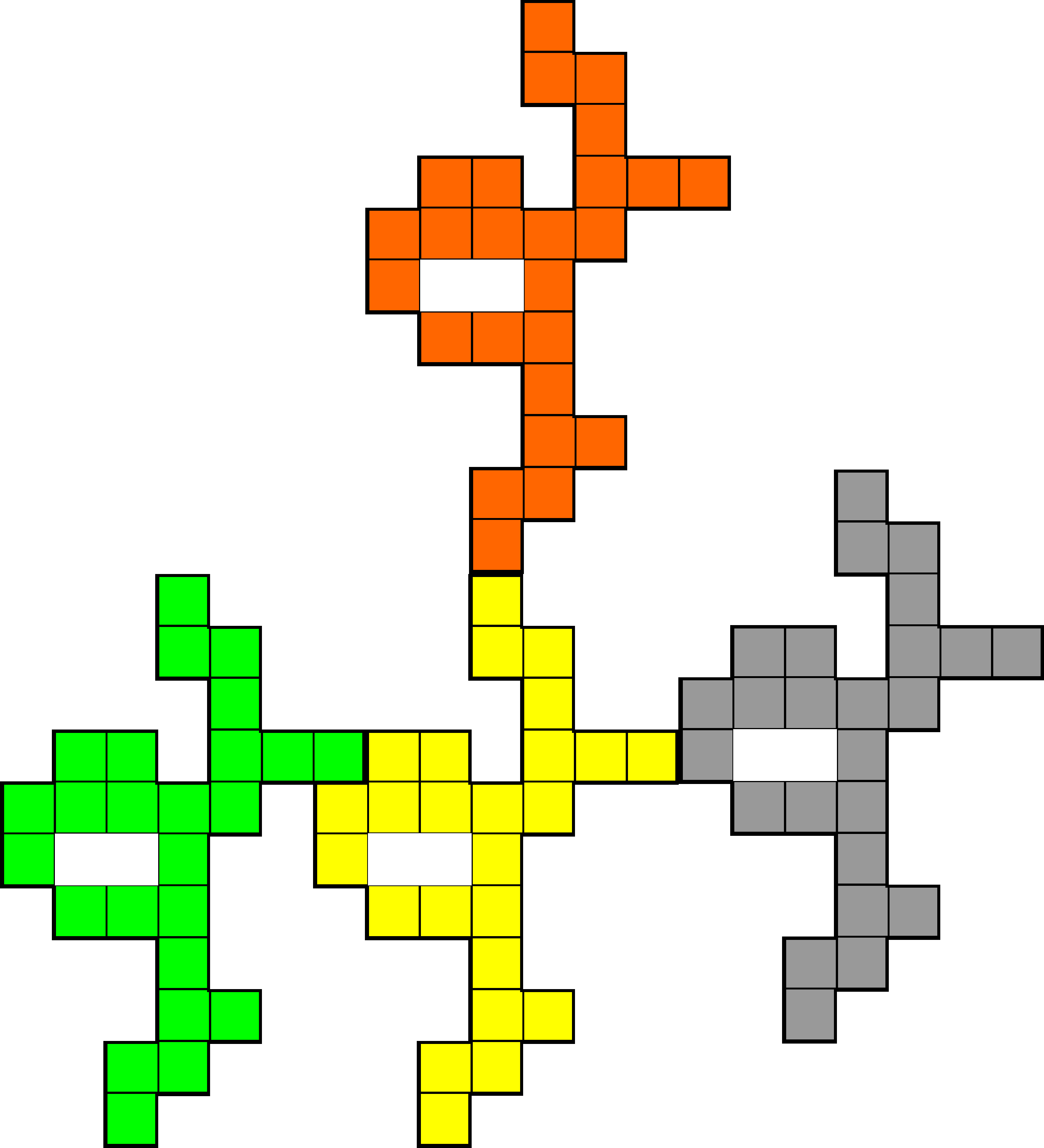}  \\
	(d) & (e) & (f) \\\hline
\end{tabular}
\caption{An example of the steps involved in placing the blocker, which prevents the aqua tile from binding to the green tile, but allows for the yellow tile to bind to the green tile and continue growth on the grid.}
\label{tbl:BitReadGenCaseE0}
\end{figure}

\begin{figure}[htp]
\centering
\begin{tabular}{| M | M |}
	\hline
	\includegraphics[scale=.13]{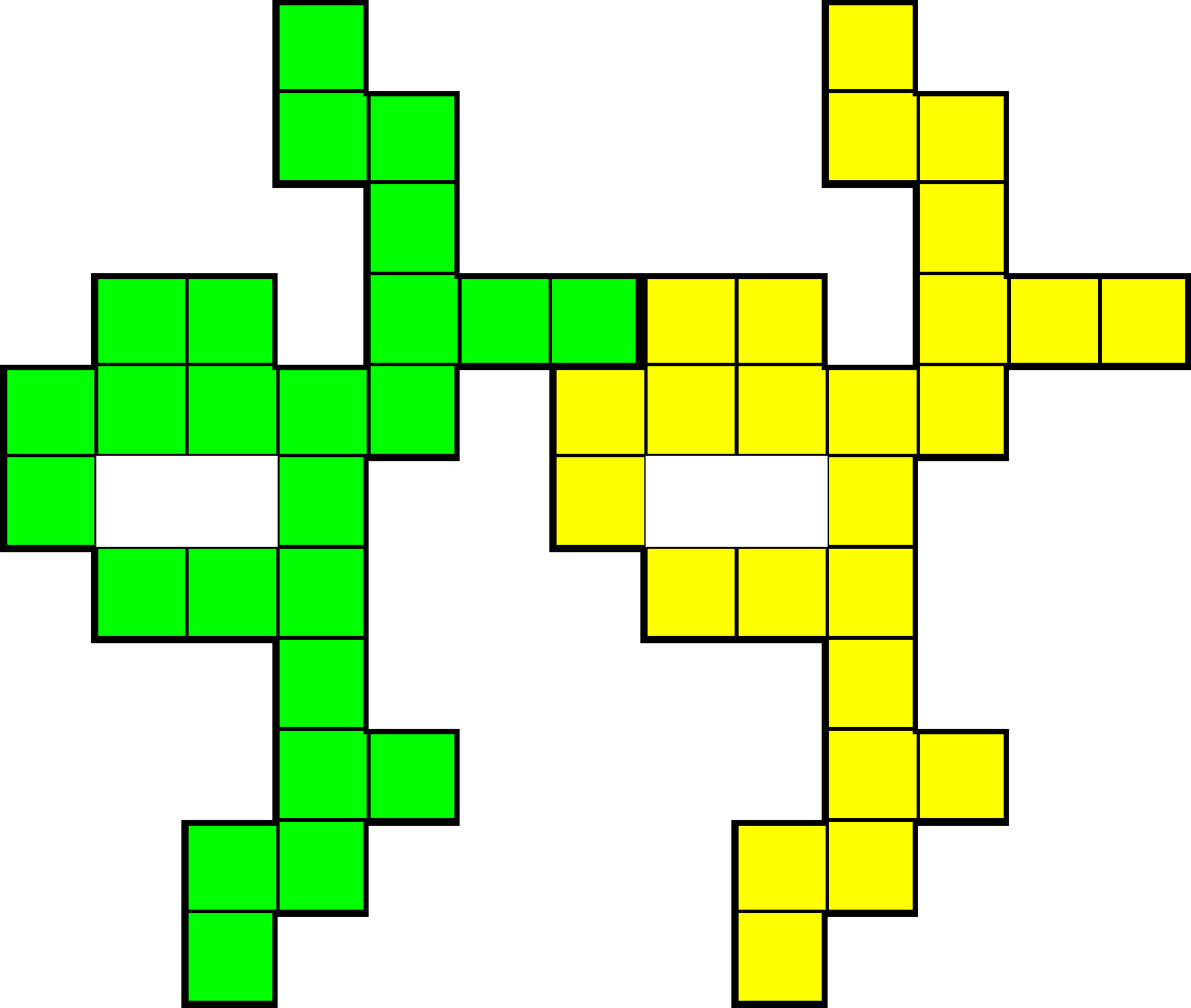}  & \includegraphics[scale=.13]{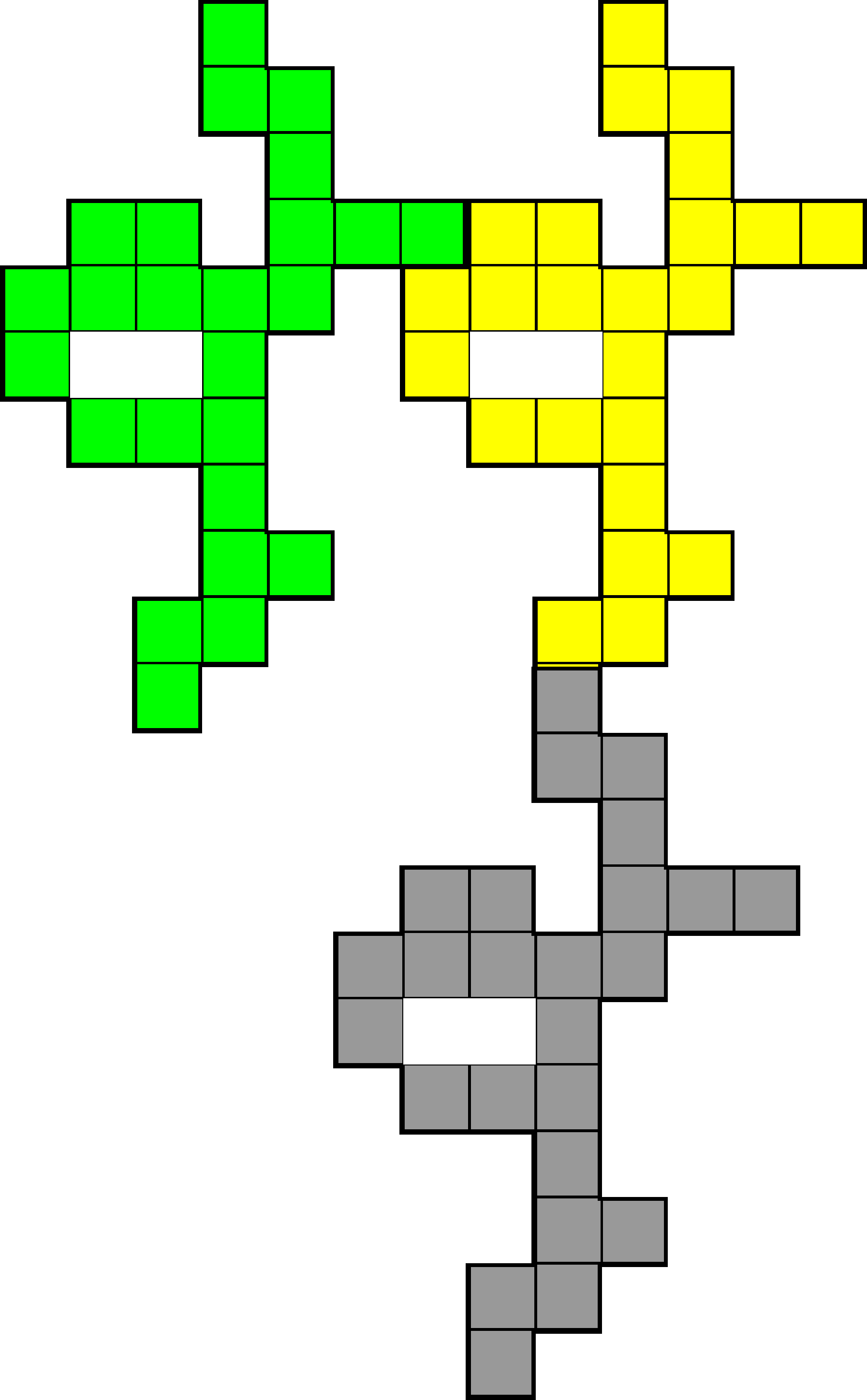} \\
	(a) & (b)  \\\hline
	\includegraphics[scale=.13]{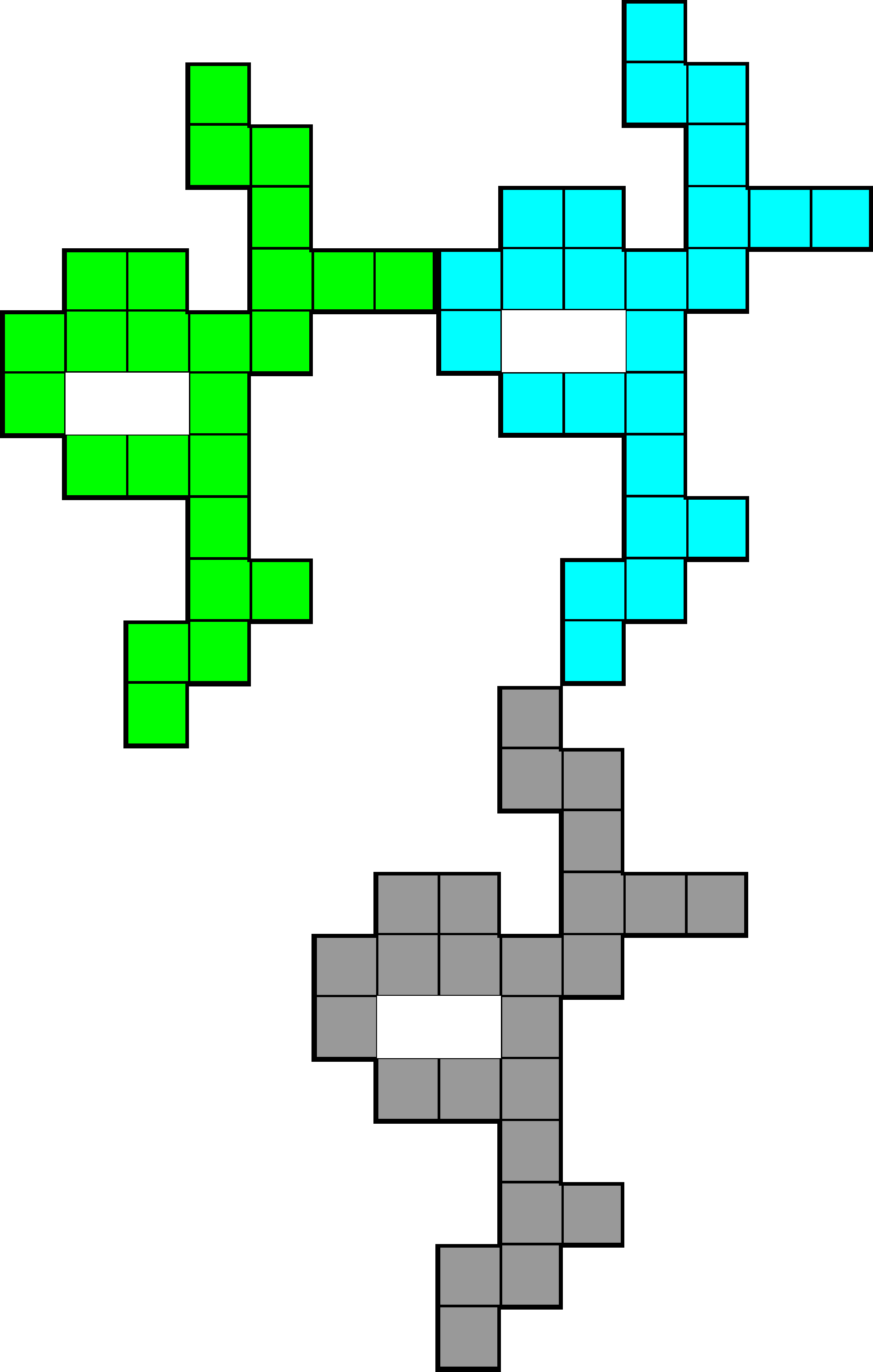}  & \includegraphics[scale=.13]{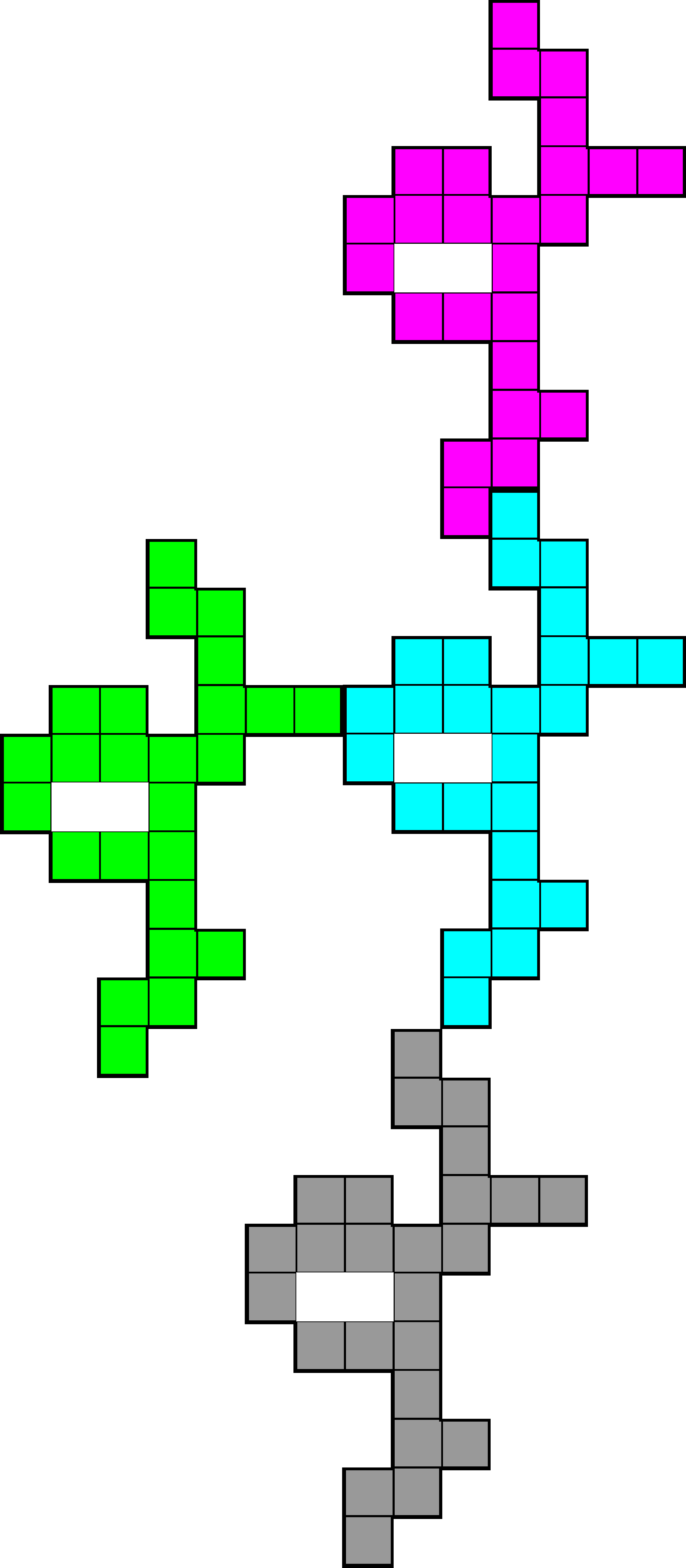} \\
	(c) & (d) \\\hline
\end{tabular}
\caption{The steps involved in placing the blocker, which prevents the yellow tile from binding to the green tile, but allows for the aqua tile to bind to the green tile and continue growth on the grid.}
\label{tbl:BitReadGenCaseE1}
\end{figure}

\begin{figure}[htp]
\begin{center}
\includegraphics[width=4.0in]{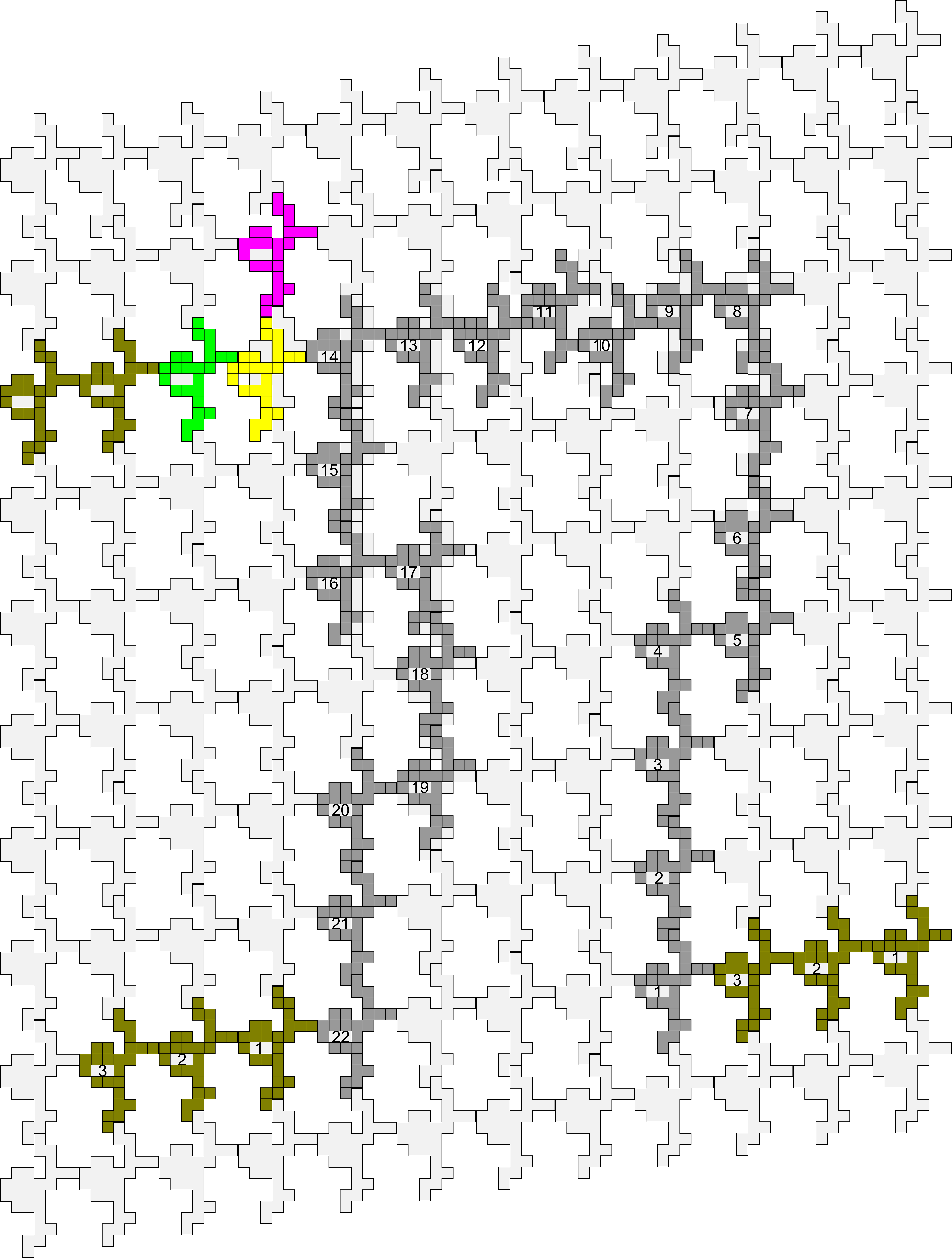}
\caption{An example showing how a ``0'' bit is read.}
\label{fig:GenComp0EX}
\end{center}
\end{figure}

\begin{figure}[htp]
\begin{center}
\includegraphics[width=4.0in]{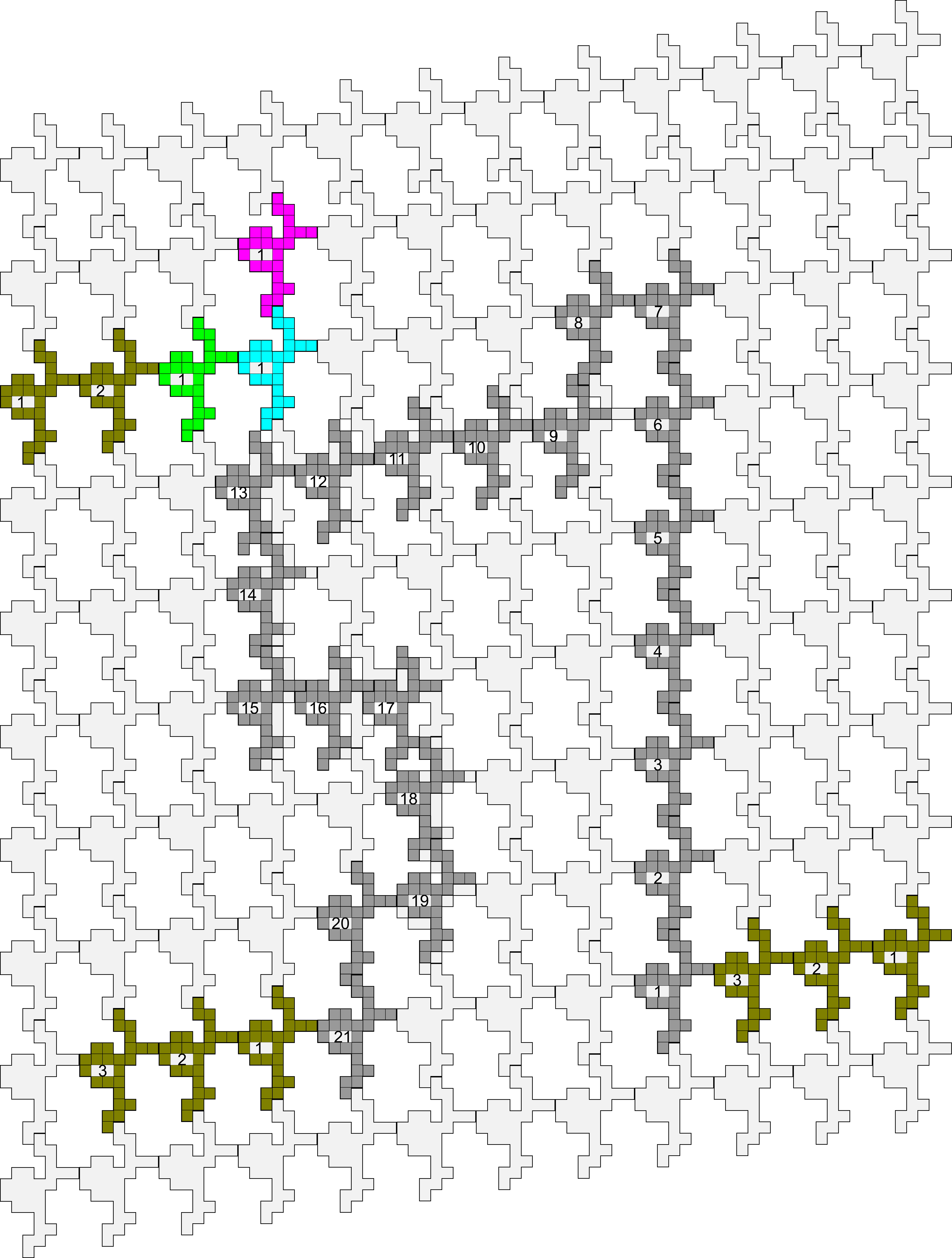}
\caption{An example showing how a ``1'' bit is read.}
\label{fig:GenComp1EX}
\end{center}
\end{figure}

\subsection{Case (4): $P$ is non-basic}
For this case, suppose that $P$ is a non-basic polyomino.  (Note that we are
not making the claim that non-basic polyominoes exist; in fact we conjecture
that they do not. However, a proof seems at least as complicated as handling
this case with an explicit construction.)

In Lemma~\ref{lem:spanningVectors}, we let $P$ be a polyomino and let $V \subset \mathbb{Z}^2$ denote the set of vectors such that $\vec{r} \in V$ provided that there exists some directed, singly seeded, single shape system $\mathcal{T} = (T,\sigma)$ with shape given by $P$ whose terminal assembly $\alpha'$ contains an $\vec{r}$-shifted polyomino tile. Informally, $V$ is the set of valid off grid shifts that can occur in a singly shaped system with shape given by $P$. Now, notice that by choosing the vector $\vec{v}$ from Lemma~\ref{lem:grids} we can assume that on grid tile attachment occurs between two polyomino tiles $P_1$ and $P_2$ when the southernmost pixel of the easternmost row of pixels of $P_1$ is horizontally adjacent to the left of the northernmost pixel of the westernmost row of pixels of $P_2$. Henceforth, we make this assumption, and observe that with this choice of grid, (i.e. this choice of $\vec{v}$) $(-1,-1)$ is contained in $V$ since, with glues placed in the necessary positions, if $P_2$ was placed on grid it could then potentially be translated one pixel left and one pixel down, causing the same two pixels of $P_1$ and $P_2$ to bind but now along north-south edges ($P_2$ being below).

\begin{lemma}\label{lem:spanningVectors}
For any polyomino $P$, if $V$ contains one of the four vectors $(\pm 1,0)$ or $(0,\pm 1)$, then $P$ is a basic polyomino.
\end{lemma}

\begin{proof}
As previously mentioned, we have chosen the grid based on $P$ such that $(-1,-1)$ is contained in $V$. Then, if $V$ contains one of the four vectors $(\pm 1,0)$ or $(0,\pm 1)$, we note $V$ contains both of the vectors $(1, 0)$ and $(0, 1)$. For example, suppose that $V$ contains the vector $(1,0)$. In this case, by Lemma~\ref{lem:genBitWrite}, $V$ also contains the vector $(0,1) = -(-1,-1) - (1,0)$. Finally, by Lemma~\ref{lem:genBitWrite}, if $V$ contains $(1,0)$ and $(0,1)$, then $V$ contains every vector in $\Z^2$ and is therefore basic.
\end{proof}

We also have the following lemma.

\begin{lemma}\label{lem:evenParityShifts}
For any polyomino $P$, let $\vec{r}=(r_1, r_2)$ where $r_1, r_2 \in Z$ be a vector. If $r_1 + r_2$ is even, then $\vec{r}\in V$.
\end{lemma}

\begin{proof}
Now, in the case where $P$ is basic, there is nothing to show since $V$ contains every integer vector.
Therefore, we assume that $P$ is non-basic.
Let $P$ be a fixed polyomino, and let $V \subset \mathbb{Z}^2$ denote the set
of vectors such that $\vec{r} \in V$ provided that there exists some directed,
singly seeded system $\mathcal{T}' = (T',\sigma')$ with shape given by $P$
whose terminal assembly $\alpha'$ contains an $\vec{r}$-shifted polyomino tile.
Suppose that $\vec{r} = (r_1, r_2)$ is of the form $r_1+r_2$ is even. Note that
$r_1+r_2$ is even if and only if $r_1$ and $r_2$ are both even or both odd.
Therefore, it suffices to show that $V$ contains every vector of the form
$\vec{r} = (r_1, r_2)$ where $r_1$ and $r_2$ are both even or both odd.

As we have previously mentioned, we may assume that $V$ contains the vector $(-1,-1)$. Now, let $\vec{v}$ and $\vec{w}$ be the two vectors
obtained from Lemma~\ref{lem:grids}, and let $P_{\vec{v}}$ be a polyomino that is the translation of $P$ by $\vec{v}$. Then, let $P'$ be the polyomino obtained by translating $P_{\vec{v}}$ by the vector $(0,1)$. Then, if $P$ has a pixel at location $(x, y)$ and $P'$ has a pixel at location $(x+1, y)$ (in other words, $P$ and $P'$ have pixels that share a common edge), then note that we can define tiles with shape $P$ and matching glues at these pixel locations to obtain a
directed, singly-seeded, single-shape system with shape given by $P$ whose terminal assembly consists of two tiles: a seed tile and a $(0,1)$-shifted tile (there is no chance of overlap since $P$ and $P_{\vec{v}}$ have no horizontal overlap of their bounding boxes and $P'$ is simply an upward shift of $P_{\vec{v}}$). Then, by Lemma~\ref{lem:spanningVectors} says that $P$ must be basic, which contradicts the assumption that $P$ is non-basic. Therefore, it must be the case that for any pixel $p$ in $P$ and pixel $p'$ in $P'$, $p$ and $p'$ do not share a horizontally adjacent edge. Hence, we can translate $P'$ by the vector $(-1,0)$ to obtain the polyomino $P''$, and moreoever, the set of locations of $P$ and the set of locations of $P''$ are disjoint. Note that $P''$ is $(-1,1)$-shifted. Now, notice that there is a pixel in $P$ at location $(x_1,y_1)$ and a pixel in $P_{\vec{v}}$  at location $(x_2,y_2)$ such that these pixels share a common edge. Therefore, the pixel in $P$ at location $(x_1,y_1)$ and the pixel in $P''$ at location $(x_1 - 1,y_1 + 1)$ must also share a common edge. Hence, we can see that $V$ contains the vector $(-1,1)$.

Then, by Lemma~\ref{lem:genBitWrite},
we know that $V$ contains the vectors $(2, 0) = (1,-1) - (-1,-1)$ and $(0,2)
= -(-1,-1) - (1,-1)$. Therefore, by using Lemma~\ref{lem:genBitWrite} again, we
can see that for any $n,m\in \Z$, $V$ contains vectors of the form $(2n, 2m) =
n\cdot (2,0) + m\cdot (0,2)$ and $(2n-1, 2m-1) = (2n, 2m) + (-1,-1)$.
\end{proof}

The schematic diagram for the system we construct in this section which contains the bit-reading gadget will be the same as that shown in Section~\ref{sec:3over}.

Let $t_a$ and $t_y$ be the first tiles to be placed that are a part of the aqua path of tiles and yellow path of tiles respectively.  We construct the light grey tiles, green tile, $t_a$ tile, and $t_y$ tile in the same manner as they are constructed in Section~\ref{sec:3over}.

\subsubsection*{Case (4) Bit-Writer Construction} \label{sec:4B0}
We now describe the construction of the bit-writer subassemblies of the bit-reading gadget by giving a description of the placement of the blocker in relation to the position of the green tile.  Let $x_w$ represent the x-coordinate of the westernmost pixel on the south perimeter of the polyomino $P$ and let $x_e$ represent the x-coordinate of the easternmost pixel on the north perimeter of $P$.

\paragraph{$0$-blocker placement.}
We first place the green tile with the $t_y$ tile attached (recall the $t_y$ tile is $(-1,-1)$-shifted) in the plane, and place a $0$-blocker tile in an on grid position directly below the $t_y$ tile so that the bounding rectangles of the $0$-blocker and the $t_y$ tile overlap (see Figure~\ref{fig:NonCB01}).    We now consider two cases for the placement of the $0$-blocker: 1) $x_w-x_e$ is even and 2) $x_w-x_e$ is odd.

In the case that $x_w-x_e$ is even we can simply translate the $0$-blocker by the vector $(x_w-x_e,0)$ so that the easternmost pixel on the north perimeter of the $0$-blocker overlaps the westernmost pixel on the south perimeter of $t_y$.  It follows from Lemma~\ref{lem:evenParityShifts} that the blocker can be shifted by such an amount.  This case is shown in Figure~\ref{fig:NonCB02}.

In the case that $x_w-x_e$ is odd we translate the $0$-blocker by $(x_w-x_e,1)$ (shown in Figure~\ref{fig:NonCB13}).  Once again, Lemma~\ref{lem:evenParityShifts} allows us to know that such a shift is valid because $(x_w-x_e)+1$ is even (since $x_w-x_e$ is odd).  We also know that the $0$-blocker and $t_y$ tile must have pixels which overlap.  Indeed, for the sake of contradiction, suppose that the pixel to the south of the easternmost pixel on the northern perimeter of the $0$-blocker did not overlap the westernmost pixel on the souther perimeter of $t_y$.  Then that means that the easternmost pixel on the northern perimeter of the $1$-blocker is attached either via its east or west (in order for $P$ to be connected) which implies it has a neighbor to its east or west.  But, this means that we can find a system with a $(1,0)$-shift, which by Lemma~\ref{lem:spanningVectors} contradicts the assumption that $P$ is basic.  In addition, we make the observation that the $t_a$ tile can still be placed since the easternmost pixel on the northern perimeter of the $0$-blocker will lie to the west of the westernmost pixel on the south perimeter of $t_a$.  Consequently, in the presence of the $0$-blocker, $t_a$ can be placed.
\begin{figure}[htp]
\centering
  \subfloat[][The green tile is on grid and the yellow tile is $(-1,-1)$-shifted.]{%
        \label{fig:NonCB01}%
        \includegraphics[width=2.0in]{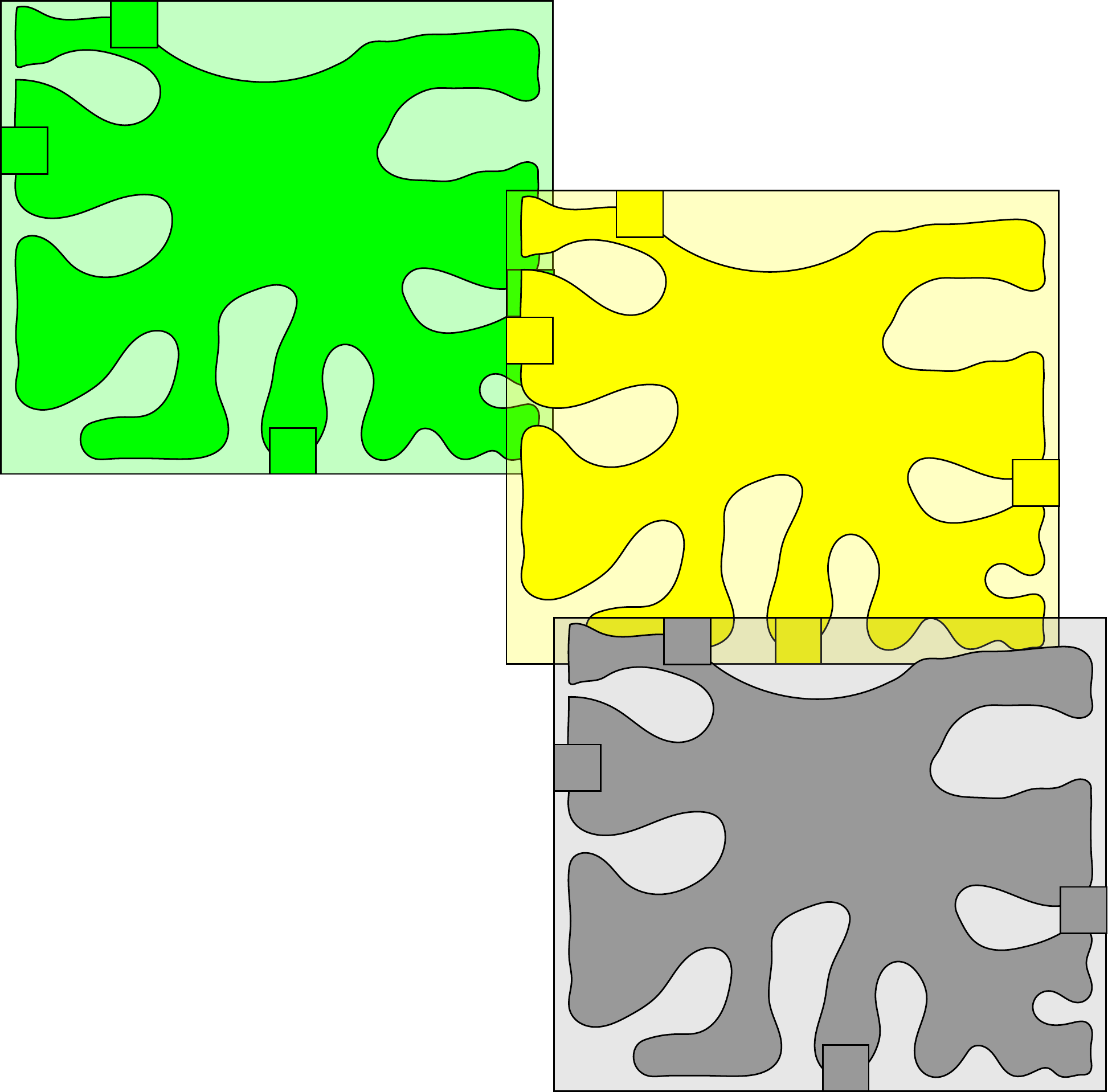}
        }%
        \quad
  \subfloat[][The $0$-blocker tile is placed in this configuration in the event that $x_w-x_e$ is even.]{%
        \label{fig:NonCB02}%
        \includegraphics[width=2.0in]{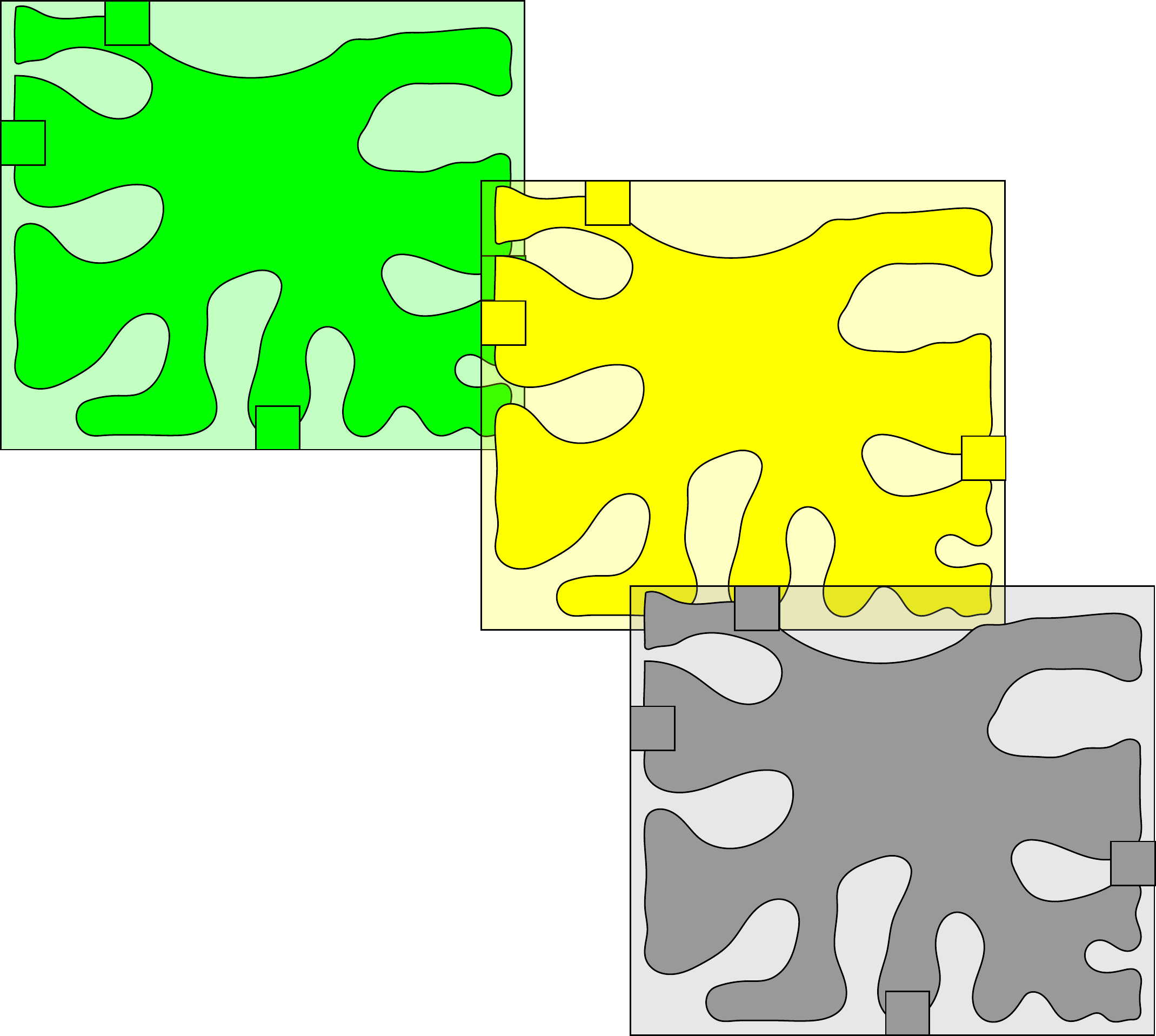}
        }%
        \quad

  \subfloat[][The $0$-blocker tile is placed in this position in the event that $x_w-x_e$ is odd.]{%
  \label{fig:NonCB13}%
  \includegraphics[width=2.0in]{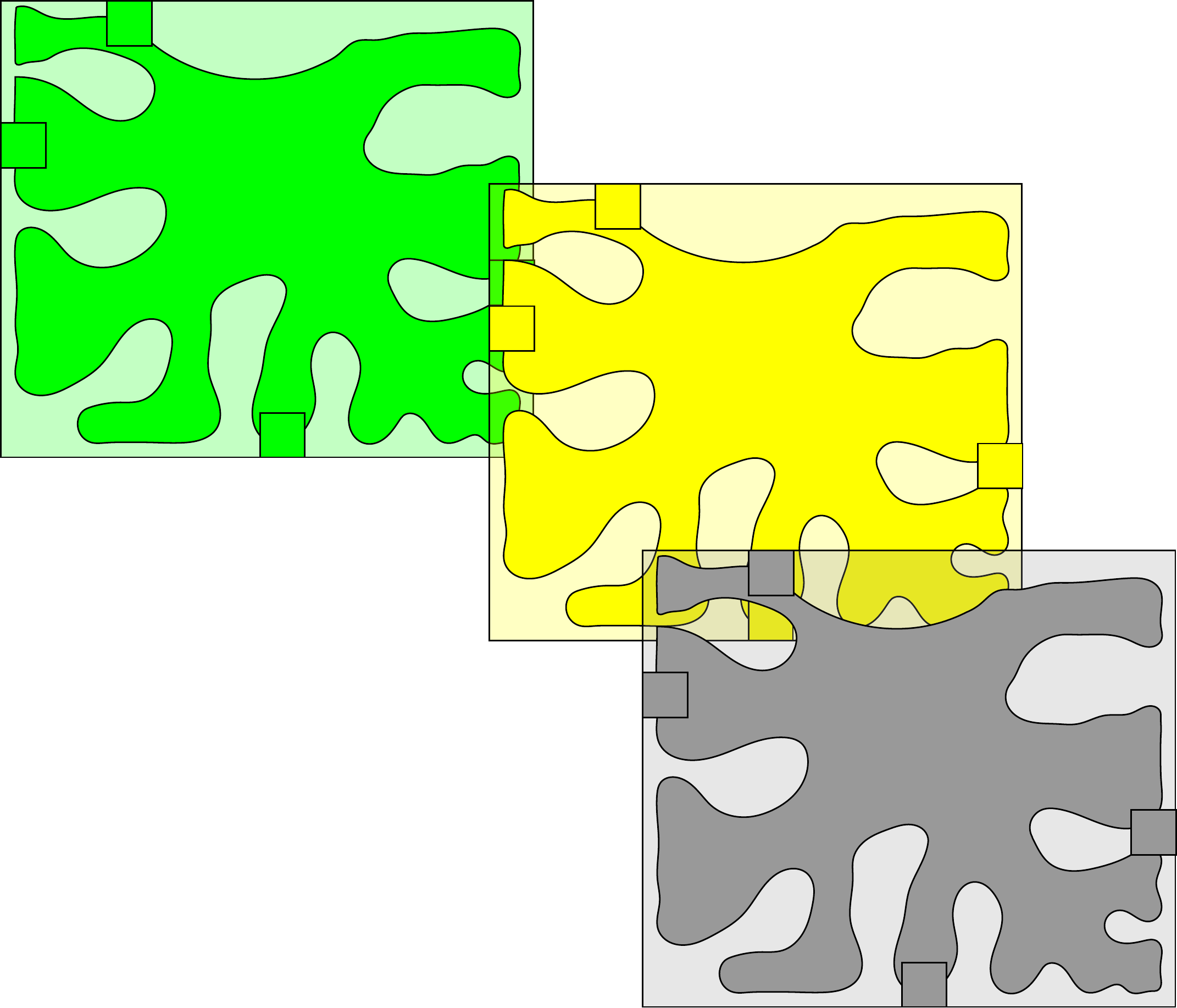}
  }%
  \quad
  \caption{Placement of the $0$-blocker, which blocks the yellow (i.e. $0$-reader) path.}
  \label{fig:NonCB0}
\end{figure}
\paragraph{$1$-blocker placement.}
The $1$-blocker is positioned by first laying down the green tile, the $t_a$ tile and the $1$-blocker in on grid positions (shown in Figure~\ref{fig:NonCB11}) and then shifting the $1$-blocker by the vector $(-2,0)$.  The translation of the $1$-blocker yields a configuration as shown in Figure~\ref{fig:NonCB12}.  Once again, because of connectivity and the assumption that $P$ is non-basic, it must be the case that the tile $t_a$ and the $1$-blocker have overlapping pixels as indicated in the figure by the red square which represents overlapping pixels.
\begin{figure}[htp]
\centering
  \subfloat[][The green tile and aqua tile are both on grid.]{%
        \label{fig:NonCB11}%
        \includegraphics[width=2.3in]{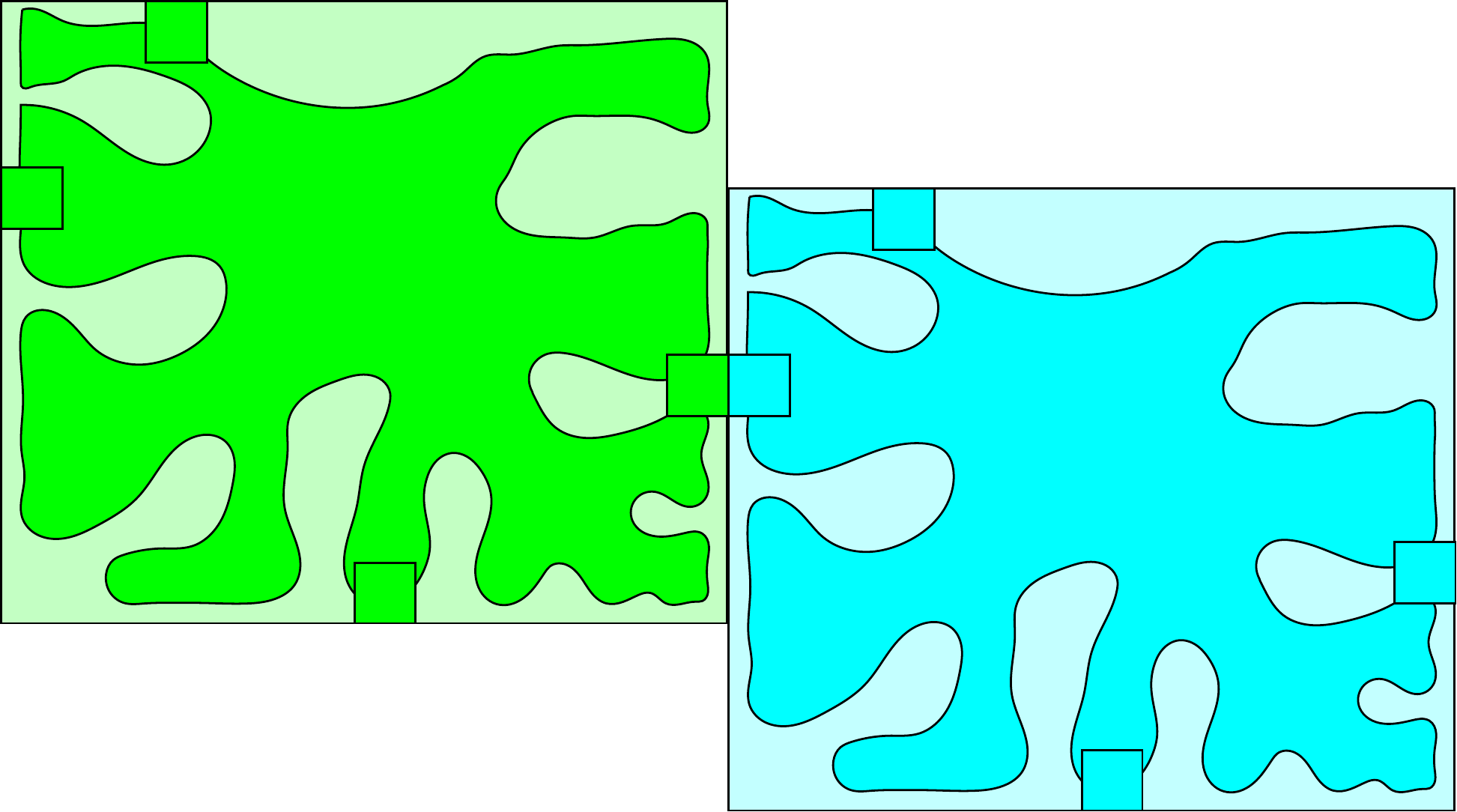}
        }%
        \quad
  \subfloat[][The $1$-blocker is translated by the vector $(-2,0)$ so that it overlaps a pixel in $t_a$ in the position indicated by the red square.]{%
        \label{fig:NonCB12}%
        \includegraphics[width=2.3in]{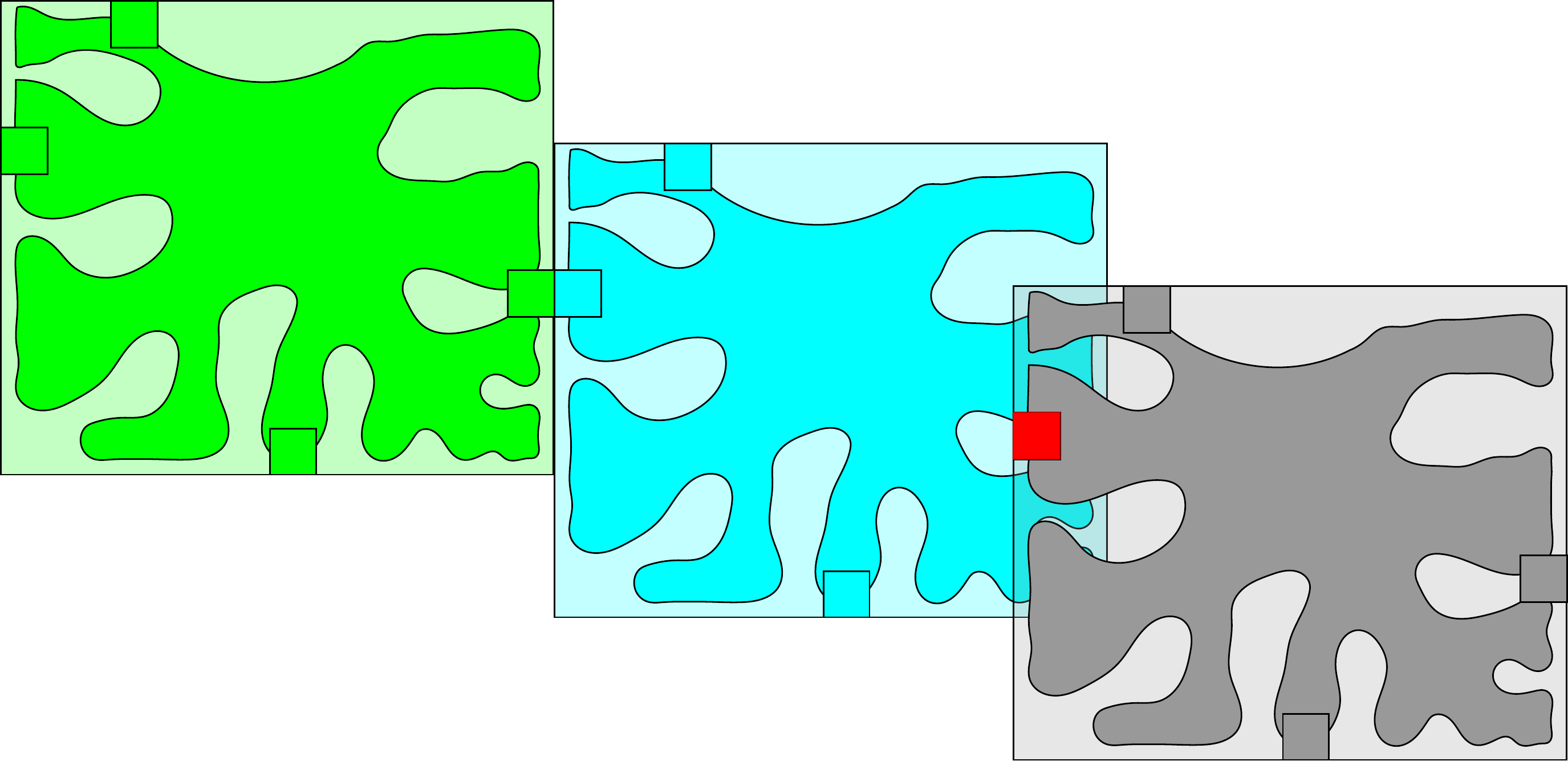}
        }%
        \quad
  \caption{Placement of the $1$-blocker, which blocks the aqua (i.e. $1$-reader) path.}
  \label{fig:NonCB1}
\end{figure}

We can now construct the bit writer gadgets in the same way as we did in Section~\ref{sec:3BW} (see Figure~\ref{tbl:3write}).

\subsubsection*{Case (4) Bit-Reader Construction}

Figure~\ref{fig:NonCY} shows how the glues are placed on tiles in the yellow path so that the yellow tiles assemble in the presence of a $1$-blocker.  Figure~\ref{fig:NonCY1} shows how the $t_y$ tile lies relative to the $1$-blocker.  We place a glue on the second tile in the path of yellow tiles so that it attaches to the north of $t_y$ and lies on grid with respect to $t_y$.

We now claim that in the current configuration there are not any pixels that overlap.  Notice that none of the yellow tiles can have a pixel which overlaps a pixel on the green tile.  Indeed, for the sake of contradiction, suppose that one of the yellow tiles did overlap with the green tile.  Recall that both of the yellow tiles are $(-1,-1)$-shifted polyominoes.  It must be the case that the yellow tile is overlapping a perimeter pixel of the green tile as shown in Figure~\ref{fig:NonGY1}.  Then it follows that we can then shift the yellow tile by the vector $(1,0)$ as shown in Figure~\ref{fig:NonGY2}.  But this means, that the shifted yellow tile can be attached to the green tile to form a $(0,-1)$-shifted polyomino which by Lemma~\ref{lem:spanningVectors} contradicts our assumption that $P$ is non-basic.

To see that neither of the yellow tiles have a pixel which overlaps a pixel on the $1$-blocker tile, once again, suppose for the sake of contradiction that a pixel in one of the yellow tiles did overlap a pixel in the $1$-blocker tile.  Recall that the $1$-blocker tile is a $(-2,0)$-shifted polyomino, and observe that the pixel which overlaps must lie on the west perimeter of the $1$-blocker as shown in Figure~\ref{fig:NonGG1}.  Then it follows that we can shift the $1$-blocker by the vector $(1,0)$ as shown in Figure~\ref{fig:NonGG2}.  Again, this means that the shifted $1$-blocker tile can attached to the yellow tile to form a $(-1,0)$-shifted polyomino which by Lemma~\ref{lem:spanningVectors} contradicts our assumption that $P$ is non-basic.

Similarly, we place glues so that the second tile in the aqua path attaches to the north of tile $t_a$ such that it lies in an on grid position. The argument that there is not any overlap is similar to that argued above.
\begin{figure}[htp]
\centering
  \subfloat[][The green tile is on grid and the yellow tile is $(-1,-1)$-shifted, grey tile shifted by $(-2,0)$.]{%
        \label{fig:NonCY1}%
        \includegraphics[width=2.0in]{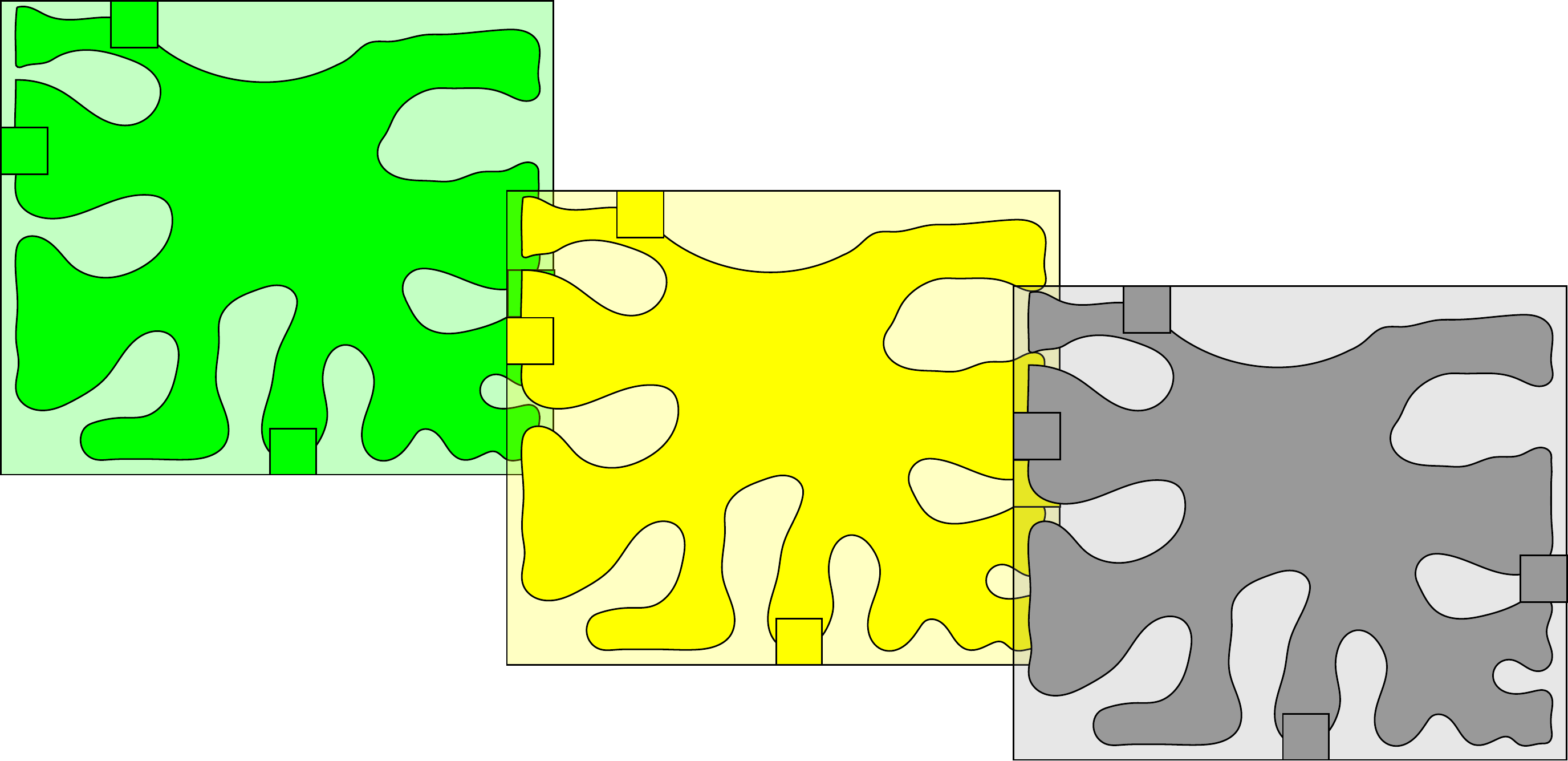}
        }%
        \quad
  \subfloat[][The second tile in the yellow path attaches to the first by an on grid north attachment.]{%
        \label{fig:NonCY2}%
        \includegraphics[width=2.0in]{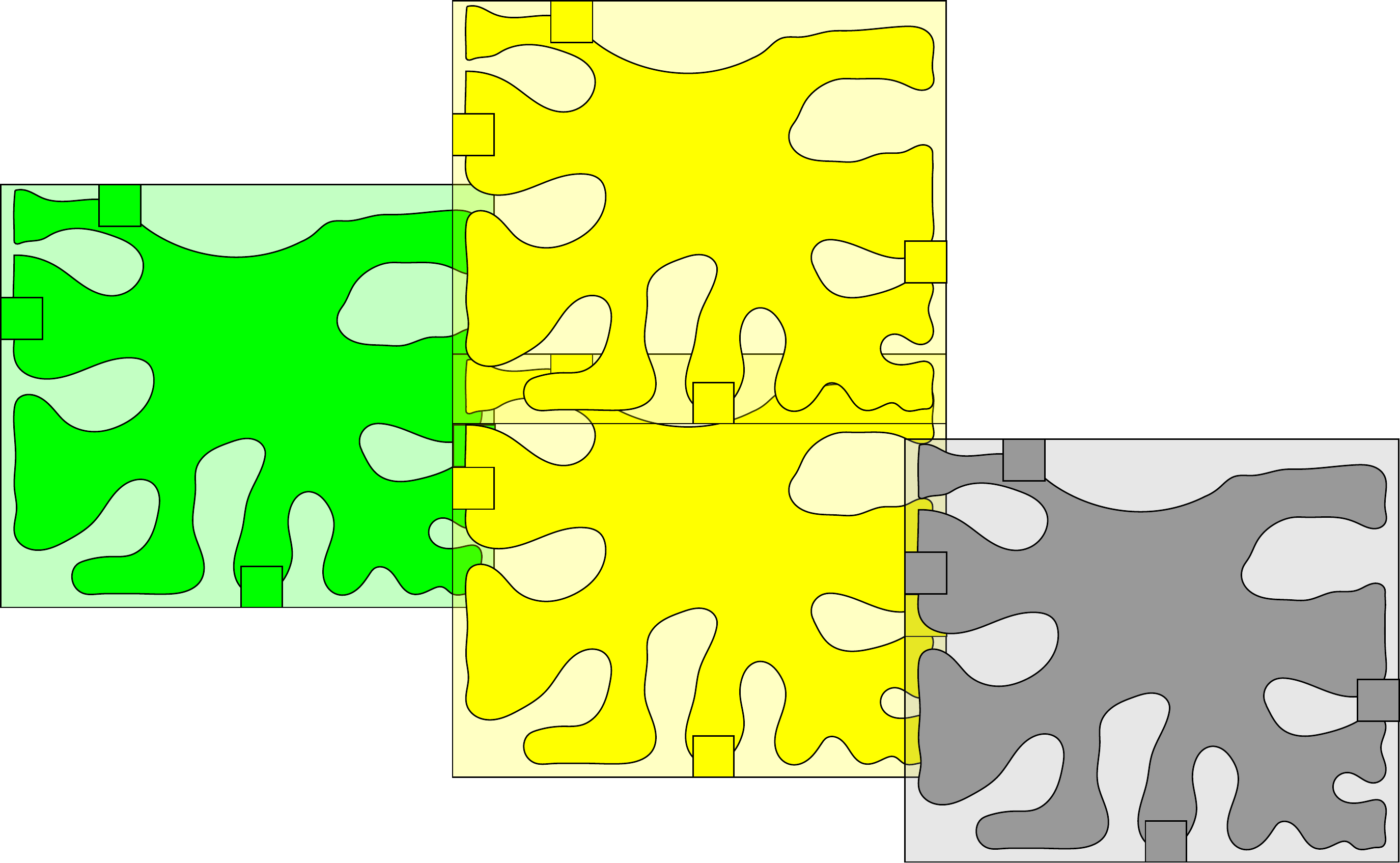}
        }%
        \quad

  \subfloat[][Shift the yellow path of tiles back onto the grid by placing a $(1,1)$-shifted polyomino.]{%
  \label{fig:NonCY3}%
  \includegraphics[width=2.0in]{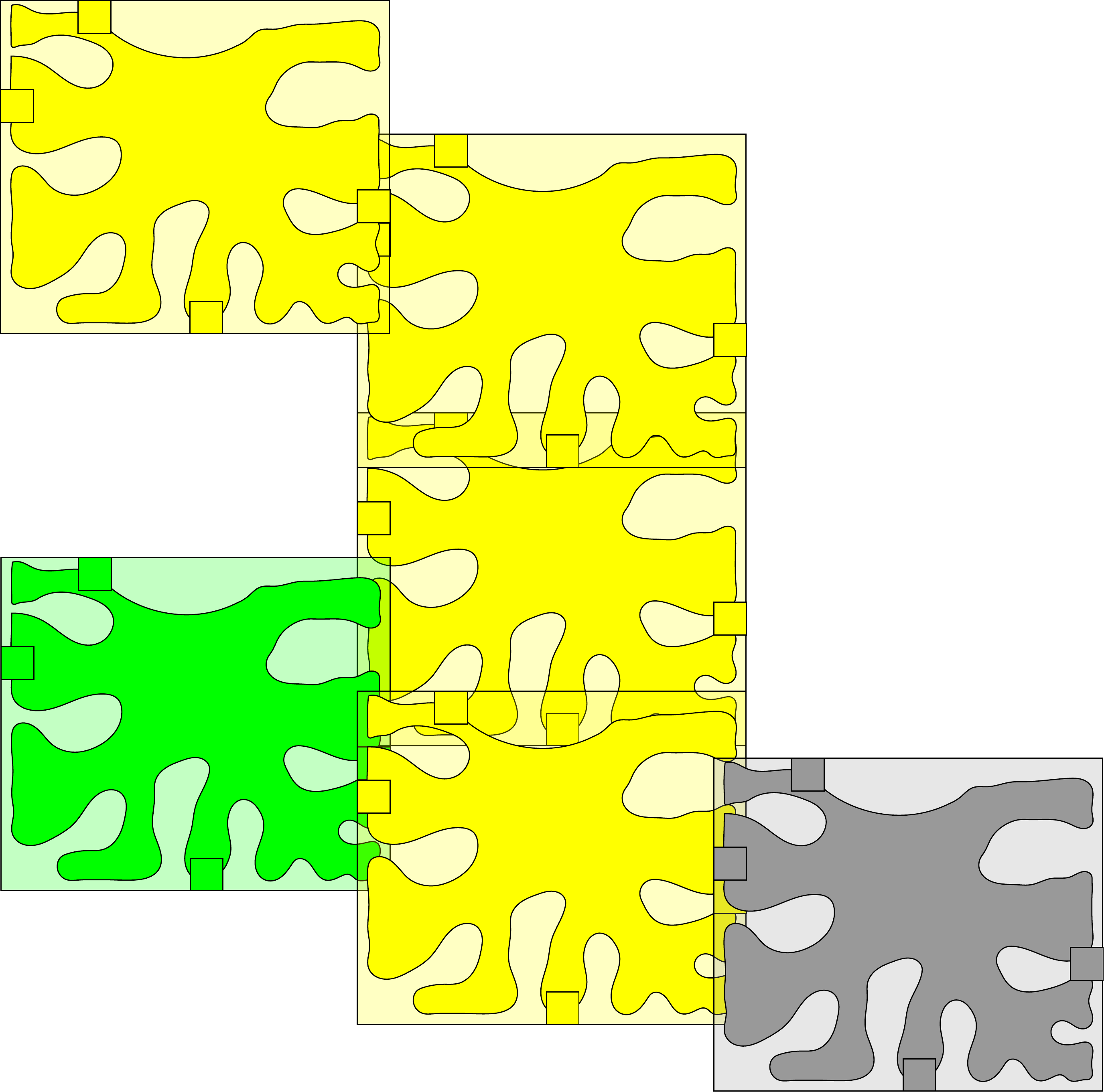}
  }%
  \quad
  \caption{Growth of the yellow path in the presence of a $1$-blocker.}
  \label{fig:NonCY}
\end{figure}

\begin{figure}[htp]
\centering
  \subfloat[][Suppose that the green tile and the yellow tile overlap pixels at the pixel shown in red.]{%
        \label{fig:NonGY1}%
        \includegraphics[width=2.3in]{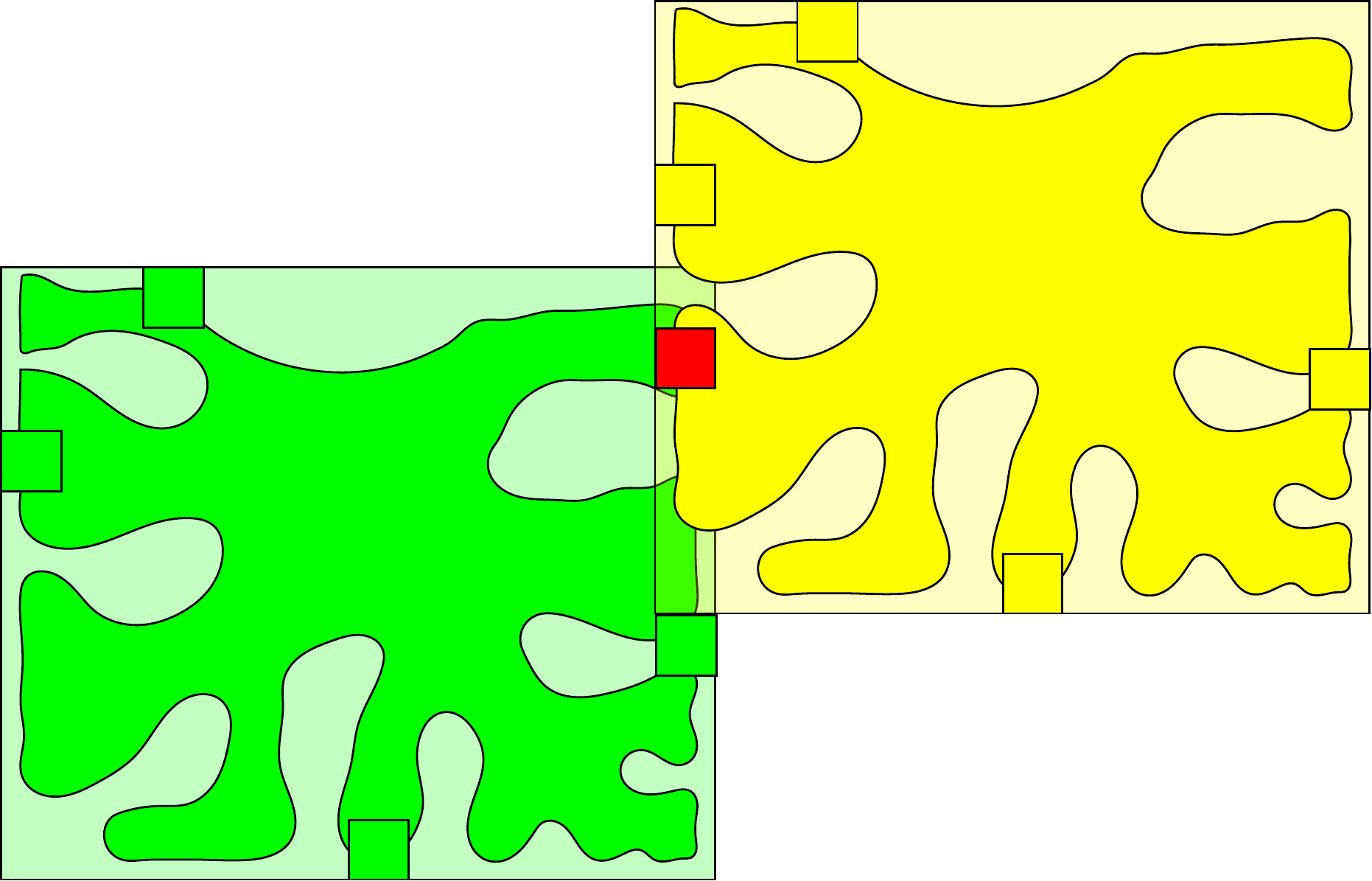}
        }%
        \quad
  \subfloat[][Then since the yellow tile is shifted by $(-1,-1)$, we can attach the yellow tile to the green tile as shown forming a $(0,-1)$-shifted polyomino, a contradiction.]{%
        \label{fig:NonGY2}%
        \includegraphics[width=2.3in]{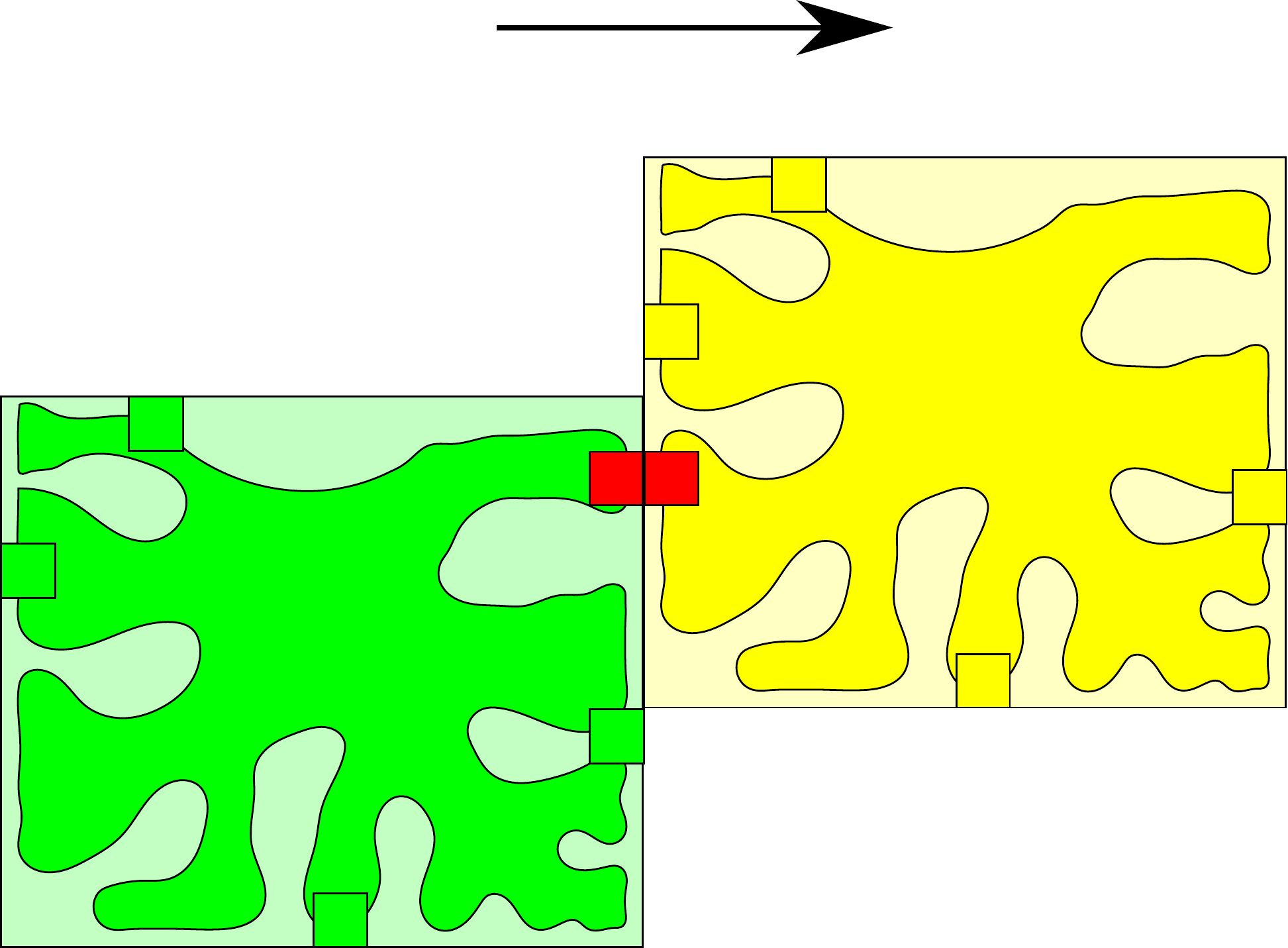}
        }%
        \quad
  \caption{The yellow tiles and the green tile cannot overlap.}
  \label{fig:NonGY}
\end{figure}

\begin{figure}[htp]
\centering
  \subfloat[][Suppose that the yellow tile and the grey tile overlap pixels at the pixel shown in red.]{%
        \label{fig:NonGG1}%
        \includegraphics[width=2.3in]{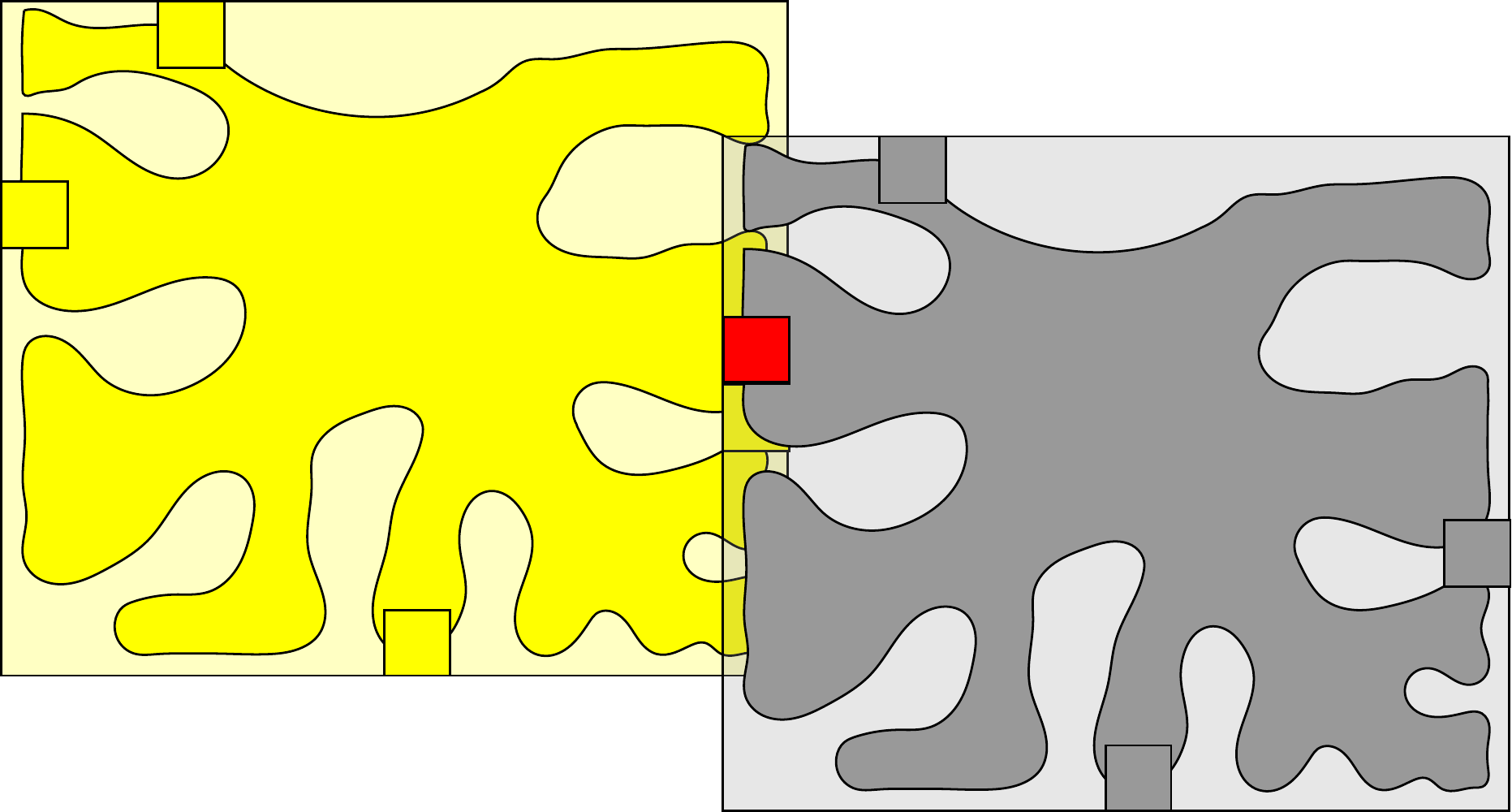}
        }%
        \quad
  \subfloat[][Then since the yellow tile was shifted by $(-2,0)$, we can find a $(-1,0)$-shifted polyomino by attaching the grey tile as shown, a contradiction.]{%
        \label{fig:NonGG2}%
        \includegraphics[width=2.3in]{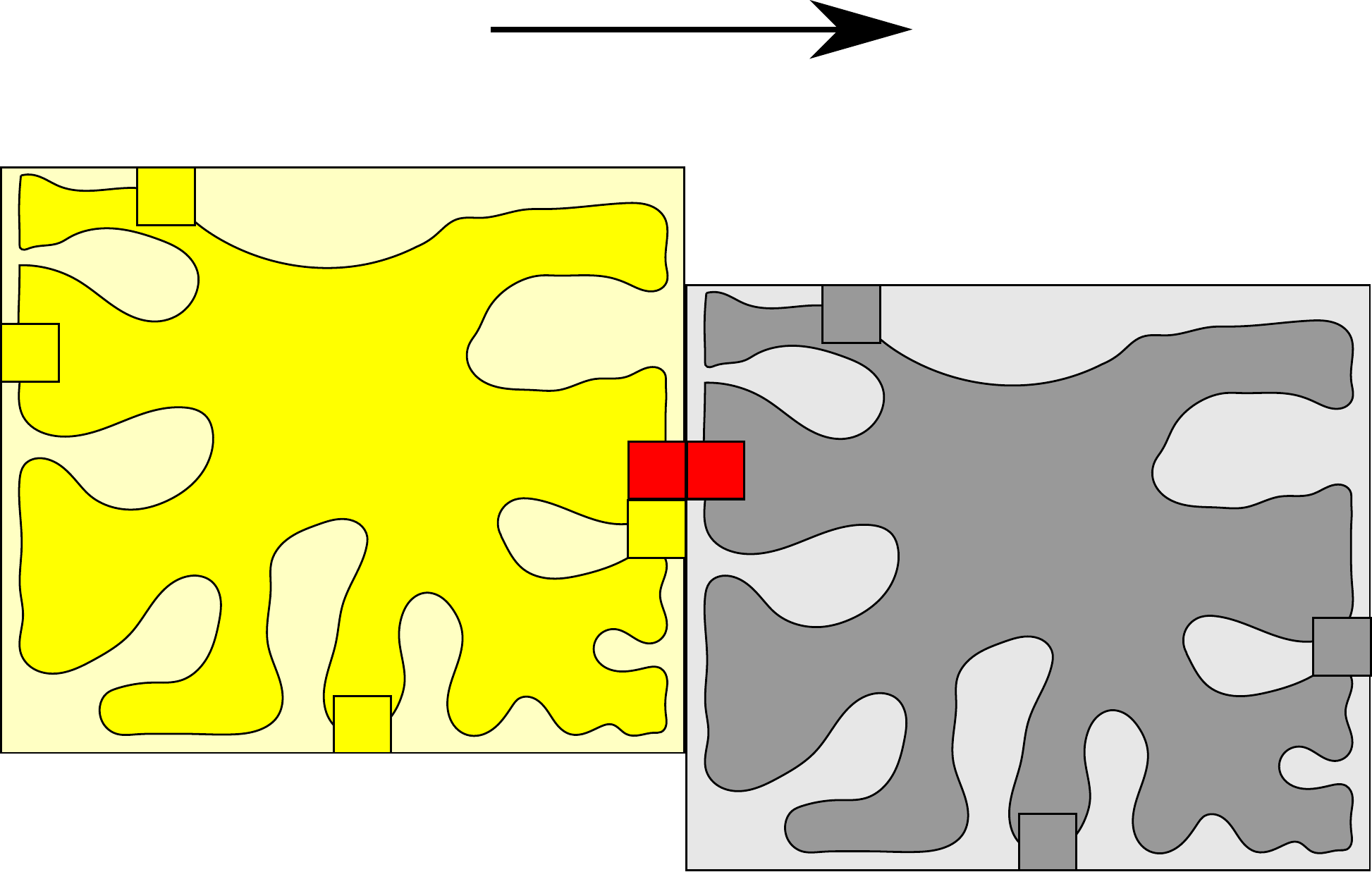}
        }%
        \quad
  \caption{The yellow tiles and the grey tile cannot overlap.}
  \label{fig:NonGG}
\end{figure}

\subsubsection*{Case (4) Right-to-Left Bit Reading Gadget Construction}
Since we can always construct a system with a non-basic polyomino that has an assembly that contains $(-1,-1)$-shifted and $(1,-1)$-shifted polyominoes, we can replay the mirrored version of the arguments above to construct a system which contains a right-to-left bit-reading gadget construction.  That is, in the right-to-left bit reading gadget, the yellow tile attaches to the green tile so that it is a $(1,-1)$-shifted polyomino.  Also, instead of placing the $1$-blocker as above (where it was placed as a $(-2,0)$ polyomino), we place the $1$-blocker so that it is $(2,0)$-shifted.  Likewise, we place the $0$-blocker as in Section~\ref{sec:4B0} with the words easternmost and westernmost swapped.

\subsubsection*{Case (4) Correctness of the Bit-Reading Gadget}
The correctness of the bit-reading gadget for this proof follows the same argument as presented in Section~\ref{sec:3PC}.

\section{Computationally Limited Systems}\label{sec:limited-systems}

In this section, we provide a set of results which suggest that some systems of polyominoes are incapable of universal computation by showing that they are either unable to utilize bit-reading gadgets (which are fundamental features of all known computational tile assembly systems), or that they can be simulated by standard aTAM temperature 1 systems (which are conjectured to be incapable of universal computation), and are thus no more powerful than them.

\subsection{Monomino and Domino Systems Cannot Read Bits}

\begin{theorem}\label{thm:no-atam-bit-readers}
There exists no temperature 1 monomino system (a.k.a. aTAM temperature-$1$ system) $\mathcal{T}$ such that a bit-reading gadget exists for $\mathcal{T}$.
\end{theorem}

\begin{proof}

We prove Theorem~\ref{thm:no-atam-bit-readers} by contradiction.  Therefore, assume that there exists an aTAM system $\mathcal{T} = (T,\sigma,1)$ such that $\mathcal{T}$ has a bit-reading gadget.  (Without loss of generality, assume that the bit-reading gadget reads from right to left and has the same orientation as in Definition~\ref{def:bit-reader}.)  Let $(t_x,t_y)$ be the coordinate of the tile $t$ from which the bit-reading paths originate (recall that it is the same coordinate regardless of whether or not a $0$ or a $1$ is to be read from $\alpha_0$ or $\alpha_1$, respectively).  By Definition~\ref{def:bit-reader}, it must be the case that if $\alpha_0$ is the only portion of $\alpha$ in the first quadrant to the left of $t$, then at least one path can grow from $t$ to eventually place a tile from $T_0$ at $x=0$ (without placing a tile below $y=0$ or to the right of $(t_x-1)$.  We will define the set $P_0$ as the set of all such paths which can possibly grow.  Analogously, we will define the set of paths, $P_1$, as those which can grow in the presence of $\alpha_1$ and place a tile of a type in $T_1$ at $x=0$.  Note that by Definition~\ref{def:bit-reader}, neither $P_0$ nor $P_1$ can be empty.

Since all paths in $P_0$ and $P_1$ begin growth from $t$ at $(t_x,t_y)$ and must always be to the left of $t$, at least the first tile of each must be placed in location $(t_x-1,y)$.  We now consider a system where $t$ is placed at $(t_x,t_y)$ and is the only tile in the plane (i.e. neither $\alpha_0$ nor $\alpha_1$ exist to potentially block paths), and will inspect all paths in $P_0$ and $P_1$ in parallel.  If all paths follow exactly the same sequence of locations (i.e. they overlap completely) all the way to the first location where they place a tile at $x=0$, we will select one that places a tile from $T_0$ as its first at $x=0$ and call this path $p_0$, and one which places a tile from $T_1$ as its first at $x=0$ and call it $p_1$.  This situation will then be handled in Case (1) below.  In the case where all paths do not occupy the exact same locations, then there must be one or more locations where paths branch.  Since all paths begin from the same location, we move along them from $t$ in parallel, one tile at a time, until the first location where some path, or subset of paths, diverge.  At this point, we continue following only the path(s) which take the clockwise-most branch.  We continue in this manner, taking only clockwise-most branches and discarding other paths, until reaching the location of the first tile at $x = 0$.  (Figure~\ref{fig:bit-reader-right-turns} shows an example of this process.)  We now check to see which type(s) of tiles can be placed there, based on the path(s) which we are still following.  We again note that by Definition~\ref{def:bit-reader}, some path must make it this far, and must place a tile of a type either in $T_0$ or $T_1$ there.  If there is more than one path remaining, since they have all followed exactly the same sequence of locations, we randomly select one and call it $p'$.  If there is only one, call it $p'$.  Without loss of generality, assume that $p'$ can place a tile from $T_0$ at that location.  This puts us in Case (2) below.

\begin{figure}[ht]
\centering
	\begin{minipage}[t]{0.45\linewidth}
	\centering
    \includegraphics[width=3.1in]{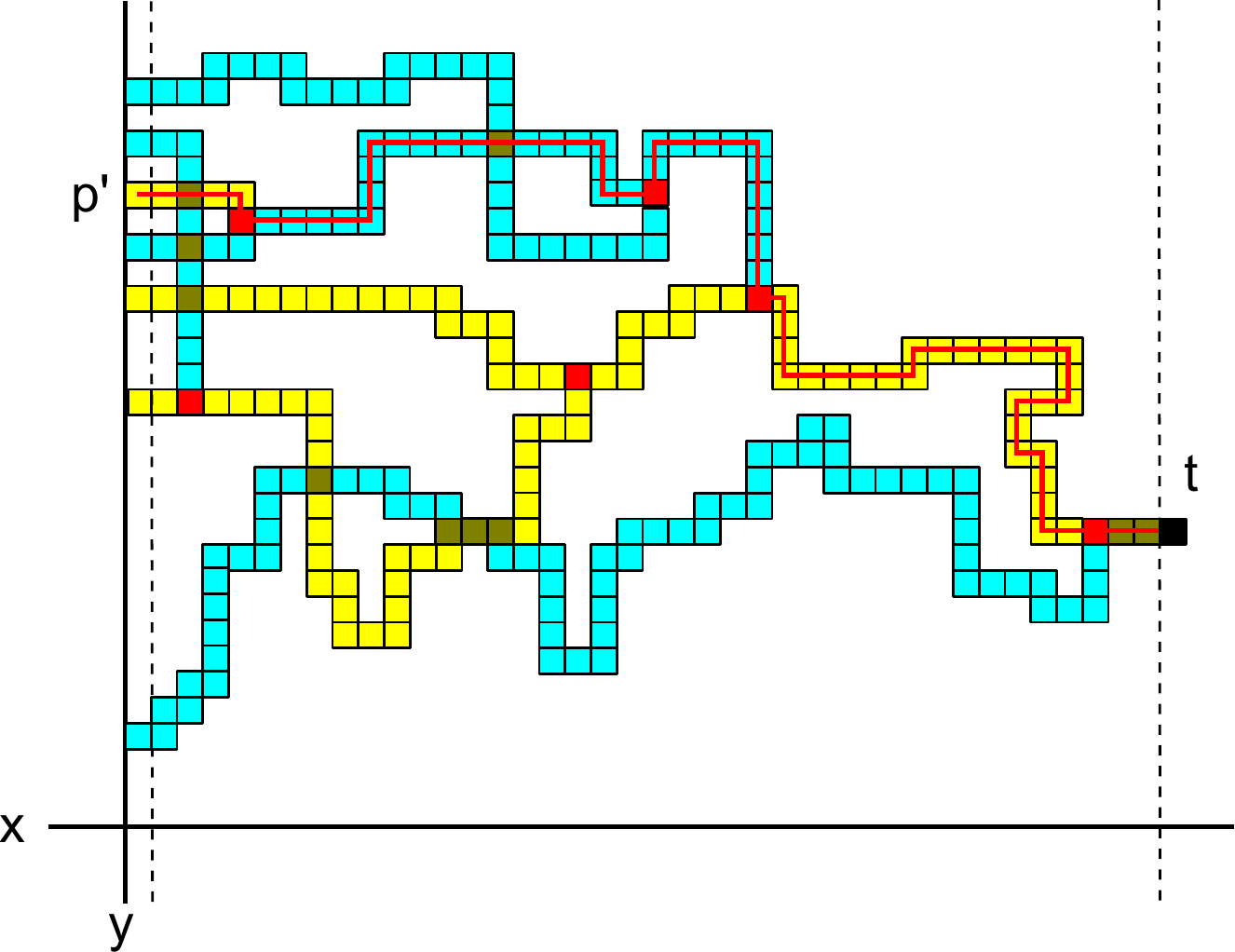}
    \caption{Example sets $P_0$ and $P_1$, with $p'$ traced with a red line.  Red squares represent branching points of paths, gold squares represent overlapping points of different branches.}
    \label{fig:bit-reader-right-turns}
	\end{minipage}
\quad\quad\quad\quad
	\begin{minipage}[t]{0.45\linewidth}
	\centering
    \includegraphics[width=3.1in]{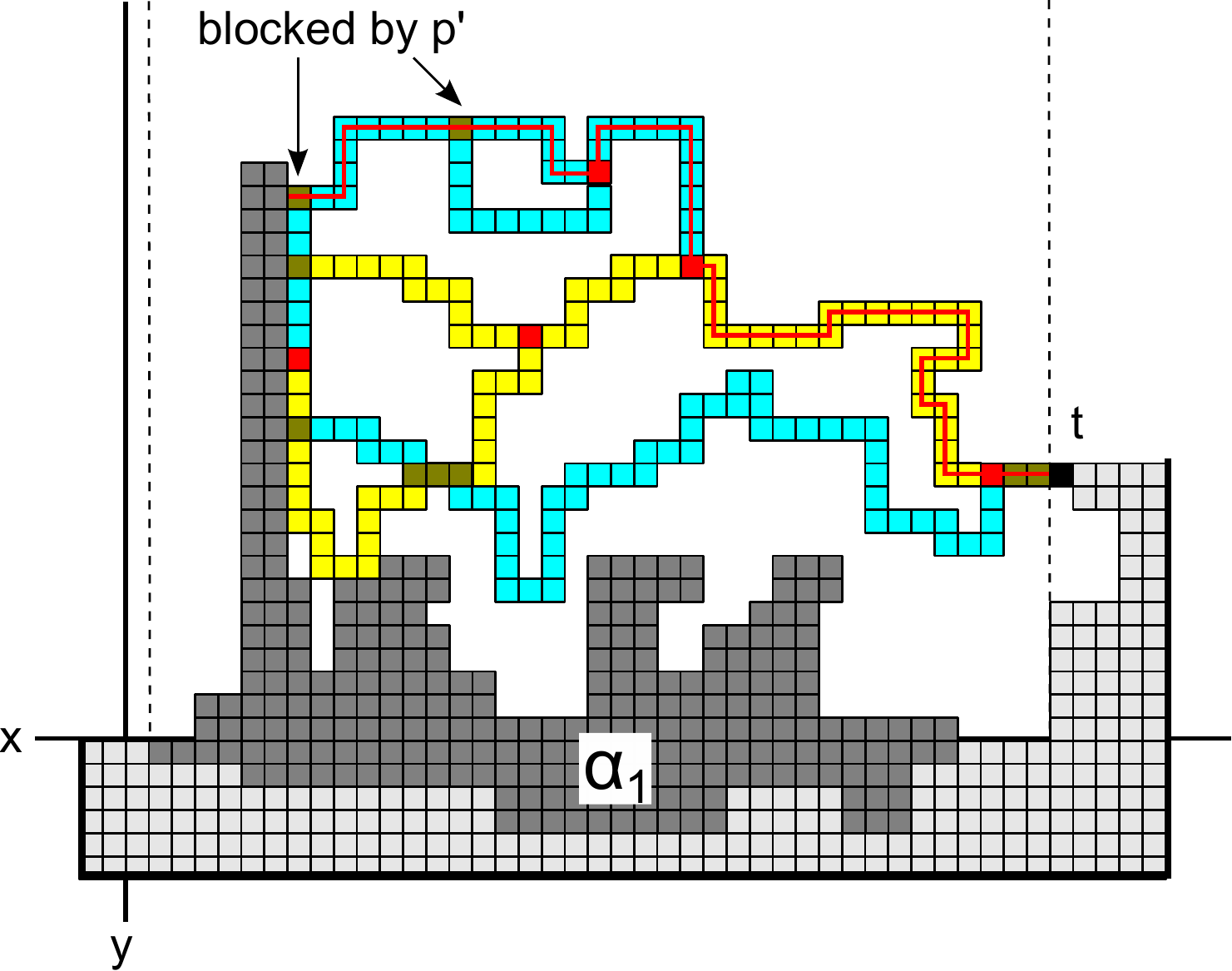}
    \caption{An example of the growth of $p'$ (traced with a red line) blocked by $\alpha_1$.  By first letting as much of $p'$ grow as possible, it is guaranteed that all other paths must be blocked from reaching $x=0$.}
	\label{fig:bit-reader-blocked}
	\end{minipage}
\end{figure}

Case (1)  Paths $p_0$ and $p_1$ occupy the exact same locations through all tile positions and their placement of their first tiles at $x=0$.  Also, there are no other paths which can grow from $t$, so, since by Definition~\ref{def:bit-reader} some path must be able to complete growth in the presence of $\alpha_0$, either must be able to.  Therefore, we place $\alpha_0$ appropriately and select an assembly sequence in which $p_1$ grows, placing a tile from $T_1$ as its first at $x=0$.  This is a contradiction, and thus Case (1) cannot be true.

Case (2)  We now consider the scenario where $\alpha_1$ has been placed as the bit-writer according to Definition~\ref{def:bit-reader}, and with $t$ at $(t_x,t_y)$.  Note that path $p'$ must now always, in any valid assembly sequence, be prevented from growing to $x=0$ since it places a tile from $T_0$ at $x=0$, while some path from $T_1$ must always succeed.  We use the geometry of the paths of $T_1$ and path $p'$ to analyze possible assembly sequences.

We create a (valid) assembly sequence which attempts to first grow only $p'$ from $t$ (i.e. it places no tiles from any other branch).  If $p'$ reaches $x=0$, then this is not a valid bit-reader and thus a contradiction.  Therefore, $p'$ must not be able to reach $x=0$, and since the only way to stop it is for some location along $p'$ to be already occupied by a tile, then some tile of $\alpha_1$ must occupy such a location.  This means that we can extend our assembly sequence to include the placement of every tile along $p'$ up to the first tile of $p'$ occupied by $\alpha_1$, and note that by the definition of a connected path of unit square tiles in the grid graph, that means that some tile of $p'$ has a side adjacent to some tile of $\alpha_1$.  At this point, we can allow any paths from $P_1$ to attempt to grow.  However, by our choice of $p'$, as the ``outermost'' path due to always taking the clockwise-most branches, any path in $P_1$ (and also any other path in $P_0$ for that matter) must be surrounded in the plane by $p'$, $\alpha_1$, and the lines $y=0$ and $x=t_x$ (which they are not allowed to grow beyond). (An example can be seen in Figure~\ref{fig:bit-reader-blocked}.)  Therefore, no path from $P_1$ can grow to a location where $x=0$ without colliding with a previously placed tile or violating the constraints of Definition~\ref{def:bit-reader}.  (This situation is analogous to a prematurely aborted computation which terminates in the middle of computational step.)  This is a contradiction that this is a bit-reader, and thus none must exist.
\end{proof}

\begin{theorem}\label{thm:no-duple-bit-readers}
There exists no single shape polyomino tile system $\Gamma = (T,\sigma,1)$ where all tiles of $T$ consist of either two unit squares arranged in a vertical bar, or of two unit squares arranged in a horizontal bar (i.e. vertical or horizontal duples), such that a bit-reading gadget exists for $\Gamma$.
\end{theorem}

\begin{proof}
The proof of Theorem~\ref{thm:no-duple-bit-readers} is nearly identical to that of Theorem~\ref{thm:no-atam-bit-readers}.  Without loss of generality, we prove the impossibility of a bit-reader with horizontally oriented duples.  The only differing point to consider between the proof for squares vs. duples is in the analysis of Case (2) where the claim is made that if $p'$ is blocked by $\alpha_1$, some tile of $p'$ must have a side adjacent to some tile of $\alpha_1$.  In the case of duples, as can be seen in Figure~\ref{fig:blocked-duples}, there are also the possibilities that tiles of $p'$ and $\alpha_1$ are diagonally adjacent or separated by a gap of a single unit square.  However, neither of these possibilities can allow for duples of any path in $P_1$ to pass through, and they therefore remain blocked, thus again proving that no bit-reader must exist.

\begin{figure}[htp]
\begin{center}
\includegraphics[width=2.0in]{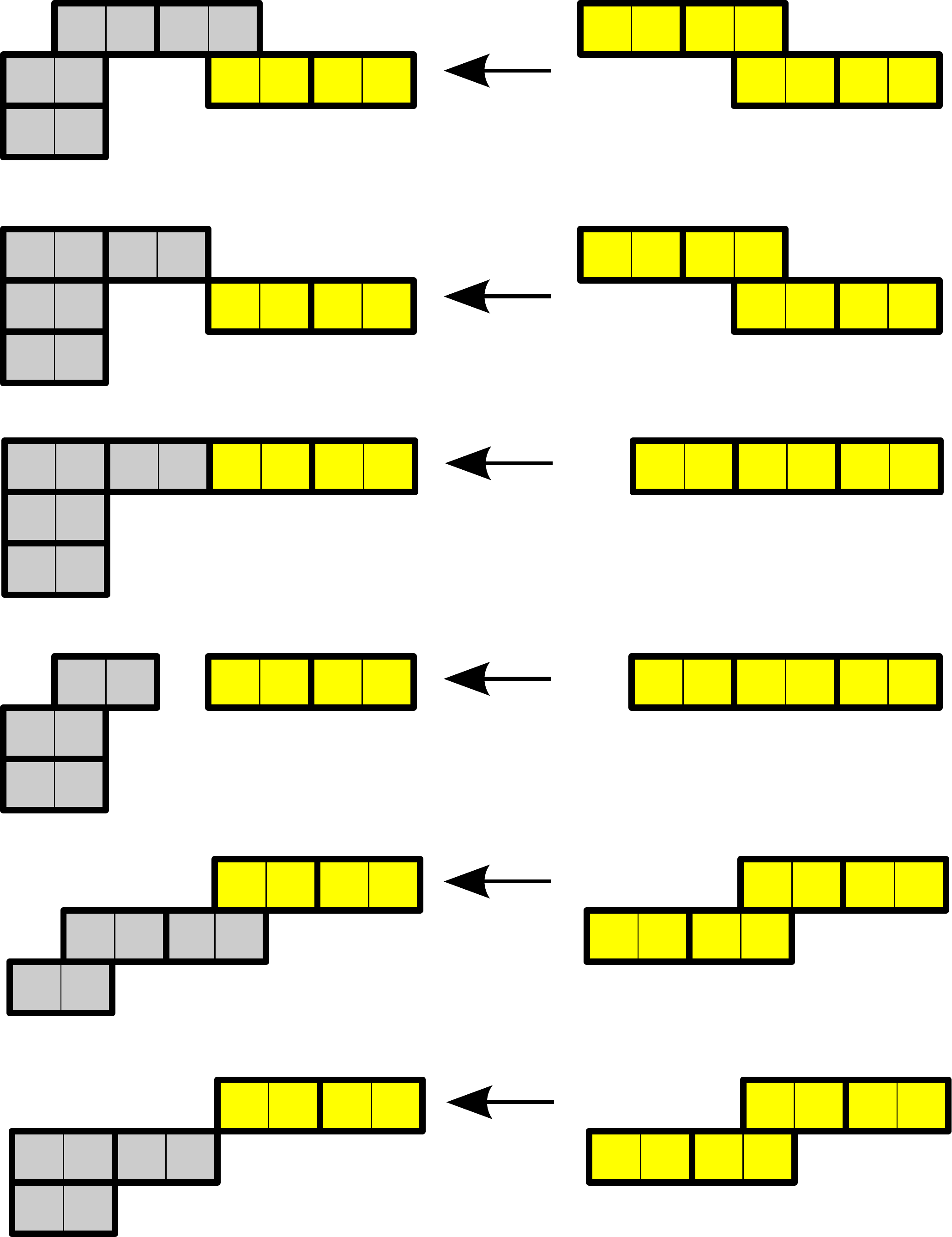}
\caption{(Right) Partial paths of duples growing from right to left, (Left) portions of those paths being blocked by grey tiles.  Any possible way of blocking a path of duples results in either (1) no gap between the path and the blocking assembly, or (2) a single unit square gap, which no path of only horizontally oriented duples can pass through.}
\label{fig:blocked-duples}
\end{center}
\end{figure}
\end{proof}

It is interesting to note that by the addition of a single extra square to a duple, creating a $3\times1$ polyomino, it is possible to create gaps between blocking assemblies and blocked paths which allow another path to pass through.  This is because the gap can be diagonally displaced from the last tile of the blocked path.  An example can be seen in Figure~\ref{fig:comp_1by_case}.

\subsection{Scale-Free Simulation}

We now provide a definition which captures what it means for one polyomino system to simulate another.  This definition is meant to capture a very simple notion of simulation in which the simulating system follows the assembly sequences of the simulated system via a simple mapping of tile types and with no scale factor (as opposed to more complex notions of simulation which allow for scaled simulations such as in \cite{IUSA,2HAMIU,Duples,IUNeedsCoop}, for instance).

\begin{definition}[Scale-free simulation] A tile system $\mathcal{T} = (T,\sigma,\tau)$ is said to \emph{scale-free simulate} a tile system $\mathcal{T}' = (T',\sigma',\tau')$ if there exists a surjective function $f:T \rightarrow T'$ and a bijective function $M: \prodasm{T} \rightarrow \prodasm{T'}$ such that the following properties hold.
\begin{enumerate}
\item For $A,A' \in \prodasm{T}$, $A \rightarrow^{\calT}_1 A'$ via the addition of tile $t \in T$ if and only if $M(A) \rightarrow^{\mathcal{T'}}_1 M(A')$ via the addition of a tile in the preimage $f^{-1}[t]$. 
\item Any sequence $(A_1=\sigma, \ldots, A_k)$ is an assembly sequence for $\mathcal{T}$ if and only if $(M(\sigma)=\sigma', \ldots, M(A_k))$ is an assembly sequence for $\mathcal{T'}$.
\end{enumerate}

\end{definition}

\subsection{Polyominoes with Limited Glue Positions}
\begin{figure}[htp]
\begin{center}
\vspace{-20pt}
\includegraphics[width=4in]{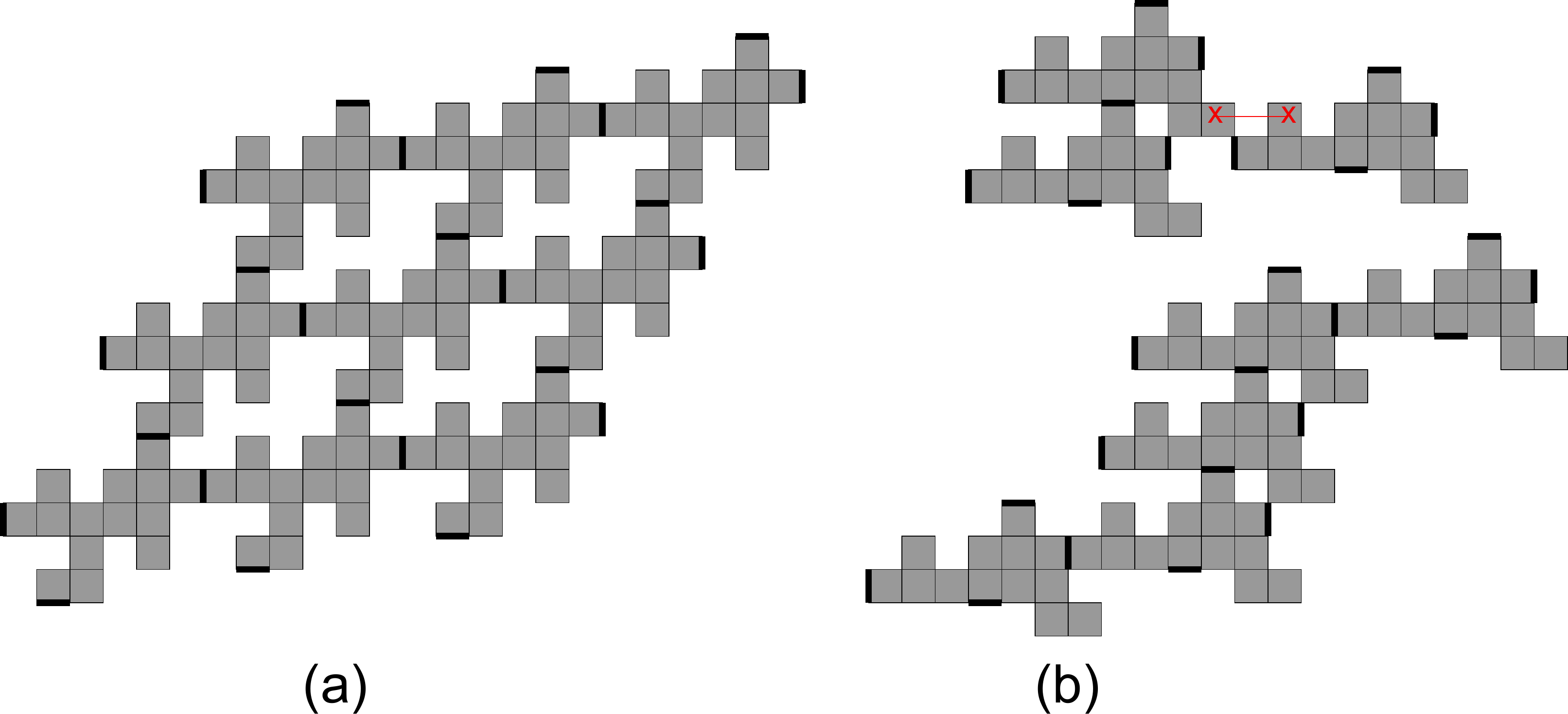}
\caption{(a) Case (1) considers the scenario in which a 4-position limited polyomino with uniquely paired glues does not not have overlapping neighboring positions, or an overlapping diagonal neighbor, and can therefore form a 2D lattice.  (b) In case (2), neighbor positions are mutually exclusive.  In this case the assembly is linear and is scale-free simulated by a linear monomino tile system with a quadratic increase in tile types to account for the non-determinism introduced by the choice of glue positions.}
\label{fig:4gluesPicture}
\vspace{-20pt}
\end{center}
\end{figure}

In this section we analyze the potential of polyomino systems to compute if the number of distinct positions on the polyominoes at which glues may be placed is bounded.  We show that any polyomino system which utilizes 3 or fewer distinct glue locations, or a system that uses 4 glue locations but adheres to a ``unique pairing" constraint, is scale-free simulated by a temperature 1 aTAM system (Theorems \ref{thm:3positions}, \ref{thm:4positionsUniquePair}), and is thus very likely to be incapable of universal computation.  On the other hand, we show that with only 4 glue positions and no unique pairing restriction, universal computation is possible (Theorem \ref{thm:4positionsUniversal}).

\begin{definition}[$c$-position limited] Consider a set $T$ of polyomino tiles all of some shape polyomino $P$.  Consider the subset $S$ of all edges of $P$ such that some $t \in T$ places a glue label on a side in $S$.  We say that $T$ has \emph{glue locations} $S$.  If $c \ge |S|$, we say that $T$ is \emph{$c$-position limited}.  Further, any single shape polyomino system $\calT =(T,\sigma,\tau)$ is said to be $c$-position limited if $T$ is $c$-position limited.
\end{definition}

\begin{definition}[uniquely paired]
A polyomino system with glue locations $S$ is said to be uniquely paired if for each $s \in S$, there is a unique $s' \in S$ such that glues in position $s$ can only bind with glues in position $s'$.
\end{definition}

Monomino systems, for example, are uniquely paired as the north face glue position only binds with the south face position, and the east position only binds with the west position.

\begin{theorem}\label{thm:3positions}
Any $3$-position limited polyomino system $\Gamma = (T,\sigma,1)$ is scale-free simulated by a monomino tile system (a.k.a. a temperature 1 aTAM system).
\end{theorem}

\begin{proof}
If the system is $2$-position limited, a monomino system that replaces each $t\in T$ with a linear east/west glue monomino tile (i.e. a tile which only has glues its west and east sides) will do the trick.  (Note that Lemma 3 of \cite{OneTile} implies that if the glue positions on a polyomino are sufficient to allow it to bind, without overlap, in some position to another copy of itself, then an infinite sequence of copies can bind in a line at the same relative positions to their neighbors.)  In the case of a $3$-position limited system, the construction described for $4$-position limited uniquely paired systems (see the proof of Theorem~\ref{thm:4positionsUniquePair}) works by the same technique, simply creating tile types in the simulating system whose input glues are concatenations of pairs of output glues which are found on the same tile type in the simulated system, and output glues are concatenations of both output glues.
\end{proof}

\begin{theorem}\label{thm:4positionsUniquePair}
Any $4$-position limited, uniquely paired polyomino system $\Gamma = (T,\sigma,1)$ is scale-free simulated by a monomino tile system at temperature 1 (a.k.a. a temperature 1 aTAM system).
\end{theorem}

\begin{figure}[htp]
\begin{center}
\includegraphics[width=5.0in]{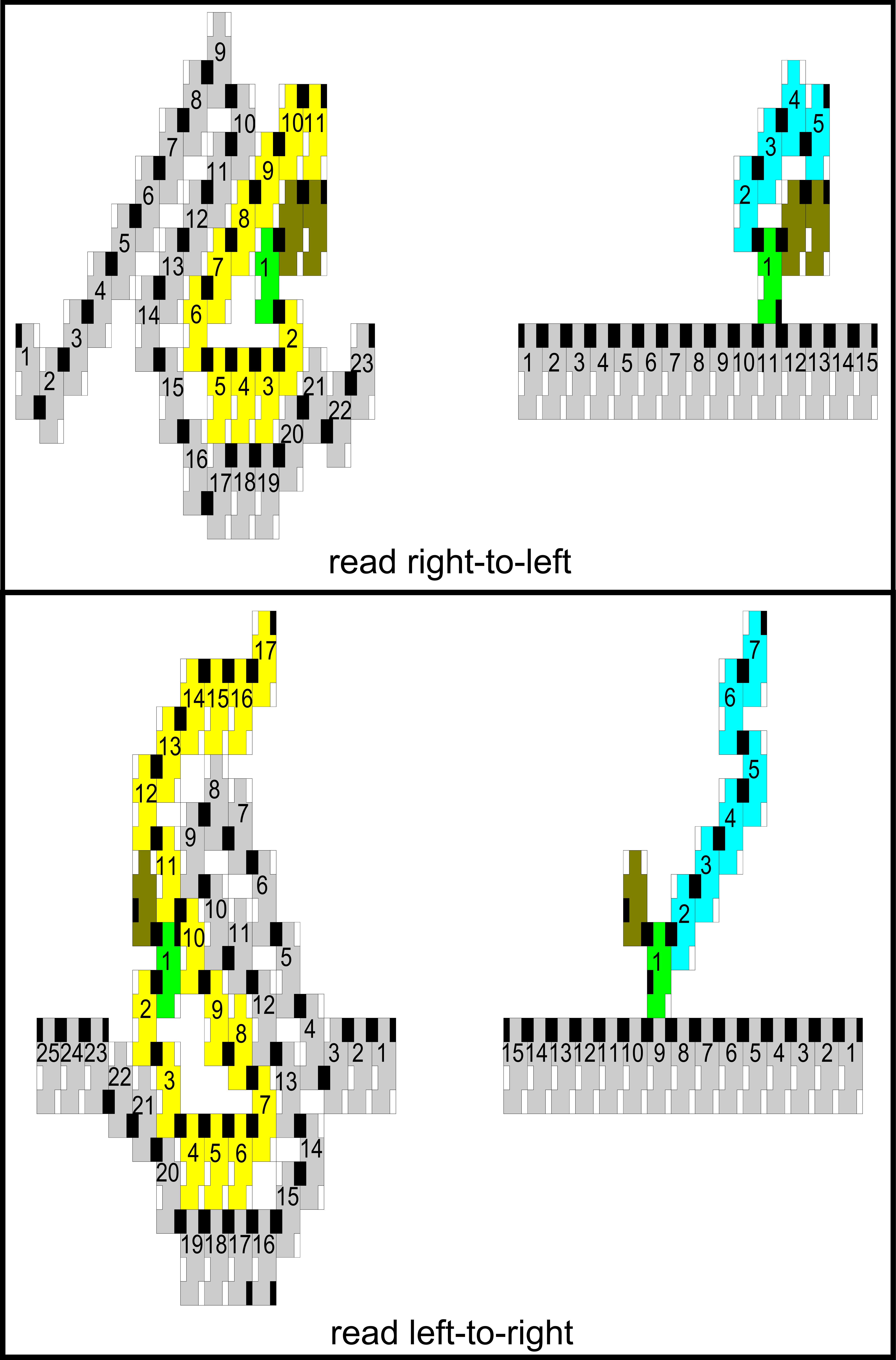}
\caption{The following bit reading gadget demonstrates how a $1\times 4$ polyomino that uses only 4 distinct glue positions can perform universal computation.  Due to the asymmetry of the glue locations on the east and west edges of the polyomino, the right-to-left and left-to-right bit reading gadgets are also asymmetric.  Geometric bit reading is not possible with fewer than 4 glue positions, and is not possible with only 4 positions and unique pairing of positions.}
\label{fig:cuddlerGadgets}
\end{center}
\end{figure}

\begin{proof}
Consider some $4$-position limited, uniquely paired system $\Gamma=(T,\sigma,1)$, and denote the shape of the tiles in $\Gamma$ as polyomino $P$.  Let $\vec{v}$ denote the translation difference between two bonded polyominoes from $T$ that are bonded with the first pair of glue positions, and let $\vec{u}$ denote the translation for the second pair of bonding positions.  As an example, consider the the polyomino of Figure~\ref{fig:4gluesPicture}(a).  The north-south glue positions are separated by vector $\vec{u}=(3,5)$, and the east-west glue positions are separated by vector $\vec{v}=(6,1)$. 
To show that $\Gamma$ is simulated by some monomino system, we consider two cases.  For polyomino $P$, let $P_{\vec{x}}$ denote the polyomino obtained by translating $P$ by some vector $\vec{x}$.
 For case 1, we assume $P$, $P_{\vec{v}}$, $P_{\vec{u}}$, and $P_{\vec{v}+\vec{u}}$ are mutually non-overlapping. For case 2, we assume that either $P_{\vec{v}}$ and $P_{\vec{u}}$ overlap, or that $P_{\vec{v}}$ and $P_{-\vec{u}}$ overlap (note that $P$ overlaps $P_{\vec{u}+\vec{v}}$ if and only if $P_{\vec{v}}$ overlaps $P_{-\vec{u}}$, and thus is covered by case 2.)  The two cases are depicted in Figure~\ref{fig:4gluesPicture}.

Case 1:  In this scenario, the tiles of $\Gamma$ grow in a 2D lattice with basis vectors $\vec{u}$ and $\vec{v}$ and can be simulated by a monomino system that simply creates a square monomino tile for each element of $t \in T$, placing the glue types of the first pair of uniquely paired glues of $t$ on the north and south edges of the representing monomino, and the other pair on the east and west edges.  The bijective mapping that satisfies the scale-free simulation requirement simply replaces each $t$ in a producible assembly of $\Gamma$ with the corresponding monomino for that $t$, thereby yielding an appropriate assembly over the unit square tiles.

Case 2:  Without loss of generality, assume it is the case that $P_{\vec{v}}$ and $P_{\vec{u}}$ overlap. 
We first observe that growth in this scenario is linear (see Figure~\ref{fig:4gluesPicture}).  We will simulate $\Gamma$ with a linear east/west monomino tile system (i.e. a system where tiles have glues only on their west and east sides).  However, unlike in case 1, a simple tile for tile replacement is not sufficient as this does not allow for the simulation of the potential non-deterministic placement of one of two mutually overlapping tiles.  To deal with this, we increase the tile set size by a quadratic factor.  Let $\vec{u}$ and $\vec{v}$ be defined as before for case 1.  As a first subcase, assume that the system we are interested in simulating is such that no polyomino is attachable to the seed at translation $\vec{-u}$ or $\vec{-v}$ from the seed, i.e., tiles only attach at positive linear combinations of $\vec{u}$ and $\vec{v}$.  With this restriction, the 4 glue positions of the tiles can be thought of as 2 \emph{input} positions and 2 \emph{output} positions, where a tile always attaches based on the binding of a glue on an input side.  Let $g_{in1}$ and $g_{in2}$ represent the two input positions, and $g_{out1}$ and $g_{out2}$ the two output positions, where $g_{in1}$ is uniquely paired with $g_{out1}$, and $g_{in2}$ with $g_{out2}$.  Then, for a tile $t \in T$ and $d \in \{in1,in2,out1,out2\}$, let $g_d(t)$ be the glue label at location $d$ on tile $t$.  To simulate $\Gamma$, we generate a set $T'$ of at most $2|G|$ east/west monomino tiles (where $G$ is the set of all glue labels of the polyomino tile set $T$), and we will specify the east glue position as $g_e$ and the west as $g_w$ for tiles in $T'$ (and we'll treat $g_w$ as the input and $g_e$ as the output sides). 
For $0 \leq i < |T|$, we generate a set $T'_i$ of tiles from tile $t_i \in T$ as follows.  Let $a = g_{in1}(t_i)$ and $b = g_{in2}(t_i)$.  For every tile $t' \in T$ such that $g_{out1}(t') = a$, we generate tile $t'' \in T'_i$ such that $g_w(t'') = a \cdot g_{out2}(t')$ (where $a \cdot g_{out2}(t')$ is just the concatenation of the labels of glues $a$ and $g_{out2}(t')$),  and $g_e(t'') = g_{out1}(t_i) \cdot g_{out2}(t_i)$.  Similarly, for every every tile $t' \in T$ such that $g_{out2}(t') = b$, we generate tile $t'' \in T'_i$ such that $g_w(t'') = g_{out1}(t') \cdot b$,  and $g_e(t'') = g_{out1}(t_i) \cdot g_{out2}(t_i)$.  Essentially, whenever a tile from $T'$ is placed, it presents on its output side both of the output glues of the tile from $T$ that it was designed to simulate.  (This is because in $\Gamma$, either of those output glues could be used as an input glue to bind the next tile, but only one of them.)  Therefore, any tile wishing to use one of those output glues as an input glue must now have a glue label which matches the concatenation of that glue and any other which may have been paired with it as an output.  In such a way, we generate each set $T'_i$ to simulate $t_i \in T$, and $T' = \bigcup_{i=0}^{<|T|} T'_i$.  Therefore, for our scale-free simulator $\Gamma' = (T', \sigma', 1)$ the function $f$ defined for the scale-free simulation simply maps each tile in $T'_i \subseteq T'$ to its corresponding $t_i \in T$ and $\sigma'$ is simply formed by corresponding tiles between $T$ and $T'$.

The bijection from producible assemblies $\mathcal{A}[\Gamma]$ and those in $\mathcal{A}[\Gamma']$ is defined as follows.  Consider some assembly $A \in \mathcal{A}[\Gamma]$.  Starting from the seed, for each tile attaching in sequence, attach a corresponding tile in $T'$ to the seed of $\Gamma'$.  In particular, if tile $t$ with input glues $a$ and $b$ attaches via glue $a$ in $\Gamma$, then attach a tile from $t' \in T'$ with glue label ``$a \cdot i$" where $i$ matches the unused output glue label of the tile that $t$ attached to, and the output glue of $t'$ is the concatenation of the output glues of $t$.  Similarly define the simulator tile attachment in the case that glue $b$ is the bonding glue.  By repeating this process, an element of $\mathcal{A}[\Gamma']$ is generated for any element of $\mathcal{A}[\Gamma]$, and this mapping is bijective and thus provides a scale-free simulation of the input system $\Gamma$ with at most quadratic increase in tile complexity.

Finally, it is easy to see that the restriction that the seed only grows in direction $\vec{u}$ or $\vec{v}$ (and not $\vec{-u}$ or $\vec{-v}$) can easily be removed by taking a more general system and first doubling it's tile set so that each tile type only every attaches in a $\vec{u}$/$\vec{v}$ or $\vec{-u}$/$\vec{-v}$ direction.  The simulation system is then constructed with two symmetric applications of the previous construction.
\end{proof}

\begin{theorem}\label{thm:4positionsUniversal}
There exist $4$-position limited polyTAM systems that are computationally universal at temperature 1.
\end{theorem}
\begin{proof}
We prove this by providing constructions for right-to-left and left-to-right bit reading gadgets for the $1\times 4$ polyomino that uses only 4 unique glues positions.  The details of the bit reading gadgets are presented in Figure~\ref{fig:cuddlerGadgets}.  It is straightforward to connect the bit reading gadgets in a fashion similar to previous results to construct a zig-zag Turing machine simulation.
\end{proof}

\section{Multiple Polyomino Systems}\label{sec:multi-poly}

In previous sections we have focussed on the computational power of systems consisting of singly shaped polyominoes and showed that single polyomino systems are universal for polyominoes of size $\ge 3$, while monomino and domino systems are likely not capable of such computation.  We now show that any multiple shape polyomino system (i.e. one that utilizes at least 2 distinct polyomino shapes, regardless of their size) is capable of universal computation.

\begin{lemma}\label{lemma:doubleDomino}
For every standard Turing Machine $M$ and input $w$, there exists a TAS with $\tau=1$ consisting only of tiles shaped as dominoes, $2\times1$ and $1\times2$ polyominoes, that simulates $M$ on $w$.
\end{lemma}


\begin{proof}
To see this we provide a sketch of a right-to-left bit reading gadget in Figure~\ref{fig:doubleDomino-gadgets}.  A left-to-right gadget can be derived similarly, and together these gadgets can be used to construct a zig-zag Turing machine simulation in the same manner as with previous constructions in this paper.
\begin{figure}[htp]
\begin{center}
\includegraphics[width=2.25in]{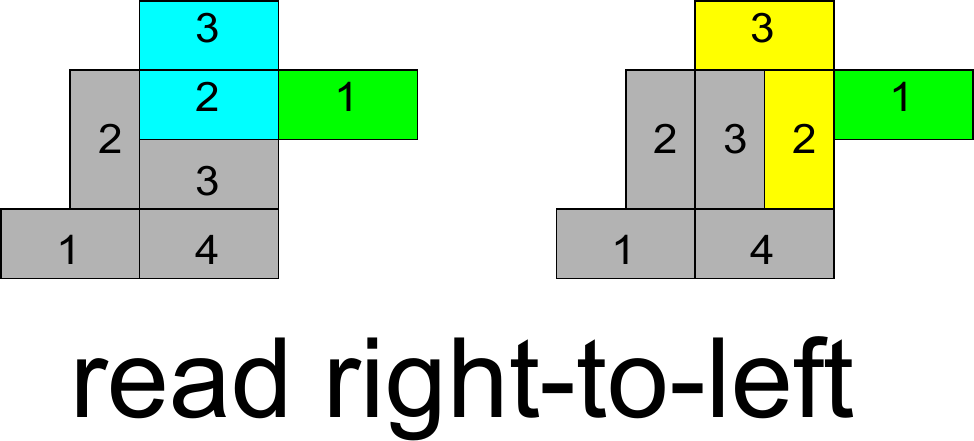}
\caption{A 2-shape system consisting of the two distinct domino polyominoes can be designed to read bits for the simulation of a zig-zag Turing machine.}
\label{fig:doubleDomino-gadgets}
\end{center}
\end{figure}
\end{proof}

\begin{theorem}
For every standard Turing Machine $M$ and input $w$, and any 2 distinct polyominoes $P$ and $Q$, there exists a TAS with $\tau=1$ consisting only of tiles shaped as $P$ or $Q$ that simulates $M$ on $w$.
\end{theorem}

\begin{proof}
If either $P$ or $Q$ are size 3 or larger, we get the result from theorem~\ref{thm:comp-univ-poly}.  If one polyomino is a monomino and the other a domino, we get the result from the paper~\cite{Duples}.  Finally, if $P$ and $Q$ are the two distinct domino shapes, then we get the result by Lemma~\ref{lemma:doubleDomino}.
\end{proof}

\pagebreak

\section*{Acknowledgements}
This work was initiated at the 29th Bellairs Winter Workshop on Computational Geometry on March 21-28, 2014 in Holetown, Barbados.
We thank the other participants of that workshop for a fruitful and collaborative environment.
We especially thank Erik Demaine, Scott Summers, Damien Woods, Andrew Winslow, and Dave Doty for helpful discussions and many of the initial ideas.

\bibliographystyle{amsplain}
\bibliography{tam}

\providecommand{\bysame}{\leavevmode\hbox to3em{\hrulefill}\thinspace}
\providecommand{\MR}{\relax\ifhmode\unskip\space\fi MR }
\providecommand{\MRhref}[2]{%
  \href{http://www.ams.org/mathscinet-getitem?mr=#1}{#2}
}
\providecommand{\href}[2]{#2}
\begin{thebibliography}{10}

\bibitem{BarSchRotWin09}
Robert~D. Barish, Rebecca Schulman, Paul~W. Rothemund, and Erik Winfree,
  \emph{An information-bearing seed for nucleating algorithmic self-assembly},
  Proceedings of the National Academy of Sciences \textbf{106} (2009), no.~15,
  6054--6059.

\bibitem{CookFuSch11}
Matthew Cook, Yunhui Fu, and Robert~T. Schweller, \emph{Temperature 1
  self-assembly: Deterministic assembly in 3{D} and probabilistic assembly in
  2{D}}, SODA 2011: Proceedings of the 22nd Annual ACM-SIAM Symposium on
  Discrete Algorithms, SIAM, 2011.

\bibitem{OneTile}
E.~D. Demaine, M.~L. Demaine, S.~P. Fekete, M.~J. Patitz, R.~T. Schweller,
  A.~Winslow, and D.~Woods, \emph{One tile to rule them all: Simulating any
  tile assembly system with a single universal tile}, Proceedings of the 41st
  International Colloquium on Automata, Languages, and Programming (ICALP
  2014), {\rm IT University of Copenhagen, Denmark, July 8-11, 2014}
  (J.~Esparza, P.~Fraigniaud, T.~Husfeldt, and E.~Koutsoupias, eds.), LNCS,
  vol. 8572, Springer Berlin Heidelberg, 2014, pp.~368--379.

\bibitem{2HAMIU}
Erik~D. Demaine, Matthew~J. Patitz, Trent~A. Rogers, Robert~T. Schweller,
  Scott~M. Summers, and Damien Woods, \emph{The two-handed assembly model is
  not intrinsically universal}, 40th International Colloquium on Automata,
  Languages and Programming, ICALP 2013, Riga, Latvia, July 8-12, 2013, Lecture
  Notes in Computer Science, Springer, 2013.

\bibitem{IUSA}
David Doty, Jack~H. Lutz, Matthew~J. Patitz, Robert~T. Schweller, Scott~M.
  Summers, and Damien Woods, \emph{The tile assembly model is intrinsically
  universal}, Proceedings of the 53rd Annual IEEE Symposium on Foundations of
  Computer Science, FOCS 2012, 2012, pp.~302--310.

\bibitem{SFTSAFT}
David Doty, Matthew~J. Patitz, Dustin Reishus, Robert~T. Schweller, and
  Scott~M. Summers, \emph{Strong fault-tolerance for self-assembly with fuzzy
  temperature}, Proceedings of the 51st Annual IEEE Symposium on Foundations of
  Computer Science (FOCS 2010), 2010, pp.~417--426.

\bibitem{jLSAT1}
David Doty, Matthew~J. Patitz, and Scott~M. Summers, \emph{Limitations of
  self-assembly at temperature 1}, Theoretical Computer Science \textbf{412}
  (2011), 145--158.

\bibitem{evans2014crystals}
Constantine~Glen Evans, \emph{Crystals that count! physical principles and
  experimental investigations of dna tile self-assembly}, Ph.D. thesis,
  California Institute of Technology, 2014.

\bibitem{GeoTiles}
Bin Fu, Matthew~J. Patitz, Robert~T. Schweller, and Robert Sheline,
  \emph{Self-assembly with geometric tiles}, Proceedings of the 39th
  International Colloquium on Automata, Languages and Programming, ICALP, 2012,
  pp.~714--725.

\bibitem{BreakableDuples}
Jacob Hendricks, Matthew~J. Patitz, and Trent~A. Rogers, \emph{Doubles and
  negatives are positive (in self-assembly)}, Proceeding of Unconventional
  Computation and Natural Computation 2014 (UCNC 2014), {\rm University of
  Western Ontario, London, Ontario, Canada, 7/14/2014 - 7/18/2014}, 2014, to
  appear.

\bibitem{Duples}
Jacob Hendricks, Matthew~J. Patitz, Trent~A. Rogers, and Scott~M. Summers,
  \emph{The power of duples (in self-assembly): It's not so hip to be square},
  Proceedings of 20th International Computing and Combinatorics Conference
  (COCOON 2014), {\rm Atlanta, Georgia, USA, 8/04/2014 - 8/06/2014}, 2014, to
  appear.

\bibitem{LaWiRe99}
T.H. LaBean, E.~Winfree, and J.H. Reif, \emph{Experimental progress in
  computation by self-assembly of {DNA} tilings}, {DNA} Based Computers
  \textbf{5} (1999), 123--140.

\bibitem{jCCSA}
James~I. Lathrop, Jack~H. Lutz, Matthew~J. Patitz, and Scott~M. Summers,
  \emph{Computability and complexity in self-assembly}, Theory Comput. Syst.
  \textbf{48} (2011), no.~3, 617--647.

\bibitem{MaoLabReiSee00}
Chengde Mao, Thomas~H. LaBean, John~H. Relf, and Nadrian~C. Seeman,
  \emph{Logical computation using algorithmic self-assembly of {D}{N}{A}
  triple-crossover molecules.}, Nature \textbf{407} (2000), no.~6803, 493--6.

\bibitem{ManuchTemp1}
J\'{a}n Ma\v{n}uch, Ladislav Stacho, and Christine Stoll, \emph{Two lower
  bounds for self-assemblies at temperature 1}, Journal of Computational
  Biology \textbf{17} (2010), no.~6, 841--852.

\bibitem{IUNeedsCoop}
Pierre-Etienne Meunier, Matthew~J. Patitz, Scott~M. Summers, Guillaume
  Theyssier, Andrew Winslow, and Damien Woods, \emph{Intrinsic universality in
  tile self-assembly requires cooperation}, Proceedings of the ACM-SIAM
  Symposium on Discrete Algorithms (SODA 2014), (Portland, OR, USA, January
  5-7, 2014), 2014, pp.~752--771.

\bibitem{Signals}
Jennifer~E. Padilla, Matthew~J. Patitz, Raul Pena, Robert~T. Schweller,
  Nadrian~C. Seeman, Robert Sheline, Scott~M. Summers, and Xingsi Zhong,
  \emph{Asynchronous signal passing for tile self-assembly: Fuel efficient
  computation and efficient assembly of shapes}, UCNC, 2013, pp.~174--185.

\bibitem{SingleNegative}
Matthew~J. Patitz, Robert~T. Schweller, and Scott~M. Summers, \emph{Exact
  shapes and turing universality at temperature 1 with a single negative glue},
  Proceedings of the 17th international conference on DNA computing and
  molecular programming (Berlin, Heidelberg), DNA'11, Springer-Verlag, 2011,
  pp.~175--189.

\bibitem{jSADSSF}
Matthew~J. Patitz and Scott~M. Summers, \emph{Self-assembly of discrete
  self-similar fractals}, Natural Computing \textbf{1} (2010), 135--172.

\bibitem{jSADS}
Matthew~J. Patitz and Scott~M. Summers, \emph{Self-assembly of decidable sets},
  Natural Computing \textbf{10} (2011), no.~2, 853--877.

\bibitem{RothTriangles}
Paul W.~K Rothemund, Nick Papadakis, and Erik Winfree, \emph{Algorithmic
  self-assembly of dna sierpinski triangles}, PLoS Biol \textbf{2} (2004),
  no.~12, e424.

\bibitem{RotWin00}
Paul W.~K. Rothemund and Erik Winfree, \emph{The program-size complexity of
  self-assembled squares (extended abstract)}, STOC '00: Proceedings of the
  thirty-second annual ACM Symposium on Theory of Computing (Portland, Oregon,
  United States), ACM, 2000, pp.~459--468.

\bibitem{SchWin07}
Rebecca Schulman and Erik Winfree, \emph{Synthesis of crystals with a
  programmable kinetic barrier to nucleation}, Proceedings of the National
  Academy of Sciences \textbf{104} (2007), no.~39, 15236--15241.

\bibitem{SolWin07}
David Soloveichik and Erik Winfree, \emph{Complexity of self-assembled shapes},
  SIAM Journal on Computing \textbf{36} (2007), no.~6, 1544--1569.

\bibitem{Winf98}
Erik Winfree, \emph{Algorithmic self-assembly of {D}{N}{A}}, Ph.D. thesis,
  California Institute of Technology, June 1998.

\end{thebibliography}

\end{document}